\documentclass[journal]{IEEEtran}

\usepackage{amsmath}
\usepackage{amssymb}
\usepackage{amsmath}
\usepackage{amssymb}
\usepackage{epsfig}
\usepackage{epsf}
\usepackage{subfigure}
\usepackage{graphicx}
\usepackage{url}
\usepackage{subeqn}

\def\BibTeX{{\rm B\kern-.05em{\sc i\kern-.025em b}\kern-.08em
    T\kern-.1667em\lower.7ex\hbox{E}\kern-.125emX}}

\newtheorem{claim}{Claim}
\newtheorem{corollary}{Corollary}
\newtheorem{theorem}{Theorem}
\newtheorem{definition}{Definition}
\newtheorem{lemma}{Lemma}

\newtheorem{remark}{Remark}

\begin{document}

\title{Interference Channels with Rate-Limited Feedback}

\author{Alireza~Vahid,
        Changho~Suh,
        and~A.~Salman~Avestimehr
        \thanks{A. Vahid and A. S. Avestimehr are with the School of Electrical and Computer Engineering, Cornell University, Ithaca, USA. Email: {\sffamily av292@cornell.edu} and {\sffamily avestimehr@ece.cornell.edu}.}
\thanks{C. Suh is with the Research Laboratory of Electronics at Massachusetts Institute of Technology, Cambridge, USA. Email: {\sffamily chsuh@mit.edu}.}
\thanks{The research of A. S. Avestimehr and A. Vahid was supported in part by
the NSF CAREER award 0953117 and U.S. Air Force Young Investigator Program award FA9550-11-1-0064.}
\thanks{Copyright (c) 2011 IEEE.}}

\maketitle

\begin{abstract}
We consider the two-user interference channel with \emph{rate-limited} feedback. Related prior works focus on the case where feedback links have \emph{infinite} capacity, while no research has been done for the rate-limited feedback problem. Several new challenges arise due to the capacity limitations of the feedback links, both in deriving inner-bounds and outer-bounds. We study this problem under three different interference models: the El Gamal-Costa deterministic model, the linear deterministic model, and the Gaussian model. For the first two models, we develop an achievable scheme that employs three techniques: Han-Kobayashi message splitting, quantize-and-binning, and decode-and-forward. We also derive new outer-bounds for all three models and we show the optimality of our scheme under the linear deterministic model. In the Gaussian case, we propose a transmission strategy that incorporates lattice codes, inspired by the ideas developed in the first two models. For symmetric channel gains, we prove that the gap between the achievable sum-rate of the proposed scheme and our new outer-bounds is bounded by a constant number of bits, independent of the channel gains.
\end{abstract}


\section{Introduction}
\label{introduction}
The history of feedback in communication systems traces back to Shannon. It is well-known that feedback does not increase the capacity of discrete memoryless point-to-point channels~\cite{Sha}. However, feedback can enlarge the capacity region of multi-user networks, even in the most basic case of the two-user memoryless multiple-access channel~\cite{Gaar,Oza}. Hence, there has been a growing interest in developing feedback strategies and understanding the fundamental limits of communication over multi-user networks with feedback, in particular the two-user interference channel (IC). See~\cite{Kra, Kram, Gas, Tandon:08, Jia, Suh, Sahai} for example.

Especially in~\cite{Suh}, the infinite-rate feedback capacity of the two-user Gaussian IC has been characterized to within a 2-bit gap. One consequence of this result is that interestingly feedback can provide an unbounded capacity increase. This is in contrast to point-to-point and multiple-access channels where feedback provides no gain and bounded gain respectively.

While the feedback links are assumed to have {\em infinite} capacity in~\cite{Suh}, a more realistic feedback model is one where feedback links are {\em rate-limited}. In this paper, we study the impact of the rate-limited feedback in the context of the two-user IC. We focus on two fundamental questions: (1) what is the maximum capacity gain that can be obtained with access to feedback links at a specific rate of $C_{\sf FB}$? (2) what are the transmission strategies that exploit the available feedback links efficiently? Specifically, we address these questions under three channel models: the El Gamal-Costa deterministic model~\cite{ElGamal:it82}, the linear deterministic model of~\cite{ADT10}, and the Gaussian model.

Under the El Gamal-Costa deterministic model, we derive inner-bounds and outer-bounds on the capacity region. As a result, we show that the capacity region can be enlarged using feedback by at most the amount of available feedback, \emph{i.e.}, ``one bit of feedback is at most worth one bit of capacity''. Our achievable scheme employs three techniques: $(1)$  Han-Kobayashi message splitting; $(2)$ quantize-and-binning; and $(3)$ decode-and-forward.
Unlike the infinite-rate feedback case~\cite{Suh}, in the rate-limited feedback case, a receiver cannot provide its \emph{exact} received signal to its corresponding transmitter; therefore, the main challenge is how to smartly decide what to send back through the available rate-limited feedback links. We overcome this challenge as follows. We first split each transmitter's message into three parts: the cooperative common, the non-cooperative common, and the private message. Next, each receiver quantizes its received signal and then generates a binning index so as to capture part of the other user's common message (that we call the cooperative common message) which causes interference to its own message. The receiver will then send back this binning index to its intended transmitter through the rate-limited feedback links.
With this feedback, each transmitter decodes the other user's cooperative common message by exploiting its own message as side information. This way transmitters will be able to cooperate by means of the feedback links, thereby enhancing the achievable rates. This result will be described in Section~\ref{elgamal}.



We then study the problem under the linear deterministic model~\cite{ADT10} which captures the key properties of the wireless channel, and thus provides insights that can lead to an approximate capacity of Gaussian networks~\cite{ADT10, Suh, Guy, Guy2, Ave3}. We show that our inner-bounds and outer-bounds match under this linear deterministic model, thus establishing the capacity region. While this model is a special case of the El Gamal-Costa model, it has a significant role to play in motivating a generic achievable scheme for the El Gamal-Costa model. Moreover, the explicit achievable scheme in this model provides a concrete guideline to the Gaussian case.
We will explain this result in Section~\ref{linear}.

Inspired by the results in the deterministic models, we develop an achievable scheme and also derive new outer-bounds for the Gaussian channel. In order to translate the main ideas in our achievability strategy for the deterministic models into the Gaussian case, we employ lattice coding which enables receivers to decode superposition of codewords. Specifically at each transmitter, we employ lattice codes for cooperative messages. By appropriate power assignment of the codewords, we make the desired lattice codes arrive at the same power level, hence, receivers being able to decode the superposition of codewords. Each receiver will then decode the index of the lattice code corresponding to the superposition and sends it back to its corresponding transmitter where the cooperative common message of the other user will be decoded. For symmetric channel gains, we show that the gap between the achievable sum-rate and the outer-bounds can be bounded by a constant, independent of the channel gains. This will be explained in Section~\ref{gaussian}.

{\bf Related Work:} Interference channels with infinite-rate feedback have received previous attention~\cite{Kra, Kram, Gas, Tandon:08, Jia, Suh, Sahai}. Kramer~\cite{Kra,Kram} developed a feedback strategy in the Gaussian IC; In~\cite{Gas}, Gastpar and Kramer established an outer-bound on the usefulness of noisy feedback for the two-user IC. Tandon and Ulukus in~\cite{Tandon:08} derived an outer bound using the dependence balance bound technique~\cite{Hekstra:89}. However, the gap between the inner-bounds and the outer-bounds becomes arbitrarily large with the increase of signal to noise ratio ($\mathsf{SNR}$) and interference to noise ratio ($\mathsf{INR}$). Jiang-Xin-Garg~\cite{Jia} derived an achievable region in the discrete memoryless IC with feedback, based on block Markov encoding~\cite{Cover:it79} and binning. However, no outer-bounds are provided. Suh and Tse in~\cite{Suh} developed new inner bounds and outer bounds to characterize the feedback capacity of the Gaussian IC to within 2 bits. Sahai \emph{et al.} in~\cite{Sahai}, have shown that in order to achieve the infinite-rate feedback capacity, it is sufficient to have only one feedback link of infinite rate from one receiver to either of the two transmitters.

While no research has been done for the rate-limited feedback problem, some works have been done for the different yet related problem - the conferencing encoder problem~\cite{ Tuninetti:isit07, yang2011interference, CaoChen:07,Vinod:arix09, IHWang, Bagheri}. Tuninetti in~\cite{Tuninetti:isit07} has proposed a coding strategy for the two-user IC that results in higher achievable rates by expoliting overheard information by the transmitters; backward decoding is incorporated at the receivers which we also employ in our achievability scheme. This result was improved in~\cite{yang2011interference} by incorporating Gelfand-Pinsker coding~\cite{costa1983writing,gel1980coding} to send cooperatively the private messages. Prabhakaran and Viswanath~\cite{Vinod:arix09} have made a connection between the feedback problem and the conferencing encoder problem
However, the connection is loose especially when the feedback link is rate-limited, although it can be strong for the infinite-rate feedback case. It turns out this distinction between the two problems leads to developing a new lattice-code-based achievable scheme in our problem.

The rest of the paper is organized as follows. In Section~\ref{problem}, we formulate our problem and give a brief overview of the channel models. In Section~\ref{motivation}, we provide a motivating example which forms inspiration of our achievable scheme. In Section~\ref{elgamal}, we will provide our main results under the El Gamal-Costa deterministic model. We will then present the capacity theorem for the linear deterministic model in Section~\ref{linear}. Next, in Section~\ref{gaussian}, we describe our main results for the Gaussian channel. Section~\ref{conclusion} concludes the paper. 

\section{Problem Formulation and Network Model}
\label{problem}
We consider a two-user interference channel (IC) where a noiseless rate-limited feedback link is available from each receiver to its corresponding transmitter. See Figure~\ref{interference}. The feedback link from receiver $k$ to transmitter $k$ is assumed to have a capacity of $C_{\sf FBk}$, $k=1,2$.
on.

\begin{figure}[ht]
\centering
\includegraphics[height=4.4cm]{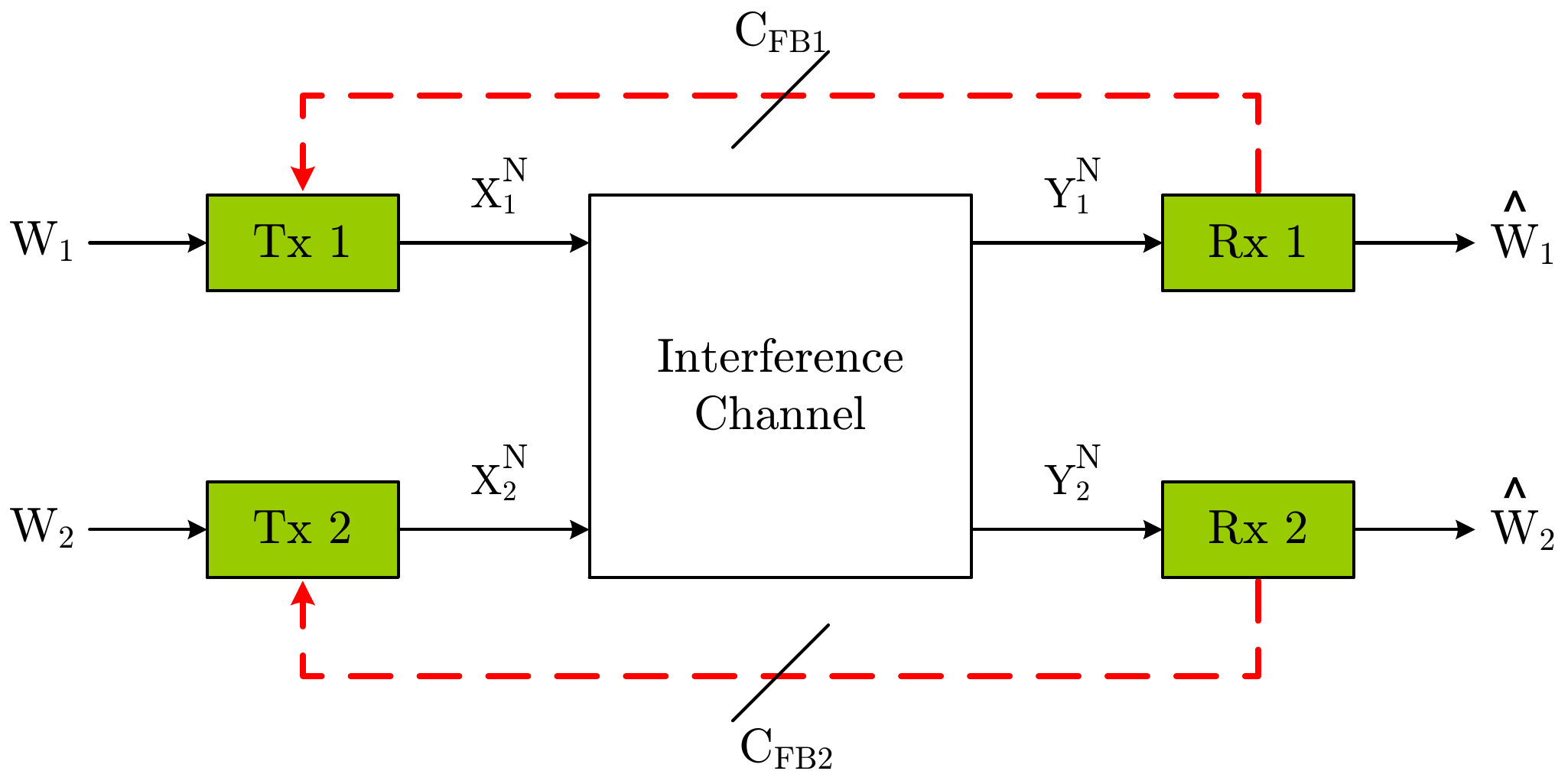}
\caption{Two-user interference channel with rate-limited feedback.}
\label{interference}
\end{figure}

Transmitters $1$ and $2$ wish to reliably communicate independent and uniformly distributed messages $W_1 \in \{1,2,\ldots,2^{N R_1} \}$ and $W_2 \in \{1,2,\ldots,2^{N R_2} \}$ to receivers $1$ and $2$ respectively, during $N$ uses of the channel. The transmitted signal of transmitter $k$, $k=1,2$, at time $i$, $1 \leq i \leq N$, and the received signal of receiver $k$, $k=1,2$, at time $i$, $1 \leq i \leq N$, are respectively denoted by $X_{k,i}$ and $Y_{k,i}$. There are two feedback encoders at the receivers that create the feedback signals from the received signals:
\begin{equation} \label{}
\tilde{Y}_{k,i} = \tilde{e}_{k,i}(Y_{k,1},\ldots,Y_{k,i-1}) = \tilde{e}_{k,i}(Y_k^{(i-1)}), \hspace{2mm} k = 1, 2.
\end{equation}
where we use shorthand notation to indicate the sequence up to $i-1$.

Due to the presence of feedback, the encoded signal $X_{k,i}$ of user $k$ at time $i$ is a function of both its own message and previous outputs of the corresponding feedback encoder:
\begin{equation}
\label{}
X_{k,i} =  e_{k,i}(W_k,\tilde{Y}_k^{(i-1)}), \hspace{5mm} k = 1, 2.
\end{equation}

Each receiver $k$, $k=1,2$, uses a decoding function $d_{k,N}$ to get the estimate $\hat{W}_k$, from the channel outputs $\{ Y_{k,i} : i = 1, \ldots, N \}$. An error occurs whenever $\hat{W}_k \neq W_k$. The average probability of error is given by
\begin{equation}\label{}
\lambda_{k,N} = \mathbb{E}[P(\hat{W}_k \neq W_k)], \hspace{5mm} k = 1, 2,
\end{equation}
where the expectation is taken with respect to the random choice of the transmitted messages $W_1$ and $W_2$.

We say that a rate pair $(R_1,R_2)$ is achievable, if there exists a block encoder at each transmitter, a block encoder at each receiver that creates the feedback signals, and a block decoder at each receiver as described above, such that the average error probability of decoding the desired message at each receiver goes to zero as the block
length $N$ goes to infinity. The capacity region $\mathcal{C}$ is the closure of the set of the achievable rate pairs.

We will consider the following three channel models to investigate this problem.

\noindent {\bf 1- El Gamal-Costa Deterministic Interference Channel}:

\begin{figure}[!htp]
\centering
\includegraphics[height=5cm]{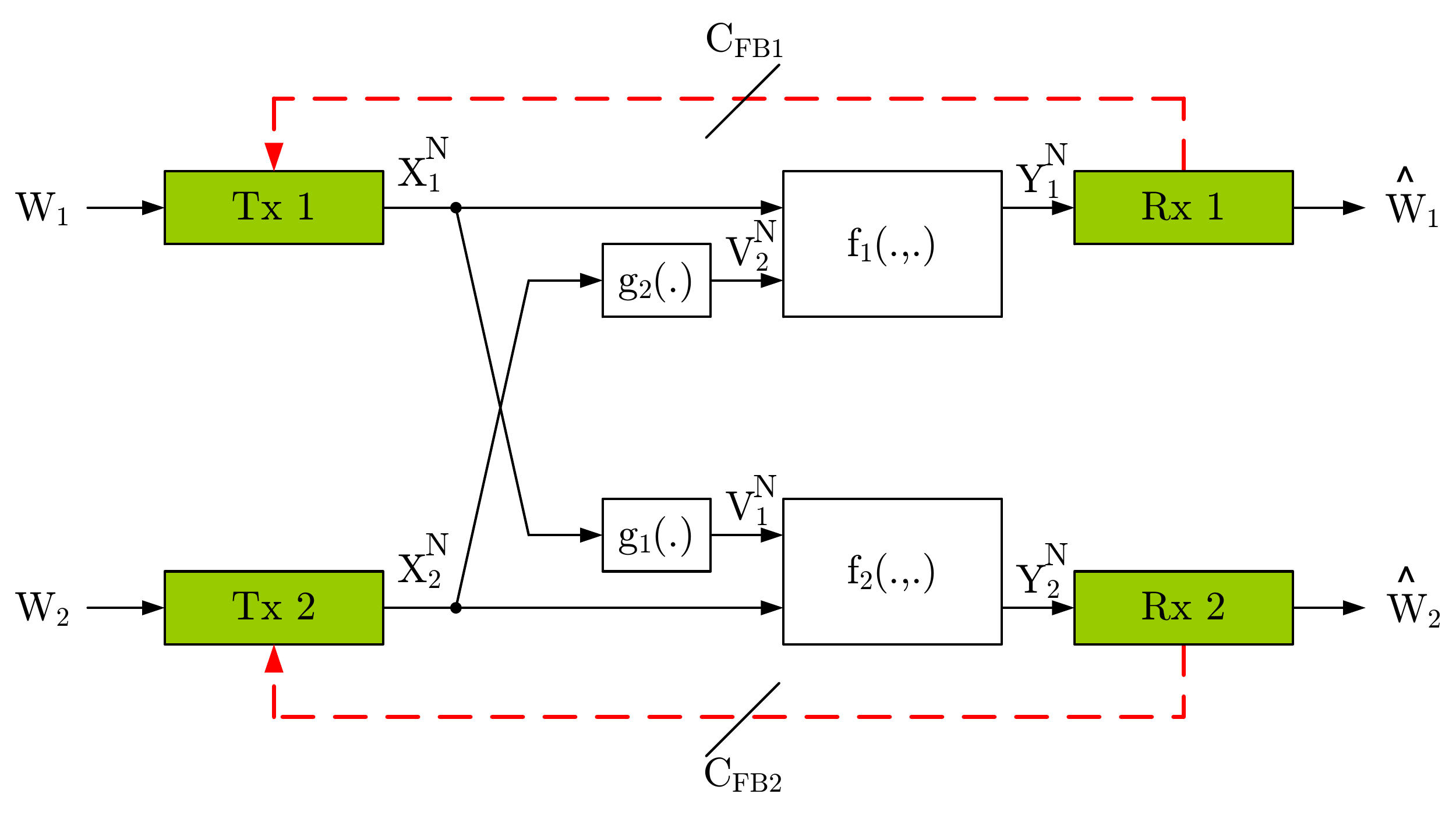}
\caption{The El Gamal-Costa deterministic IC with rate-limited feedback.} \label{fig_ElGamalCosta}
\end{figure}


Figure~\ref{fig_ElGamalCosta} illustrates the El Gamal-Costa deterministic IC~\cite{ElGamal:it82} with rate-limited feedback. In this model the outputs $Y_1$ and $Y_2$ and the interferences $V_1$ and $V_2$ are (deterministic) functions of inputs $X_1$ and $X_2$ \cite{ElGamal:it82}:
\begin{align}
\begin{split}
\label{eq-ElGamalCostaCondition0}
Y_{1,i} &= f_1(X_{1,i},V_{2,i}), \\
Y_{2,i} &= f_2(X_{2,i},V_{1,i}), \\
V_{1,i} &= g_1(X_{1,i}), \\
V_{2,i} &= g_2(X_{2,i}),
\end{split}
\end{align}
where $f_1(.,.)$ and $f_2(.,.)$ are such that
\begin{align}
\begin{split}
\label{eq-ElGamalCostaCondition}
H(V_{2,i}|Y_{1,i},X_{1,i}) &= 0, \\
H(V_{1,i}|Y_{2,i},X_{2,i}) &= 0.
\end{split}
\end{align}

Here $V_k$ is a part of $X_k$ ($k=1,2$), visible to the unintended receiver. This implies that in any system where each decoder can decode its message with arbitrary small error probability, $V_1$ and $V_2$ are completely determined at receivers $2$ and $1$, respectively, \emph{i.e.}, these are common signals.

\noindent {\bf 2- Linear Deterministic Interference Channel}:

This model, which was introduced in~\cite{ADT10}, captures the effect of broadcast and superposition in wireless networks. We study this model to bridge from general deterministic networks into Gaussian networks. In this model, there is a non-negative integer $n_{kj}$ representing channel gain from transmitter $k$ to receiver $j$, $k=1,2$, and $j=1,2$. In the linear deterministic IC, we can write the channel input to the transmitter $k$ at time $i$ as $X_{k,i} =[X_{k,i}^1~ X_{k,i}^2 \ldots X_{k,i}^q]^T \in \mathbb{F}_2^q$, $k=1,2$, such that $X^1_{k,i}$ and $X_{k,i}^q$ represent the most and the least significant bits of the transmitted signal respectively. Also, $q$ is the maximum of the channel gains in the network, \emph{i.e.}, $q=\max_{k,j} \left(  n_{kj} \right)$. At each time $i$, the received signals are given by
\begin{equation}
\begin{split}
&Y_{1,i} = \mathbf{S}^{q-n_{11}} X_{1,i} \oplus \mathbf{S}^{q-n_{21}} X_{2,i}, \\
&Y_{2,i} = \mathbf{S}^{q-n_{12}} X_{1,i} \oplus \mathbf{S}^{q-n_{22}} X_{2,i},
\end{split}
\end{equation}
where $\mathbf{S}$ is the $q \times q$ shift matrix and operations are performed in $\mathbb{F}_2$ (\emph{i.e.}, modulo two).  See Figure~\ref{fig:DIC} for an example.

\begin{figure}[!htp]
\centering
\includegraphics[height=5.5cm]{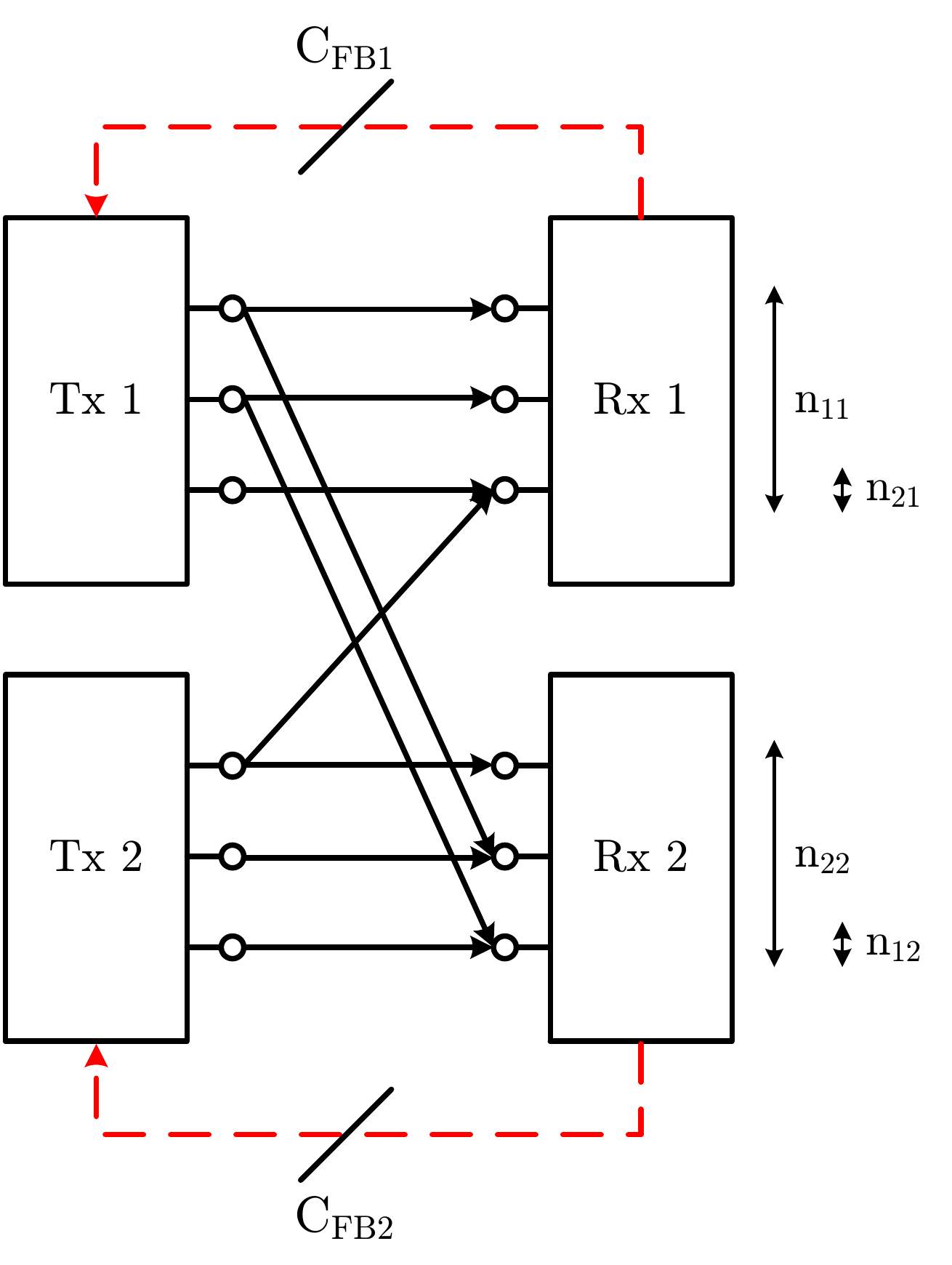}
\caption{An example of a linear deterministic IC with rate-limited feedback, where $n_{11} = n_{22} = 3$, $n_{12} = 2$, $n_{21} = 1$, and $q = 3$.} \label{fig:DIC}
\end{figure}


It is easy to see that this model also satisfies the conditions of (\ref{eq-ElGamalCostaCondition}), hence it is a special class of the El Gamal-Costa deterministic IC.

\noindent {\bf 3- Gaussian Interference Channel}:

In this model, there is a complex number $h_{kj}$ representing the channel from transmitter $k$ to receiver $j$, $k=1,2$, and $j=1,2$. The received signals are
\begin{equation}
\begin{split}
Y_{1,i} = h_{11} X_{1,i} + h_{21} X_{2,i} + Z_{1,i}, \\
Y_{2,i} = h_{12} X_{1,i} + h_{22} X_{2,i} + Z_{2,i},
\end{split}
\end{equation}
where $\{Z_{j,i}\}_{i=1}^N$ is the additive white complex Gaussian noise process with zero mean and unit variance at receiver $j$, $j=1,2$. Without loss of generality, we assume a power constraint of $1$ at all nodes, \emph{i.e.},
\begin{equation}
\frac{1}{N} \mathbb{E}(\sum_{i=1}^N{|X_{k,i}|^2}) \leq 1 \qquad k=1,2,
\end{equation}
where $N$ is the block length. We will use the following notations:
\begin{equation}
\begin{split}
{\sf SNR}_1 = |h_{11}|^2, &\hspace{4mm} {\sf SNR}_2 = |h_{22}|^2, \\
{\sf INR}_{12} = |h_{12}|^2, &\hspace{4mm} {\sf INR}_{21} = |h_{21}|^2.
\end{split}
\end{equation} 

\section{Motivating Example}
\label{motivation}
We start by analyzing a motivating example. Consider the linear deterministic IC with rate-limited feedback as depicted in Figure~\ref{fig:motivation}(a). As we will see in Section~\ref{linear}, the capacity region of this network is given by the region shown in Figure~\ref{fig:motivation}(b). Our goal in this section is to  demonstrate how feedback can help increase the capacity. In particular, we describe the achievability strategy for one of the corner points, \emph{i.e.}, $(R_1,R_2) = (4,1)$. From this example, we will make important observations that will later provide insights into the achievable scheme.

\begin{figure}
\centering
\subfigure[]{\includegraphics[height = 6cm]{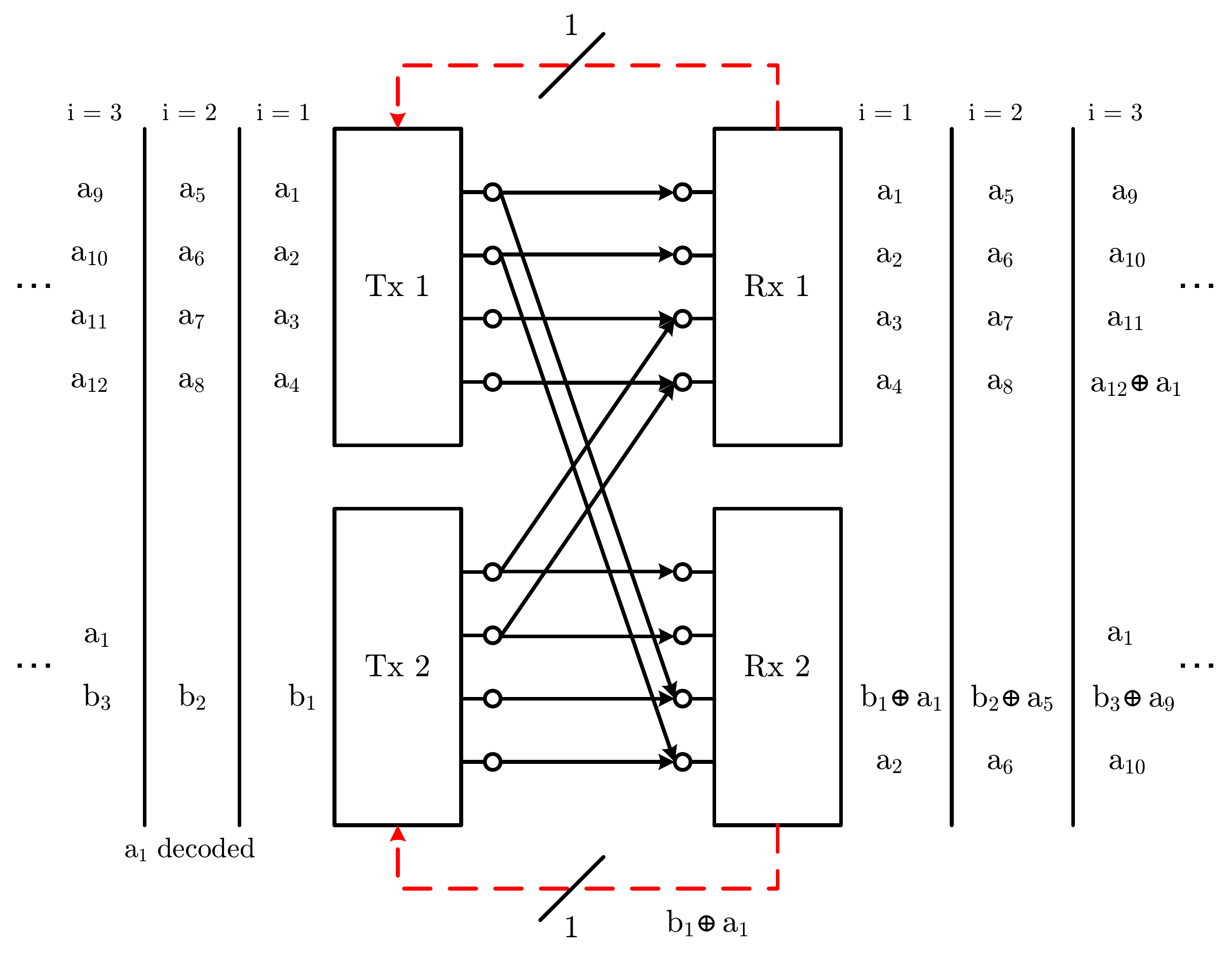}}
\subfigure[]{\includegraphics[height = 5cm]{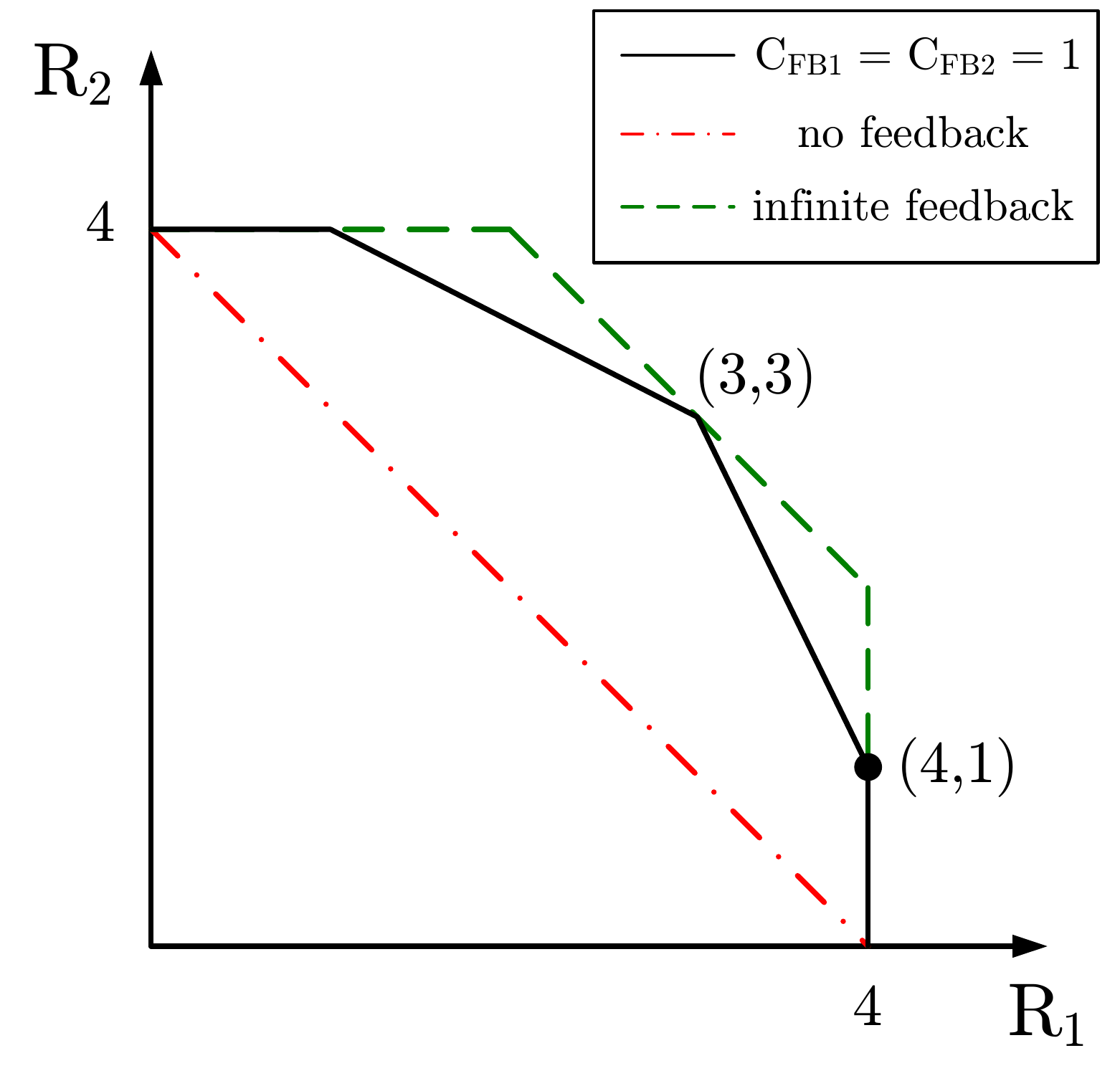}}
\caption{\it (a) A two-user linear deterministic IC with channel gains $n_{11} = n_{22} = 4$, $n_{12} = n_{21} = 2$ and feedback rates $C_{\sf FB1} = C_{\sf FB2} = 1$, and (b) its capacity region.\label{fig:motivation}}
\end{figure}

The achievability strategy works as follows. In the first time slot, transmitter $1$ sends four bits $a_1, \ldots, a_4$ and transmitter $2$ sends only one bit $b_1$ at the third level, see Figure~\ref{fig:motivation}(a). This way receiver $1$ can decode its intended four bits interference free, while receiver $2$ has access to only $a_1 \oplus b_1$ and $a_2$. In the second time slot, through the feedback link, receiver $2$ feeds $a_1 \oplus b_1$ back to transmitter $2$ who can remove $b_1$ from it to decode $a_1$. Also during the second time slot, transmitter $1$ sends four fresh bits $a_5, \ldots, a_8$, whereas transmitter $2$ sends one new bit $b_2$. In the third time slot, through the feedback link, receiver $2$ feeds $b_2 \oplus a_5$ back to transmitter $2$ who can remove $b_2$ from it to decode $a_5$. Moreover, during the third time slot, transmitter $1$ sends four new bits $a_9, \ldots, a_{12}$, while transmitter $2$ sends one new bit $b_3$ and at the level shown in Figure~\ref{fig:motivation}($a$), sends the other user's information bit $a_1$ decoded in the second time slot with the help of feedback. With this strategy receiver $2$ has now access to $a_1$ and can use it to decode $b_1$. Note that receiver $1$ already knows $a_1$ and hence can decode $a_{12}$. This procedure will be repeated over the next time slots. During the last two time slots, only transmitter $2$ sends the other user's information decoded before, while transmitter $1$ sends nothing. Therefore, after $B$ time slots, we achieve a rate of $(R_1,R_2) = \frac{B-2}{B}(4,1)$, which converges to $(4,1)$ as $B$ goes to infinity.

Based on this simple capacity-achieving strategy, we can now make several observations:

\begin{figure*}[h!tb]
\hrule
\begin{subequations}
\begin{eqnarray}
\label{lemmaeq:R1-1}
 R_1  &\leq & I(U,V_2, X_1;Y_1) \\
\label{lemmaeq:R1-2}
 R_1  &\leq &  I(X_1;Y_1|U, U_1, V_2) +  \min(I(U_1;Y_2|U, X_2), C_{\sf FB2}- \delta_2) \\
\label{lemmaeq:R2-1}
 R_2  &\leq &  I(U,V_1, X_2;Y_2) \\
\label{lemmaeq:R2-2}
 R_2 &\leq &  I(X_2;Y_2|U, U_2, V_1) + \min( I(U_2;Y_1|U, X_1), C_{\sf FB1}- \delta_1) \\
\label{lemmaeq:R1R2-1}
 R_1 + R_2 &\leq &   I(X_1;Y_1 |U, V_1, V_2) + I(U,V_1, X_2;Y_2) \\
\label{lemmaeq:R1R2-2}
 R_1 + R_2 &\leq & I(X_2;Y_2 |U, V_1, V_2) + I(U,V_2, X_1;Y_1) \\
\label{lemmaeq:R1R2-dummy1}
 R_1 + R_2  &\leq &   \min( I(U_2;Y_1|U, X_1), C_{\sf FB1}- \delta_1)  + \min(I(U_1;Y_2|U, X_2), C_{\sf FB2}- \delta_2) +I(X_1,V_2;Y_1|U, U_1, U_2) \nonumber \\  && + I(X_2;Y_2|U, V_1, V_2 ) \\
\label{lemmaeq:R1R2-dummy2}
  R_1 + R_2  &\leq &   \min( I(U_2;Y_1|U, X_1), C_{\sf FB1}- \delta_1)   + \min(I(U_1;Y_2|U, X_2), C_{\sf FB2}- \delta_2) + I(X_2,V_1;Y_2|U, U_1, U_2 ) \nonumber \\ && + I(X_1;Y_1 |U, V_1, V_2) \\
\label{lemmaeq:R1R2-3}
   R_1 + R_2  &\leq &   \min( I(U_2;Y_1|U, X_1), C_{\sf FB1}- \delta_1)   + \min(I(U_1;Y_2|U, X_2), C_{\sf FB2}- \delta_2) + I(X_1,V_2;Y_1|U, V_1, U_2 )  \nonumber \\&& + I(X_2, V_1; Y_2|U, V_2, U_1 ) \\
\label{lemmaeq:2R1R2}
 2R_1 + R_2  &\leq &  I(U,V_2, X_1;Y_1) + I(X_1;Y_1 |U, V_1, V_2) + I(X_2,V_1; Y_2|U, U_1, V_2 ) \\
 & + &  \min(I(U_1;Y_2|U, X_2), C_{\sf FB2}- \delta_2) \nonumber \\
\label{lemmaeq:2R1R2-dummy}
  2R_1 + R_2  &\leq &  2\min(I(U_1;Y_2|U, X_2), C_{\sf FB2}- \delta_2) + \min( I(U_2;Y_1|U, X_1), C_{\sf FB1}- \delta_1) + I(X_1,V_2; Y_1|U, U_1, U_2)  \nonumber \\
  && + I(X_1;Y_1 |U, V_1, V_2) + I(X_2,V_1; Y_2|U, U_1, V_2 ) \\
\label{lemmaeq:R1-2R2}
 R_1 + 2R_2  &\leq &  I(U,V_1, X_2;Y_2) + I(X_2;Y_2|U, V_2, V_1 ) + I(X_1,V_2; Y_1|U, U_2, V_1 ) \\
 & + &  \min( I(U_2;Y_1|U, X_1), C_{\sf FB1}- \delta_1)  \nonumber \\
\label{lemmaeq:R1-2R2-dummy}
 R_1 + 2R_2  &\leq &  2\min( I(U_2;Y_1|U, X_1), C_{\sf FB1}- \delta_1)  + \min(I(U_1;Y_2|U, X_2), C_{\sf FB2}- \delta_2)   + I(X_2,V_1; Y_2|U, U_1, U_2) \nonumber \\
 & & + I(X_2;Y_2|U, V_1, V_2 ) + I(X_1,V_2; Y_1|U, U_2, V_1)
\end{eqnarray}
\end{subequations}
over all joint distributions
\begin{align*}
p(u)p(u_1|u) p(u_2|u) p(v_1|u, u_1)p(v_2|u,u_2) p(x_1|u,u_1,v_1) p(x_2|u,u_2,v_2) p(\hat{y}_1|y_1) p(\hat{y}_2|y_2),
\end{align*}
where
\begin{align*}
\delta_1 &= I(\hat{Y}_1; Y_1|U, U_2, X_1), \\
\delta_2 &= I(\hat{Y}_2; Y_2|U, U_1, X_2).
\end{align*}
\hrule
\end{figure*}

$\bullet$ The messages coming from transmitter $1$ can be split into three parts: $(1)$ ``cooperative common'': this message is visible to both receivers, while  interfering with the other user's signals (e.g., $a_1$ at transmitter $1$ in the first time slot). This should be fed back to the transmitter so that it can be used later in refining the desired signals corrupted by the interfering signal; $(2)$ ``non-cooperative common'': this message is visible to both receivers, however it does not cause any interference (e.g., $a_2$ at transmitter $1$ in the first time slot); $(3)$ ``private'': this message is visible only to the intended receiver (e.g., $a_3$ and $a_4$ at transmitter $1$ in the first time sot). Denote these messages by  $w_{kcc}$, $w_{knc}$, and $w_{kp}$, respectively, where $k=1,2$ is the transmitter index.

$\bullet$ To refine the desired signal corrupted by the cooperative common signal of transmitter $1$ (\emph{i.e.}, $a_1$), receiver $2$ utilizes the feedback link to send the interfered signal (\emph{i.e.}, $a_1 \oplus b_1$) back to transmitter $2$. Transmitter $2$ then employs a partial decode-and-forward to help receiver $1$ decode its messages, \emph{i.e.}, the cooperative message of transmitter $1$ is decoded at transmitter $2$ and it will be forwarded to receiver $2$ during another time slot.

$\bullet$ As we can see in this example, encoding operations at each time slot depend on previous operations, thereby motivating us to employ block Markov encoding. As for the decoding, we implement backward decoding at receivers. Each receiver waits until the last time $B$ and we use the last received signal to decode the message received at time $B-2$. We then decode the message received at time $B-3$ and all the way down to the message received at time 1.

These observations will form the basis for our achievable schemes in the following sections.

\section{Deterministic Interference Channel}
\label{elgamal}
In this section,  we consider the El Gamal-Costa deterministic IC with rate-limited feedback, described in Section~\ref{problem}. The motivating example in the previous section leads us to develop a generic achievable scheme based on three ideas: (1) Han-Kobayashi message splitting~\cite{HanKoba:it81}; $(2)$ quantize-and-binning; and $(3)$ decode-and-forward~\cite{Cover:it79}. As mentioned earlier, we split the message into three parts: the cooperative common message; non-cooperative common message; and private message. We employ quantize-and-binning to feed back part of the interfered signals through the rate-limited feedback link. With feedback, each transmitter decodes part of the other user's common information (cooperative common) that interfered with its desired signals. We accomplish this by using the partial decode-and-forward scheme. We also derive a new outer bound based on the genie-aided argument~\cite{ElGamal:it82} and the dependence-balance-bound technique~\cite{Willems:it82, Willems:it89,Tandon:08}.

\subsection{Achievable Rate Region}

\begin{theorem}
\label{theorem:DICachievableregion}
The capacity region of the El Gamal-Costa deterministic IC with rate-limited feedback includes the set $\mathcal{R}$ of $(R_1, R_2)$ satisfying inequalities (\ref{lemmaeq:R1-1})--(\ref{lemmaeq:R1-2R2-dummy}).
\end{theorem}

\begin{proof}
We first provide an outline of our achievable scheme. We employ block Markov encoding with a total size $B$ of blocks. In block $1$, transmitter 1 splits its own message into cooperative common, non-cooperative common and private parts and then sends a codeword superimposing all of these messages. The cooperative common message is sent via the codeword $u_1^{N,(1)}$. The non-cooperative common message is added to this, being sent via $v_1^{N,(1)}$. The private message is then superimposed on top of the previous messages, being sent via $x_1^{N,(1)}$. Similarly, transmitter 2 sends $x_2^{N,(1)}$. In block 2, receiver 1 quantizes its received signal $y_1^{N,(1)}$ into $\hat{y}_1^{N,(1)}$ with the rate of $\hat{R}_1$. Next it generates a bin index by considering the capacity of its feedback link and then feeds the bin index back to its corresponding transmitter. Similarly, receiver 2 feeds back the corresponding bin index. In block 3, with feedback, each transmitter decodes the other user's cooperative common message (sent in block 1) that interfered with its desired signals. The following messages are then available at the transmitter: (1) its own message; and (2) the other user's cooperative common message decoded with the help of feedback.
Using its own cooperative common message as well as the other user's counterpart, each transmitter generates the codeword $u^{N,(3)}$. This captures the correlation between the two transmitters that might induce the cooperative gain. Conditioned on these previous cooperative common messages, each transmitter generates new cooperative common, non-cooperative common, and private messages. It then sends the corresponding codeword. This procedure is repeated until block $B-2$. In the last two blocks $B-1$ and $B$, to facilitate backward decoding, each transmitter sends the predetermined common messages and a new private message. Each receiver waits until total $B$ blocks have been received and then performs backward decoding.


\noindent {\bf Codebook Generation}: Fix a joint distribution $p(u)p(u_1|u) p(u_2|u) p(x_1|u_1, u)p(x_2|u_2,u) p(\hat{y}_1|y_1) p(\hat{y}_2|y_2)$. We will first show that $p(x_1|u, u_1, v_1)$ and $p(x_2|u, u_2, v_2)$ are functions of the above distributions. To see this, let us write a joint distribution $p(u,u_1,u_2, v_1, v_2, x_1,x_2)$ in two different ways:
\begin{align}
\begin{split}
&p(u,u_1, u_2, v_1,v_2,x_1,x_2) \\
&= p(u)p(u_1|u)p(u_2|u) p(x_1|u, u_1)p(x_2|u, u_2) \delta (v_1 - g_1(x_1))  \\ & \quad \times \delta (v_2-g_2(x_2)) \\
&= p(u)p(u_1|u)p(u_2|u) p(v_1|u, u_1)p(v_2|u,u_2) p(x_1|u,u_1,v_1) \\ & \quad \times p(x_2|u, u_2, v_2),
\end{split}
\end{align}
where $\delta(\cdot)$ indicates the Kronecker delta function. Notice that by the El Gamal-Costa model assumption (\ref{eq-ElGamalCostaCondition0}), $p(v_1|u, u_1, x_1) = \delta (v_1 - g_1 (x_1))$ and $p(v_2|u, u_2, x_2) = \delta (v_2 - g_2 (x_2))$. From this, we can easily see that
\begin{align*}
p(x_1|u, u_1, v_1) = \frac{p(x_1|u, u_1) \delta(v_1 - g_1(x_1))}{p(v_1|u, u_1)}, 
\end{align*}
\begin{align}
p(x_2|u, u_2, v_2) = \frac{p(x_2|u, u_2) \delta(v_2 - g_2(x_1))}{p(v_2|u, u_2)}.
\end{align}

We now generate codewords as follows. Transmitter $1$ generates $2^{N(R_{1cc} +R_{2cc}) }$ independent codewords $u^N(i,j)$ according to $\prod_{i=1}^{N} p(u_i)$\footnote{With a slight abuse of notation, we use the same index $i$ to represent time.}, where $i \in \{1, \cdots, 2^{NR_{1cc}} \}$ and $j\in \{1, \cdots, 2^{NR_{2cc}} \}$. For each codeword $u^N(i,j)$, it generates $2^{NR_{1cc}}$ independent codewords $u_1^{N}((i,j),k)$ according to $\prod_{i=1}^{N} p(u_{1i}|u_i)$ where $k \in \{1, \cdots, 2^{NR_{1cc}} \}$. Subsequently, for each pair of codewords $\left( u^N(i,j),u_1^N((i,j),k) \right) $, generate $2^{NR_{1nc}}$ independent codewords $v_1^N((i,j),k,l)$ according to $\prod_{i=1}^{N} p(v_{1i}|u_i, u_{1i})$ where $l \in \{1, \cdots, 2^{NR_{1nc}} \}$. Lastly, for each triple of codewords  $\left( u^N(i,j),u_1^N((i,j),k),v_1^N((i,j),k,l) \right) $, generate $2^{NR_{1p}}$ independent codewords $x_1^N((i,j),k,l,m)$ according to $\prod_{i=1}^{N} p(x_{1i}|u_i, u_{1i}, v_{1i})$ where $m \in \{1, \cdots, 2^{NR_{1p}} \}$. On the other hand, receiver 1 generates $2^{N \hat{R}_1}$ sequences $\hat{y}_1^N (q)$ according to $\prod_{i=1}^{N} p(\hat{y}_{1i})$ where $q \in \{1, \cdots, 2^{N\hat{R}_{1}} \}$. In the feedback strategy (to be described shortly), we will see how this codebook generation leads to the joint distribution $p(\hat{y}_1|y_1)$. Similarly, receiver 2 generates $\hat{y}_2^N$.

As it will be clarified later, for a given block $b$, indices $i$ and $j$ in $u^N(i,j)$ correspond to the cooperative common message of transmitter $1$ and transmitter $2$ sent during block $(b-2)$ respectively. Then $\prod_{i=1}^{N} p(u_{1i}|u_i)$ is used to create $2^{NR_{1cc}}$ independent codewords corresponding to the cooperative common message of transmitter $1$ in $\{1, \cdots, 2^{NR_{1cc}} \}$. Similarly, $2^{NR_{1nc}}$ independent codewords are created according to $\prod_{i=1}^{N} p(v_{1i}|u_i, u_{1i})$, corresponding to the non-cooperative common message of transmitter $1$ in $\{1, \cdots, 2^{NR_{1nc}} \}$. Finally, $\prod_{i=1}^{N} p(x_{1i}|u_i, u_{1i}, v_{1i})$ is used to create $2^{NR_{1p}}$ independent codewords corresponding to the private message of transmitter $1$ in $\{1, \cdots, 2^{NR_{1p}} \}$.


\textit{Notation:} Notations are independently used only for this section. The index $k$ indicates the cooperative common message of user 1 instead of user index. The index $i$ is used for both purposes: (1) indicating the previous cooperative common message of user 1; (2) indicating time index. It could be easily differentiated from contexts.

\noindent {\bf Feedback Strategy (Quantize-and-Binning)}:
Focus on the $b$-th transmission block. First receiver 1 quantizes its received signal $y_1^{N,(b)}$ into $\hat{y}_1^{N,(b+1)}$ with the rate of $\hat{R}$. Next it finds an index $q$ such that $\left( \hat{y}_1^{N,(b+1)}(q), y_{1}^{N,(b)} \right) \in \mathcal{T}_{\epsilon}^{(N)}$, where $q \in [1:2^{N \hat{R}_1}]$ and $\mathcal{T}_{\epsilon}^{(N)}$ indicates a jointly typical set.
The quantization rate $\hat{R}_1$ is chosen so as to ensure the existence of such an index with probability 1. The covering lemma in~\cite{ElGamalKim:NIT} guides us to choose $\hat{R}_1$ such that
\begin{align}
\label{eq:tildeR1}
\hat{R}_1  \geq I(\hat{Y}_1; Y_1),
\end{align}
since under the above constraint, the probability that there is no such an index becomes arbitrarily small as $N$ goes to infinity. Notice that with this choice of $\hat{R}_1$, the codebook $\hat{y}_1^N (q)$ according to $\prod_{i=1}^N p(\hat{y}_{1i})$ would match the codeword according to $\prod_{i=1}^N p(\hat{y}_{1i}|y_{1i})$.

We then partition the set of indices $q\in [1:2^{N \hat{R}_1}]$ into the number $2^{N C_{\sf FB1}}$ of equal-size subsets (that we call {\em bins}):

\begin{align*}
\begin{split}
 \mathcal{B}(r)=& \left[ (r-1) 2^{N(\hat{R}_{1}- C_{\sf FB1}) }+1: r 2^{N(\hat{R}_{1}- C_{\sf FB1}) } \right], \\
 & \; r \in [1:2^{N C_{\sf FB1}}].
\end{split}
\end{align*}

Now the idea is to feed back the bin index $r$ such that $q \in \mathcal{B}(r)$. This helps transmitter 1 to decode the quantized signal $\hat{y}_1^{N,(b)}$. Specifically, using the bin index $r$, transmitter 1 finds a unique index $q \in \mathcal{B}(r)$ such that
$\left( \hat{y}_1^{N,(b+1)}(q), x_{1}^{N,(b)}, u^{N, (b)} \right) \in \mathcal{T}_{\epsilon}^{(N)}$.
Notice that by the packing lemma in~\cite{ElGamalKim:NIT}, the decoding error probability goes to zero if
\begin{align}
\label{eq:constraintY1tilde}
\hat{R}_1 - C_{\sf FB1} \leq I(\hat{Y}_1; X_1, U).
\end{align}
Using~(\ref{eq:tildeR1}) and~(\ref{eq:constraintY1tilde}), transmitter 1 can now decode the quantized signal as long as
\begin{align}
\label{eq:constraint2Y1tilde}
I(\hat{Y}_1; Y_1 |X_1, U ) \leq C_{\sf FB1},
\end{align}
Similarly, transmitter 2 can decode $\hat{y}_2^{N,(b+1)}(q)$ if
\begin{align}
\label{eq:constraint2Y2tilde}
I(\hat{Y}_2; Y_2 |X_2, U ) \leq C_{\sf FB2}.
\end{align}

\noindent {\bf Encoding}: Given $\hat{y}_1^{N,(b-1)}$ (decoded with the help of feedback), transmitter 1 finds a unique index $\hat{w}_{2cc}^{(b-2)} = \hat{k}$ (sent from transmitter 2 in the $(b-2)$-th block) such that
\begin{align*}
\begin{split}
&\left( u^N \left( \cdot \right), u_1^N \left( \cdot \right), v_1^N (\cdot), x_1^N (\cdot), u_2^N ( \cdot, \hat{k} ) , \hat{y}_1^{N,(b-1)}  \right) \in \mathcal{T}_{\epsilon}^{(N)},
\end{split}
\end{align*}
where $(\cdot)$ indicates the known messages $(w_{1cc}^{(b-4)}, \hat{w}_{2cc}^{(b-4)}, w_{1nc}^{(b-2)}, w_{1p}^{(b-2)})$. Notice that due to the feedback delay, the fed back signal contains information of the $(b-2)$-th block. We assume that $\hat{w}_{2cc}^{(b-2)}$ is correctly decoded from the previous block. By the packing lemma~\cite{ElGamalKim:NIT}, the decoding error probability becomes arbitrarily small (as $N$ goes to infinity) if
\begin{align}
\nonumber R_{2cc} &\leq I(U_2;\hat{Y}_1|X_1,U) \\ \nonumber &= I(\hat{Y}_1; Y_1|U, X_1) - I(\hat{Y}_1;Y_1|U, U_2, X_1)\\
\label{eq-R2c-constraint} &\leq  \min ( C_{\sf FB1} - \delta_1, I(U_2;Y_1|X_1,U)  ),
\end{align}
where the last inequality follows from~(\ref{eq:constraint2Y1tilde}), $\delta_1:= I(\hat{Y}_1; Y_1| U, U_2, X_1)$ and $I(U_2;\hat{Y}_1|X_1,U) \leq I(U_2;Y_1|X_1,U)$.

Based on $(w_{1cc}^{(b-2)},\hat{w}_{2cc}^{(b-2)})$, transmitter 1 generates a new cooperative-common message $w_{1cc}^{(b)}$, a non-cooperative-common message $w_{1nc}^{(b)}$ and a private message $w_{1p}^{(b)}$. It then sends $x_1^N$. Similarly transmitter 2 decodes $\hat{w}_{1cc}^{(b-2)}$, generates $(w_{2cc}^{(b)},w_{2nc}^{(b)}, w_{2p}^{(b)})$ and then sends $x_2^N$.

\noindent {\bf Decoding}: Each receiver waits until total $B$ blocks have been received and then does backward decoding.
Notice that a block index $b$ starts from the last $B$ and ends to $1$. For block $b$, receiver 1 finds the unique indices $(\hat{i},\hat{j},\hat{k}, \hat{l})$ such that for some $m \in [1:2^{N R_{2nc}}]$
\begin{align*}
\begin{split}
& \left( u^N ( \hat{i}, \hat{j} ), u_1^N ( ( \hat{i}, \hat{j}) , \hat{w}_{1cc}^{(b)} ), v_1^N ( ( \hat{i}, \hat{j}) , \hat{w}_{1cc}^{(b)} , \hat{k} ),  x_1^N ( (\hat{i}, \hat{j}), \hat{w}_{1cc}^{(b)} \right. \\ & \left. ,  \hat{k}, \hat{l}),  u_2^N ( (\hat{i}, \hat{j}), \hat{w}_{2cc}^{(b)} ), v_2^N ( (\hat{i}, \hat{j}), \hat{w}_{2cc}^{(b)}, m ) , y_1^{N,(b)}  \right) \in \mathcal{T}_{\epsilon}^{(N)},
\end{split}
\end{align*}
where we assumed that a pair of messages $(\hat{w}_{1cc}^{(b)}, \hat{w}_{2cc}^{(b)})$ was successively decoded from the future blocks. Similarly receiver 2 decodes $(\hat{w}_{1cc}^{(b-2)}, \hat{w}_{2cc}^{(b-2)}, \hat{w}_{2nc}^{(b)}, \hat{w}_{2p}^{(b)})$.

\noindent {\bf Analysis of Probability of Error}: By symmetry, we consider the probability of error only for block $b$ and for a pair of transmitter 1 and receiver 1. We assume that $(w_{1cc}^{(b-2)},w_{2cc}^{(b-2)}, w_{1nc}^{(b)}, w_{1p}^{(b)}) = (1,1,1,1)$ was sent through the blocks; and there was no backward decoding error from the future blocks, \emph{i.e.}, $(\hat{w}_{1cc}^{(b)}, \hat{w}_{2cc}^{(b)})$ are successfully decoded.

Define an event:

\begin{align*}
E_{ijklm} = \left\{ \left( u^N ( i, j ), u_1^N ( ( i, j) , \hat{w}_{1cc}^{(b)} ), v_1^N ( ( i, j) , \hat{w}_{1cc}^{(b)}, k )  \right. \right. \\
 \left. \left. ,  x_1^N ( (i, j), \hat{w}_{1cc}^{(b)},  k, l),  u_2^N ( (i, j), \hat{w}_{2cc}^{(b)} ) , v_2^N ( (i, j), \hat{w}_{2cc}^{(b)}, m ) \right. \right. \\
 \left. \left., y_1^{N,(b)}  \right) \in \mathcal{T}_{\epsilon}^{(N)} \right\}.
\end{align*}
Let $E_{1111m}^c$ be the complement of the set $E_{1111m}$. Then, by AEP, $\textrm{Pr} (E_{1111m}^c) \rightarrow 0$ as $N$ goes to infinity. Hence, we focus only on the following error event.
\begin{align}
\begin{split}
\label{eq-errorprobability}
& \textrm{Pr} \left( \bigcup_{(i,j,k,l)  \neq (1,1,1,1), m } E_{ijklm}   \right) \\
& \quad \leq \underbrace{ \sum_{(i,j) \neq (1,1)} \textrm{Pr} \left( \bigcup_{k,l,m} E_{ijklm} \right)}_{\triangleq \textrm{Pr}(E_1)}  \\& \quad +
\underbrace{ \sum_{k \neq 1} \textrm{Pr} \left( \bigcup_{l} E_{11kl1} \right)}_{\triangleq \textrm{Pr}(E_2)}
+ \underbrace{ \sum_{k \neq 1, m \neq 1} \textrm{Pr} \left( \bigcup_{l} E_{11klm} \right)}_{\triangleq \textrm{Pr}(E_3)} \\ & \quad + \underbrace{ \sum_{l \neq 1} \textrm{Pr} \left( E_{111l1} \right)}_{\triangleq \textrm{Pr}(E_4)}
+ \underbrace{ \sum_{l \neq 1, m \neq 1 } \textrm{Pr} \left( E_{111lm} \right)}_{\triangleq \textrm{Pr}(E_5)}. \end{split} \end{align}
Here we have:
\begin{align*}
& \nonumber \textrm{Pr}(E_1) \leq 24 \\
& \quad \quad \quad \times 2^{N(R_{1cc}+ R_{2cc}+ R_{1nc} + R_{2nc} + R_{1p} - I(U,X_1,V_2;Y_1)+ 5 \epsilon)}\\
& \nonumber \textrm{Pr}(E_2) \leq 2 \times 2^{N(R_{1nc}+R_{1p} - I(X_1;Y_1|U,U_1,V_2)+  2 \epsilon)} \\
& \nonumber \textrm{Pr}(E_3) \leq 2\times 2^{N(R_{1nc}+ R_{2nc} + R_{1p} - I(X_1, V_2 ;Y_1|U,U_1, U_2)+ 3 \epsilon)}\\
& \nonumber \textrm{Pr}(E_4) \leq 2^{N(R_{1p} - I(X_1;Y_1|U,V_1,V_2)+  \epsilon)} \\
& \nonumber \textrm{Pr}(E_5) \leq 2^{N(R_{2nc} + R_{1p} - I(X_1,V_2;Y_1|U,V_1, U_2)+ 2 \epsilon)}.
\end{align*}
Notice in $\textrm{Pr}(E_1)$ that as long as $(i,j) \neq (1,1)$, all of the cases (decided depending on whether or not $k \neq 1$, $l \neq 1$ and $m \neq 1$) are dominated by the worst case bound that occurs when $k \neq 1$, $l \neq 1$ and $m \neq 1$. Since $(i,j) \neq (1,1)$ covers three different cases and we have eight different cases depending on the values of $(k,l,m)$, we have 24 cases in total. This number reflects the constant $24$ in the above first inequality. Similarly, we get the other four inequalities as above.

Hence, the probability of error can be made arbitrarily small if
\begin{align}
&\left\{
  \begin{array}{l}
    R_{2cc} \leq \min(I(U_2;Y_1|U,X_1),C_{\sf FB1} - \delta_1 ) \\
    R_{1p} \leq I (X_1; Y_1| U, V_1, V_2) \\
    R_{1p} + R_{2nc} \leq I(X_1,V_2;Y_1|U,V_1, U_2)  \\
    R_{1p} + R_{1nc} \leq I(X_1;Y_1|U,U_1, V_2)  \\
   R_{1p} + R_{1nc} + R_{2nc} \leq I(X_1,V_2;Y_1|U,U_1, U_2) \\
   R_{1p} + R_{1cc} + R_{2cc} +  R_{1nc} + R_{2nc} \leq I(U, V_2, X_1;Y_1), \\
  \end{array}
\right. \\
&\left\{
   \begin{array}{l}
    R_{1cc} \leq \min ( I(U_1;Y_2|U,X_2), C_{\sf FB2}- \delta_2 ) \\
    R_{2p} \leq I (X_2; Y_2| U, V_1, V_2) \\
    R_{2p} + R_{1nc} \leq I(X_2,V_1;Y_2|U,V_2, U_1) \\
    R_{2p} + R_{2nc} \leq I(X_2;Y_2|U,U_2, V_1) \\
   R_{2p} + R_{2nc} + R_{1nc} \leq I(X_2,V_1;Y_2|U,U_1, U_2) \\
   R_{2p} + R_{2cc} + R_{1cc} +  R_{1nc} + R_{2nc} \leq I(U, V_1, X_2;Y_2). \\
  \end{array}
\right.
\end{align}
Employing Fourier-Motzkin-Elimination, we finally get the bounds of (\ref{lemmaeq:R1-1})--(\ref{lemmaeq:R1-2R2-dummy}).
\end{proof}

\begin{remark}[Connection to Related Work~\cite{Tuninetti:isit07,Vinod:arix09}]
The three-fold message splitting in our achievable scheme is a special case of the more-than-three-fold message splitting introduced in~\cite{Tuninetti:isit07,Vinod:arix09}. Also our scheme has similarity to the schemes in~\cite{Tuninetti:isit07,Vinod:arix09} in a sense that the three techniques (message-splitting, block Markov encoding and backward decoding) are jointly employed. However, due to a fundamental difference between our rate-limited feedback problem and the conferencing encoder problem in~\cite{Tuninetti:isit07,Vinod:arix09}, a new scheme is required for feedback strategy and this is reflected as the quantize-and-binning scheme in the El Gamal-Costa deterministic model. It turns out this distinction leads to a new lattice-code-based scheme in the Gaussian case, as will be explained in Section VI.
\end{remark}

\subsection{Outer-bound}

\begin{theorem}
\label{theorem:outerboundregion}
The capacity region of the two-user El Gamal-Costa deterministic IC with rate limited feedback (as described in Section~\ref{problem}) is included by the set $\mathcal{\bar{C}}$ of $(R_1, R_2)$ such that
\begin{subequations}
\begin{eqnarray}
\label{eq:outerR1_1}
  R_1 & \leq & \min \{ H(Y_1) , H(Y_1|V_1, V_2, U_1) \\
  & + &  H(Y_2|X_2, U_1) \}  \nonumber \\
 \label{eq:outerR1_2}
  R_1 & \leq & H(Y_1|X_2, U_1) +  C_{\sf FB2} \\
  \label{eq:outerR2_1}
    R_2 & \leq & \min \{ H(Y_2) , H(Y_2|V_1, V_2, U_2) \\
    & + &  H(Y_1|X_1, U_2) \} \nonumber \\
  \label{eq:outerR2_2}
  R_2 & \leq & H(Y_2|X_1, U_2) + C_{\sf FB1} \\
\label{eq:outerR12_2}
  R_1 + R_2 & \leq & H(Y_1 |V_1, V_2, U_1) + H(Y_2 ) \\
  \label{eq:outerR12_3}
  R_1 + R_2 & \leq & H(Y_2 |V_1, V_2, U_2) + H(Y_1 ) \\
  \label{eq:outerR12_1}
R_1 + R_2 & \leq & H(Y_1|V_1 ) + H(Y_2|V_2 ) \\
& + &  C_{\sf FB1} + C_{\sf FB2} \nonumber \\
  \label{eq:outer2R1R2_1}
2R_1 + R_2 & \leq & H(Y_1) + H(Y_1|V_1, V_2, U_1 ) \\
& + & H(Y_2|V_2 ) +  C_{\sf FB2} \nonumber \\
\label{eq:outer2R1R2_2}
R_1 + 2R_2  & \leq & H(Y_2) + H(Y_2|V_1, V_2, U_2 ) \\
& + & H(Y_1|V_1 ) +  C_{\sf FB1}, \nonumber
\end{eqnarray}
\end{subequations}
for joint distributions $p(u_1, u_2) p(x_1|u_1, u_2) p(x_2|u_1,u_2)$. As depicted in Figure~\ref{fig_ElGamalCosta}, $C_{\sf FB1}$ and $C_{\sf FB2}$ indicate the capacity of each feedback link.
\end{theorem}

\begin{remark}
In the non-feedback case, \emph{i.e.}, $C_{\sf FB1} = C_{\sf FB2} = 0$, by setting $U_1 = U_2 = \emptyset$, we recover the outer-bounds of Theorem $1$ in \cite{ElGamal:it82}. Note that in this case, $H(Y_1|X_2) = H(Y_1|V_2)$ and $H(Y_2|X_1) = H(Y_2|V_1)$. In fact, our achievable region of Theorem~\ref{theorem:DICachievableregion} matches the outer-bound under this model, thereby achieving the non-feedback capacity region.
\end{remark}
\begin{remark}[Feedback gain under symmetric feedback cost] Notice from (\ref{eq:outerR12_1}) that the sum-rate capacity can be at most increased by the rate of available feedback, \emph{i.e.}, one bit of feedback provides a capacity increase of at most one bit. Therefore, if the cost of using the feedback link is the same as that of using the forward link, there is no feedback gain under the feedback cost. However, it turns out that there is indeed feedback gain when the costs are \emph{asymmetric}. This will be discussed in more details in Remark~\ref{remark:feedbackgain} of Section~\ref{linear}.
\end{remark}

\begin{proof}
By symmetry, it suffices to prove the bounds of~(\ref{eq:outerR1_1}), (\ref{eq:outerR1_2}), (\ref{eq:outerR12_2}), (\ref{eq:outerR12_1}), and (\ref{eq:outer2R1R2_1}). The bounds of (\ref{eq:outerR1_1}) and (\ref{eq:outerR1_2}) are nothing but the cutset bounds (see Appendix~\ref{deterministic-proofs} for details). Also  (\ref{eq:outerR12_2}) is the bound when the feedback link has infinite capacity~\cite{Suh}. Hence, proving the bounds of (\ref{eq:outerR12_1}) and (\ref{eq:outer2R1R2_1}) is the main focus of the proof. We will present the proof of (\ref{eq:outerR12_1}) here and for completeness, the proof for all other bounds is provided in Appendix~\ref{deterministic-proofs}.

\textbf{Proof of (\ref{eq:outerR12_1}):}
\begin{align*}
\begin{split}
N&(R_1 + R_2- \epsilon_N) \leq I(W_1;Y_1^{N}) + I(W_2;Y_2^{N})  \\
&= H(Y_1^{N}) - H(Y_1^N|W_1)  + H(Y_2^{N}) -H(Y_2^{N}|W_2) \\
&\overset{(a)}{=}  H(Y_1^{N}) - H(V_1^N|W_2)  + H(Y_2^{N}) -H(V_2^{N}|W_1) \\
&= H(Y_1^{N}) - [H(V_1^N) - I(V_1^N;W_2) ] + H(Y_2^{N}) \\& - [ H(V_2^{N}) - I(V_2^N; W_1) ] \\
&\overset{(b)}{\leq} I(V_1^N;W_2) +  I(V_2^N; W_1) + H(Y_1^{N}, V_1^{N}) - H(V_1^N) \\& + H(Y_2^{N}, V_2^N) -  H(V_2^{N}) \\
&= I(V_1^N;W_2) +  I(V_2^N; W_1) + H(Y_1^{N}|V_1^N)  + H(Y_2^{N}|V_2^N)  \\
&\overset{(c)}{\leq} I(V_1^N, W_1, \tilde{Y}_1^{N};W_2) +  I(V_2^N, W_2, \tilde{Y}_2^N; W_1) \\&+ H(Y_1^{N}|V_1^N)  + H(Y_2^{N}|V_2^N)  \\
&\overset{(d)}{=} I(W_1, \tilde{Y}_1^{N};W_2) +  I(W_2, \tilde{Y}_2^N; W_1) + H(Y_1^{N}|V_1^N) \\ & + H(Y_2^{N}|V_2^N)  \\
&=  I(\tilde{Y}_1^{N};W_2|W_1) +  I(\tilde{Y}_2^N; W_1|W_2) + H(Y_1^{N}|V_1^N)  + \\
& H(Y_2^{N}|V_2^N)  \\
&=  H(\tilde{Y}_1^{N}|W_1) +  H(\tilde{Y}_2^N|W_2) + H(Y_1^{N}|V_1^N)  + H(Y_2^{N}|V_2^N)  \\
&\overset{(e)}{\leq} N (C_{\sf FB1} + C_{\sf FB2}) + \sum H(Y_{1i} | V_{1i} ) + \sum H(Y_{2i} | V_{2i} ),
\end{split}
\end{align*}
where ($a$) follows from $H(Y_1^N|W_1)=H(V_2^N|W_1)$ and $H(Y_2^N|W_2)=H(V_1^N|W_2)$ (see Claim~\ref{claim1} below); ($b$) follows from providing $V_1^N$ and $V_2^N$ to receiver 1 and 2, respectively; ($c$) follows from the fact that adding information increases mutual information; ($d$) follows from the fact that $V_k^N$ is a function of $(W_k,\tilde{Y}_k^{N-1})$; ($e$) follows from the fact that $H(\tilde{Y}_k^N|W_k) \leq N C_{\sf FB k}$ and conditioning reduces entropy.

\begin{claim}
\label{claim1}
$H(Y_1^N|W_1)=H(V_2^N|W_1)$ and $H(Y_2^N|W_2)=H(V_1^N|W_2)$.
\end{claim}
\begin{proof}
By symmetry, it suffices to prove the first one.
\begin{align*}
\begin{split}
H&(Y_1^{N}|W_1) = \sum H(Y_{1i}|Y_1^{i-1},W_1) \\
&\overset{(a)}{=} \sum H(V_{2i}|Y_1^{i-1},W_1) \\
&\overset{(b)}{=} \sum H(V_{2i}|Y_1^{i-1},W_1,X_1^{i},V_2^{i-1}) \\
&\overset{(c)}{=}  \sum H(V_{2i}|W_1,V_2^{i-1})=H(V_2^{N}|W_1),
\end{split}
\end{align*}
where ($a$) follows from the fact that $Y_{1i}$ is a function of $(X_{1i},V_{2i})$ and $X_{1i}$ is a function of $(W_1, Y_1^{i-1})$;
($b$) follows from the fact that $X_1^{i}$ is a function of $(W_1,Y_1^{i-1})$ and $V_2^{i-1}$ is a function of $(Y_1^{i-1},X_1^{i-1})$; ($c$) follows from the fact that $Y_{1}^{i-1}$ is a function of $(X_1^{i-1},V_2^{i-1})$ and $X_1^{i}$ is a function of  $(W_1,V_2^{i-1})$ (see~Claim \ref{claim2} below).
\end{proof}

\begin{claim}
\label{claim2}
For $i\geq 1$, $X_1^{i}$ is a function of $(W_1,V_2^{i-1})$. Similarly, $X_2^{i}$ is a function of $(W_2,V_1^{i-1})$.
\end{claim}
\begin{proof}
By symmetry, it is enough to prove the first one.
Since the channel is deterministic, $X_1^{i}$ is a function of $(W_1, W_2)$. In Figure~\ref{fig_ElGamalCosta}, we see that information of $W_2$ to the first link pair must pass through $V_{2i}$. Also note that $X_{1i}$ depends on the past output sequences until $i-1$. Therefore, $X_1^{i}$ is a function of $(W_1,V_2^{i-1})$.
\end{proof}

\end{proof}

\section{Linear Deterministic Interference Channel}
\label{linear}
In this section, we consider the linear deterministic IC with rate-limited feedback described in Section~\ref{problem}. Since this model is a special case of the El Gamal-Costa model, our inner and outer bounds derived in the previous section also apply to this model. We show that the inner-bound and the outer bound derived in Theorem~\ref{theorem:DICachievableregion} and~\ref{theorem:outerboundregion} respectively, coincide under this linear deterministic model, thus establishing the capacity region.

\begin{theorem}
\label{theorem:DICcapacity}
The capacity region of the linear deterministic IC with rate-limited feedback is the set of non-negative $(R_1,R_2)$ satisfying
\begin{subequations}
\begin{eqnarray}
\label{eq:R1bound-1}
  R_1 & \leq & \min\left\{ \max( n_{11}, n_{21}), \max( n_{11}, n_{12}) \right\} \\
\label{eq:R1bound-2}
  R_1 & \leq & n_{11} +   C_{\sf FB 2}  \\
\label{eq:R2bound-1}
    R_2 & \leq & \min\left\{ \max( n_{22}, n_{12}), \max( n_{22}, n_{21}) \right \} \\
\label{eq:R2bound-2}
  R_2 & \leq & n_{22} +   C_{\sf FB 1} \\
\label{eq:R1R2bound-1}
  R_1 + R_2 & \leq & (n_{11} - n_{12})^+  +  \max( n_{22}, n_{12})  \\
\label{eq:R1R2bound-2}
  R_1 + R_2 & \leq & (n_{22} - n_{21})^+  +  \max( n_{11}, n_{21}) \\
\label{eq:R1R2bound-3}
R_1 + R_2 & \leq & \max \left\{ n_{21}, (n_{11} - n_{12})^+ \right\} \\  & + & \max \left\{ n_{12}, (n_{22} - n_{21})^+ \right\}  + C_{\sf FB1} +  C_{\sf FB2} \nonumber  \\
\label{eq:2R1R2bound}
2R_1 + R_2 & \leq & (n_{11} - n_{12})^+  +  \max( n_{11}, n_{21})  \\  & + & \max \left\{ n_{12},(n_{22}-n_{21})^+ \right\} +  C_{\sf FB2} \nonumber \\
\label{eq:R12R2bound}
R_1 + 2R_2  & \leq & (n_{22} - n_{21})^+ +  \max( n_{22}, n_{12}) \\  & + & \max \left\{n_{21},(n_{11}-n_{12})^+ \right\}   +  C_{\sf FB1}. \nonumber
\end{eqnarray}
\end{subequations}
\end{theorem}

\begin{remark}
\label{remark:ISIT}
In the non-feedback case, \emph{i.e.}, $C_{\sf FB1} = C_{\sf FB2}=0$, this theorem recovers the result of~\cite{ElGamal:it82,Guy}. In the infinite feedback case, \emph{i.e.}, $C_{\sf FB1} = C_{\sf FB2}=\infty$, this recovers the result of~\cite{Suh:allerton09,Suh}.
Considering the sum-rate capacity under symmetric setting, \emph{i.e.}, $n_{11} = n_{22} = n$, $n_{12} = n_{21} = m$, $C_{\sf FB1} = C_{\sf FB2}$, this recovers the result of~\cite{AlirezaISIT}.
\end{remark}

\begin{proof}
The converse proof is trivial due to Theorem~\ref{theorem:outerboundregion}. For achievability, we will use the result in Theorem~\ref{theorem:DICachievableregion}. By choosing the following input distribution, we will show the tightness of the outer bound. $\forall k \in \{ 1, 2 \}$ and $j \neq k$, we choose
\begin{eqnarray}
\label{choice}
\begin{split}
&U = \varnothing,   \\
&U_k = U \oplus X_{kcc},   \\
&V_k = U_k \oplus X_{knc}, \\
&X_k = V_k \oplus X_{kp},  \\
&\hat{Y}_k = {\sf LSB}_{\min(n_{kj},C_{\sf FBj})}(Y_k),
\end{split}
\end{eqnarray}
where for any column vector $A$, ${\sf LSB}_{n}(A)$ takes the bottom $n$ ($n \leq |A|$) entries of $A$ while returning zeros for the remaining part; and $X_{kcc}, X_{knc},$  and $X_{kp}$ are independent random vectors of size $\max \{ n_{kk}, n_{kj} \}$, such that
\begin{itemize}
\item The random vector $X_{kp}$ consists of $(n_{kk} - n_{kj})^+$ i.i.d. ${\sf Ber} \left(\frac{1}{2} \right)$ random variables at the bottom, denoted by $*$ in (\ref{eq:VariablesDet}), corresponding to the number of private signal levels of transmitter $k$.
\item The random vector $X_{kcc}$ consists of $(n_{kk} - n_{kj})^+$  i.i.d. ${\sf Ber} \left(\frac{1}{2} \right)$ random variables in the middle (above the private signal levels), denoted by $*$ in (\ref{eq:VariablesDet}),  corresponding to the number of common signal levels that will be re-sent cooperatively through the other commuication link with the help of feedback.
\item The random vector $X_{knc}$ consists of $(n_{kj} - C_{\sf FBj } )^+$  i.i.d. ${\sf Ber} \left(\frac{1}{2} \right)$ random variables at the top, denoted by $*$ in (\ref{eq:VariablesDet}),  corresponding to the number of non-cooperative common signal levels.
\end{itemize}
As we show in Appendix~\ref{Appendix:DICcapacity}, with this choice of random variables, the achievable region of Theorem~\ref{theorem:DICachievableregion} matches the outer-bounds in Theorem~\ref{theorem:outerboundregion}.


\begin{align}
\label{eq:VariablesDet}
\begin{split}
X_{k} =
\underbrace{\left[
          \begin{array}{c}
            0 \\
            \vdots \\
            0 \\
\hline
            0 \\
            \vdots \\
            0 \\
\hline
            * \\
            \vdots \\
            * \\
          \end{array}
\right]}_{X_{kp}}
\oplus
\underbrace{\left[
          \begin{array}{c}
            0 \\
            \vdots \\
            0 \\
\hline
            * \\
            \vdots \\
            * \\
\hline
            0 \\
            \vdots \\
            0 \\
          \end{array}
\right]}_{X_{kcc}}
\oplus
\underbrace{\left[
          \begin{array}{c}
            * \\
            \vdots \\
            * \\
\hline
            0 \\
            \vdots \\
            0 \\
\hline
            0 \\
            \vdots \\
            0 \\
          \end{array}
\right]}_{X_{knc}}, \hspace{4mm} k=1,2.
\end{split}
\end{align}
\end{proof}

It is worth utilizing Theorem~\ref{theorem:DICcapacity} to illustrate the impact of feedback on the sum-rate capacity of the linear deterministic IC. Consider a symmetric case where $n_{11}=n_{22}=n$, $n_{12} = n_{21} = \alpha n$, and $C_{\sf FB1}=C_{\sf FB2}=\beta n$. Using Theorem~\ref{theorem:DICcapacity}, we can derive the sum-rate capacity of this network (normalized by $n$)
\begin{equation}
\label{ISITresult}
\frac{C_{\sf sum}}{n}=
\left\{ \begin{array}{ll}
\min(2-2\alpha+2\beta,2-\alpha) & \textrm{for $\alpha \in [0,0.5]$}\\
\min(2\alpha+2\beta,2-\alpha) & \textrm{for $\alpha \in [0.5,\frac{2}{3}]$}\\
2-\alpha & \textrm{for $\alpha \in [\frac{2}{3},1]$}\\
\alpha & \textrm{for $\alpha \in [1,2+2\beta]$}\\
2+2\beta & \textrm{for $\alpha \in [2+2\beta,\infty)$}
\end{array} \right.
\end{equation}

Figure \ref{fig:ISITfig} illustrates the (normalized) sum-rate capacity as a function of $\alpha$,  for different values of $\beta=0$ (\emph{i.e.}, no feedback), $\beta=\infty$ (\emph{i.e.}, infinite feedback), and $\beta=0.125$. We note the following cases:
\begin{itemize}
\item Case $1$ ($\alpha \in \left[ 0,\frac{1}{2} \right]$): In this regime the sum-rate capacity is increased by the total amount of feedback rates and saturates at $2-\alpha$ once the rate of each feedback link is larger than $\frac{\alpha n}{2}$. \vspace{.2mm}
\item Case $2$ ($\alpha \in \left[ \frac{1}{2},\frac{2}{3} \right]$): In this regime the sum-rate capacity is increased by the total amount of feedback rates and saturates at $2\alpha$ once the rate of each feedback link is larger than $\frac{(2 n - 3 \alpha n)}{2}$. \vspace{.2mm}
\item Case $3$ ($\alpha \in \left[ \frac{2}{3},2 \right]$): In this regime feedback does not increase the capacity. \vspace{.2mm}
\item Case $4$ ($\alpha \in [2+2\beta,\infty)$): In this regime the sum-rate capacity is increased by at most the total amount of feedback rates.
\end{itemize}

\begin{remark}[Feedback gain under asymmetric feedback cost]
\label{remark:feedbackgain}
As it can be seen in~(\ref{ISITresult}) the sum-rate capacity is increased by at most the total amount of feedback rates. Let the {\em cost} be the amount of resources (e.g., time, frequency) paid for sending one bit. With this cost in mind, let us consider the effective gain of using feedback which counts the cost. Notice that by Case 1,2, and 4, there are many channel parameter scenarios where one bit of feedback can provide a capacity increase of exactly one bit. This implies that the effective feedback gain depends on the cost difference between feedback and forward links. So if the feedback cost is cheaper than that of using forward link, then there is indeed feedback gain. The cellular network may be this case. Suppose that downlink is used for feedback purpose, while uplink is used as a forward link. Then, this is the scenario where the feedback cost is cheaper than the cost of using the forward link, as downlink power is typically larger than uplink power, thus inducing cheaper feedback cost.
\end{remark}

\begin{figure}[t]
\centering
\includegraphics[width=7.9cm,height=4.9cm]{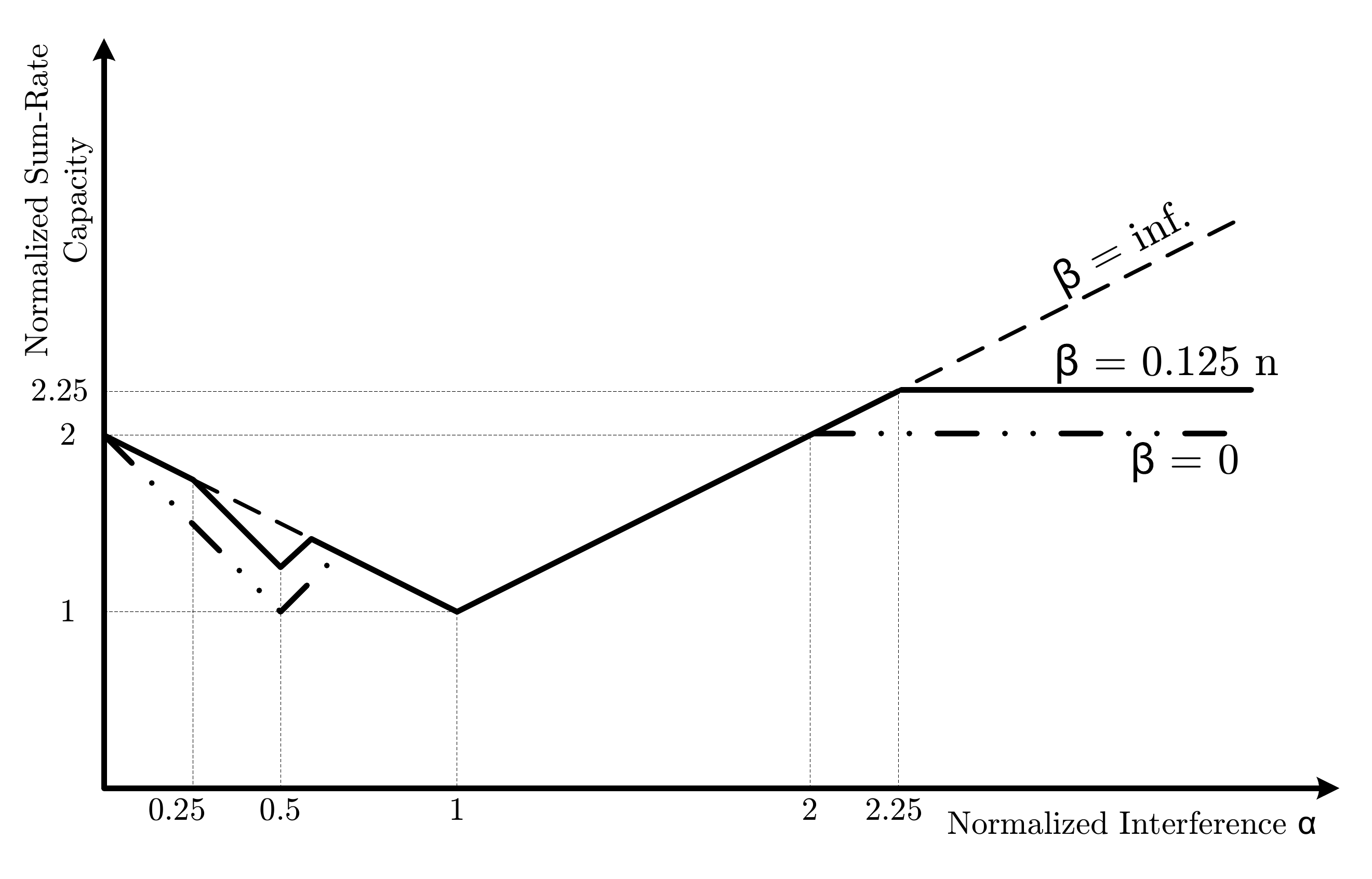}
\caption{Normalized sum-rate capacity for $\beta=0$, $\beta=0.125$ and $\beta=\infty$. \label{fig:ISITfig}}
\end{figure}


\begin{figure*}[h!tb]
\hrule
\begin{subequations}
\begin{eqnarray}
\label{eq:outR1_1}
  R_1 & \leq & \log \left( 1+ \mathsf{SNR}_1  + \mathsf{INR}_{21} + 2 \rho \sqrt{ \mathsf{SNR}_1  \cdot \mathsf{INR}_{21}}
 \right) \\
\label{eq:outR1_2}
 R_1 & \leq & \log \left( 1 + \frac{ (1-\rho^2){\sf SNR}_1 }{ 1+ (1-\rho^2){\sf INR}_{12} } \right) +\log \left( 1+ (1 - \rho^2 ) {\sf INR}_{12} \right)   \\
\label{eq:outR1_3}
 R_1 & \leq & \log \left( 1 +  (1- \rho^2) \mathsf{SNR}_1 \right)
 + C_{\sf FB2 }\\
 \label{eq:outR2_1}
  R_2 & \leq & \log \left( 1+ \mathsf{SNR}_2  + \mathsf{INR}_{12} + 2 \rho \sqrt{ \mathsf{SNR}_2  \cdot \mathsf{INR}_{12}}
 \right) \\
\label{eq:outR2_2}
 R_2 & \leq & \log \left( 1 + \frac{ (1-\rho^2){\sf SNR}_2 }{ 1+ (1-\rho^2){\sf INR}_{21} } \right) +\log \left( 1+ (1 - \rho^2 ) {\sf INR}_{21} \right)   \\
\label{eq:outR2_3}
 R_2 & \leq & \log \left( 1 +  (1- \rho^2) \mathsf{SNR}_2 \right)
 + C_{\sf FB1 }\\
\label{eq:outR12_1}
R_1 + R_2 & \leq & \log \left( 1 + (1- \rho^2)  {\sf INR}_{21} + \frac{ (1- \rho^2){\sf SNR}_1 }{1+ (1- \rho^2) {\sf INR}_{12}} \right) + \log \left( 1 + (1- \rho^2)  {\sf INR}_{12} + \frac{ (1- \rho^2){\sf SNR}_2 }{1+ (1- \rho^2) {\sf INR}_{21}} \right) \nonumber \\
& &\; +~C_{\sf FB1} + C_{\sf FB2} \\
\label{eq:outR12_2}
  R_1 + R_2 & \leq &\log \left( 1 +  \frac{ (1- \rho^2) \mathsf{SNR}_1}{ 1 + (1- \rho^2) \mathsf{INR}_{12}} \right) +  \log \left( 1+ \mathsf{SNR}_2  + \mathsf{INR}_{12} + 2 \rho \sqrt{ \mathsf{SNR}_2  \cdot \mathsf{INR}_{12}}
 \right) \\
R_ 1 + R_2 & \leq & \log \left( 1 +  \frac{ (1- \rho^2) \mathsf{SNR}_2}{ 1 + (1- \rho^2) \mathsf{INR}_{21}} \right) +  \log \left( 1+ \mathsf{SNR}_1  + \mathsf{INR}_{21} + 2 \rho \sqrt{ \mathsf{SNR}_1  \cdot \mathsf{INR}_{21}}
 \right) \\
\label{eq:out2R1R2}
2R_1 + R_2 & \leq & \log \left( 1+ \mathsf{SNR}_1  + \mathsf{INR}_{21} + 2 \rho \sqrt{ \mathsf{SNR}_1  \cdot \mathsf{INR}_{21}}
 \right)  +\log \left( 1 +  \frac{ (1- \rho^2) \mathsf{SNR}_1}{ 1 + (1- \rho^2) \mathsf{INR}_{12}} \right) \\
& & \; + \log \left( 1 + (1- \rho^2)  {\sf INR}_{12} + \frac{ (1- \rho^2){\sf SNR}_2 }{1+ (1- \rho^2) {\sf INR}_{21}} \right) + C_{\sf FB1} + C_{\sf FB2} \\
\label{eq:outR12R2}
R_1 + 2R_2  & \leq  & \log \left( 1+ \mathsf{SNR}_2  + \mathsf{INR}_{12} + 2 \rho \sqrt{ \mathsf{SNR}_2  \cdot \mathsf{INR}_{12}}
 \right)  +  \log \left( 1 +  \frac{ (1- \rho^2) \mathsf{SNR}_2}{ 1 + (1- \rho^2) \mathsf{INR}_{21}} \right) \\
& & \; +  \log \left( 1 + (1- \rho^2)  {\sf INR}_{21} + \frac{ (1- \rho^2){\sf SNR}_1 }{1+ (1- \rho^2) {\sf INR}_{12}} \right) + C_{\sf FB1}+ C_{\sf FB2}
\end{eqnarray}
\end{subequations}
\hrule
\end{figure*} 

\section{Gaussian Interference Channel}
\label{gaussian}
In this section, we consider the Gaussian IC with rate-limited feedback, described in Section~\ref{problem}. We first derive an outer-bound on the capacity region of this network. We then develop an achievability strategy based on the techniques discussed in the previous sections and then show that for symmetric channel gains it achieves a sum-rate within a constant gap to the optimality.

\subsection{Outer-bound}

\begin{theorem}
\label{theorem:Gaussianouterbound}
The capacity region of the Gaussian IC with rate-limited feedback is included in the closure of the set $\mathcal{\bar{C}}$ of $(R_1, R_2)$ satisfying inequalities (\ref{eq:outR1_1})--(\ref{eq:outR12R2}) over $0 \leq \rho \leq 1$.
\end{theorem}

\begin{proof}
By symmetry, it suffices to prove the bounds of~(\ref{eq:outR1_1}), (\ref{eq:outR1_2}), (\ref{eq:outR1_3}), (\ref{eq:outR12_1}), (\ref{eq:outR12_2}) and (\ref{eq:out2R1R2}). The bounds of (\ref{eq:outR1_1}), (\ref{eq:outR1_2}) and (\ref{eq:outR1_3}) are nothing but cutset bounds. The bound of (\ref{eq:outR12_2}) corresponds to the case of infinite feedback rate and was derived in \cite{Suh}. Hence, proving the bounds of (\ref{eq:outR12_1}) and (\ref{eq:out2R1R2}) is the main focus of this proof. We will present the proof of (\ref{eq:outR12_1}) here, and defer the proof for remaining bounds to Appendix~\ref{Appendix:Gaussian}.

\textbf{Proof of~(\ref{eq:outR12_1}):}
The proof idea mostly follows the deterministic case proof of~\ref{eq:outerR12_1}. The only difference in the Gaussian case is that we define a noisy version of $h_{12}X_1^N$ corresponding to $V_1^N$ in the deterministic case: $S_1^N:= h_{12} X_1^N + Z_2^N$. Similarly we define $S_2^N:= h_{21} X_2^N + Z_1^N$ to mimic $V_2^N$. With this, we can now get:
\begin{align}
& N\left( R_1 + R_2 - \epsilon_N \right) \overset{(a)}\leq I\left( W_1; Y_1^N \right) + I\left( W_2; Y_2^N \right) \nonumber \\
& \overset{(b)}= h\left( Y_1^N \right) + h\left( Y_2^N \right) - h\left( S_1^N| W_2\right) - h\left( S_2^N| W_1\right) \nonumber \\
& = I\left( S_1^N; W_2 \right) + I\left( S_2^N; W_1 \right) - h\left( S_1^N| Y_1^N \right) - h\left( S_2^N| Y_2^N \right) \nonumber \\
& + \underbrace{h\left( Y_1^N| S_1^N \right) + h\left( Y_2^N| S_2^N \right)}_{T} \nonumber \\
& \overset{(c)}\leq T + I\left( S_1^N, \tilde{Y}_1^N,W_1 ; W_2 \right) + I\left( S_2^N, \tilde{Y}_2^N,W_2; W_1 \right) \nonumber\\
& - h\left( S_1^N| Y_1^N, W_1, \tilde{Y}_1^N \right) - h\left( S_2^N| Y_2^N, W_2, \tilde{Y}_2^N \right) \nonumber \\
& \overset{(d)}= T + I\left( \tilde{Y}_1^N ; W_2|W_1 \right) + I\left( S_1^N ; W_2|W_1, \tilde{Y}_1^N \right) \nonumber \\
& + I\left( \tilde{Y}_2^N ; W_1|W_2 \right) + I\left( S_2^N ; W_1|W_2, \tilde{Y}_2^N \right) \nonumber \\
& - h\left( Z_1^N| S_1^N, W_2, \tilde{Y}_2^N \right) - h\left( Z_2^N| S_2^N, W_1, \tilde{Y}_1^N \right) \nonumber \\
& \overset{(e)}= \underbrace{T - h\left( Z_1^N \right) - h\left( Z_2^N \right)}_{T^\prime} + I\left( \tilde{Y}_1^N ; W_2|W_1 \right) \nonumber \\
& + I\left( \tilde{Y}_2^N ; W_1|W_2 \right) + I\left( Z_2^N;S_2^N|W_1,\tilde{Y}_1^N \right) \nonumber \\
& + I\left( Z_1^N;S_1^N|W_2,\tilde{Y}_2^N \right) - h\left( Z_1^N| W_1, W_2, \tilde{Y}_2^N \right) \nonumber \\
& + h\left( Z_1^N \right) - h\left( Z_2^N| W_1, W_2, \tilde{Y}_1^N \right) + h\left( Z_2^N \right)  \nonumber \\
& \overset{(f)}= T^\prime + I\left( \tilde{Y}_1^N ; W_2|W_1 \right) + I\left( \tilde{Y}_2^N ; W_1|W_2 \right) \nonumber 
\end{align}
\begin{align}
& + I\left( Z_2^N;S_2^N|W_1,\tilde{Y}_1^N \right) + I\left( Z_1^N;S_1^N|W_2,\tilde{Y}_2^N \right) \nonumber \\
& + I\left( Z_1^N; \tilde{Y}_2^N|W_1, W_2 \right) + I\left( Z_2^N;\tilde{Y}_1^N | W_1, W_2\right) \nonumber \\
& = T^\prime + I\left( \tilde{Y}_1^N ; W_2,Z_2^N|W_1 \right) + I\left( \tilde{Y}_2^N ; W_1,Z_1^N|W_2 \right) \nonumber \\
& + I\left( Z_2^N;S_2^N|W_1,\tilde{Y}_1^N \right) + I\left( Z_1^N;S_1^N|W_2,\tilde{Y}_2^N \right) \nonumber \\
& = T^\prime + I\left( \tilde{Y}_1^N ; W_2|W_1, Z_2^N \right) + I\left( \tilde{Y}_2^N ; W_1|W_2, Z_1^N \right) \nonumber\\
& + I\left( \tilde{Y}_1^N,S_2^N ; Z_2^N|W_1 \right)  + I\left( \tilde{Y}_2^N,S_1^N ; Z_1^N|W_2 \right) \nonumber \\
& \overset{(g)}\leq T^\prime + I\left( \tilde{Y}_1^N ; W_2|W_1, Z_2^N \right) + I\left( \tilde{Y}_2^N ; W_1|W_2, Z_1^N \right) \nonumber \\
& + I\left( Z_2^N ; \tilde{Y}_1^N,W_2,\tilde{Y}_2^N,Z_1^N  |W_1 \right)  \nonumber\\
& + I\left( Z_1^N ; \tilde{Y}_2^N,W_1,\tilde{Y}_1^N,Z_2^N  |W_2 \right) \nonumber \\
& \overset{(h)}= T^\prime + I\left( \tilde{Y}_1^N ; W_2|W_1, Z_2^N \right) + I\left( \tilde{Y}_2^N ; W_1|W_2, Z_1^N \right) \nonumber \\
& + I\left( \tilde{Y}_2^N ; Z_2^N |W_1,W_2,Z_1^N \right)  + I\left( \tilde{Y}_1^N ; Z_1^N |W_1,W_2,Z_2^N \right) \nonumber \\
& + I\left( \tilde{Y}_1^N ; Z_2^N |W_1,W_2,Z_1^N,\tilde{Y}_2^N \right) \nonumber \\
& + I\left( \tilde{Y}_2^N ; Z_1^N |W_1,W_2,Z_2^N,\tilde{Y}_1^N \right) \nonumber \\
& \overset{(i)}= T^\prime + I\left( \tilde{Y}_1^N ; W_2|W_1, Z_2^N \right) + I\left( \tilde{Y}_2^N ; W_1|W_2, Z_1^N \right) \nonumber \\
& + I\left( \tilde{Y}_2^N ; Z_2^N |W_1,W_2,Z_1^N \right)  + I\left( \tilde{Y}_1^N ; Z_1^N |W_1,W_2,Z_2^N \right)  \nonumber \\
& = h\left( Y_1^N| S_1^N \right) - h\left( Z_1^N \right) + h\left( Y_2^N| S_2^N \right) - h\left( Z_2^N \right) \nonumber \\
& + I\left( \tilde{Y}_1^N ; W_2,Z_1^N |W_1, Z_2^N \right) + I\left( \tilde{Y}_2^N ; W_1,Z_2^N|W_2, Z_1^N \right) \nonumber \\
& \leq \sum_{i=1}^N{\left[ h\left( Y_{1i}| S_{1i} \right) - h\left( Z_{1i} \right) \right]} + \sum_{i=1}^N{\left[ h\left( Y_{2i}| S_{2i} \right) - h\left( Z_{2i} \right) \right]} \nonumber\\
& + H\left( \tilde{Y}_1^N|W_1, Z_2^N \right) + H\left( \tilde{Y}_2^N|W_2, Z_1^N \right) \nonumber \\
& \leq \sum_{i=1}^N{\left[ h\left( Y_{1i}| S_{1i} \right) - h\left( Z_{1i} \right) \right]} + \sum_{i=1}^N{\left[ h\left( Y_{2i}| S_{2i} \right) - h\left( Z_{2i} \right) \right]} \nonumber \\
& + \sum_{i=1}^N{H\left( \tilde{Y}_{1i}|X_{1i} \right)} + \sum_{i=1}^N{H\left( \tilde{Y}_{2i}|X_{2i} \right)},
\end{align}
where ($a$) follows from Fano's inequality; ($b$) follows from the fact that $h(Y_1^N|W_1)=h(S_2^N|W_1)$ and $h(Y_2^N|W_2)=h(S_1^N|W_2)$ (see Claim~\ref{claim-4} below); ($c$) follows from the non-negativity of mutual information and the fact that conditioning reduces entropy; ($d$) follows from the fact that $X_k^N$ is a function of $(W_k,\tilde{Y}_k^{N-1})$, and the fact that $W_1$ and $W_2$ are independent; ($e$) follows from the fact that $X_k^N$ is a function of $(W_k,\tilde{Y}_k^{N-1})$; ($f$) follows from the fact that $Z_k^N$ is independent of $W_1$ and $W_2$; ($g$) holds since $S_k^N$ is a function of $(W_k,\tilde{Y}_k^{N-1},Z_{3-k}^N)$; ($h$) holds since $W_1$, $W_2$, $Z_1^N$, and $Z_2^N$ are mutually independent; ($i$) holds since
\begin{align}
& I\left( \tilde{Y}_1^N ; Z_2^N |W_1,W_2,Z_1^N,\tilde{Y}_2^N \right) \nonumber \\
& \quad = \sum_{i=1}^N{I\left( \tilde{Y}_{1i} ; Z_2^N |W_1,W_2,Z_1^N,\tilde{Y}_2^N,\tilde{Y}_1^{i-1} \right)} \nonumber 
\end{align}
\begin{align}
\label{eq:zeromutual}
& \quad = \sum_{i=1}^N{I\left( \tilde{Y}_{1i} ; Z_2^N |W_1,W_2,Z_1^N,\tilde{Y}_2^N,\tilde{Y}_1^{i-1},X_2^N,X_1^i \right)} \nonumber \\
& \quad = \sum_{i=1}^N{I\left( \tilde{Y}_{1i} ; Z_2^N |W_1,W_2,Z_1^N,\tilde{Y}_2^N,\tilde{Y}_1^{i-1},X_2^N,X_1^i,Y_1^i \right)} \nonumber \\
& \quad = \sum_{i=1}^N{I\left( \tilde{Y}_{1i} ; Z_2^N |W_1,W_2,Z_1^N,\tilde{Y}_2^N,\tilde{Y}_1^{i},X_2^N,X_1^i,Y_1^i \right)} \nonumber \\
& \quad = 0.
\end{align}

\begin{claim}
\label{claim-4}
$h(S_1^{N}|W_2) = h(Y_2^{N}|W_2).$
\end{claim}
\begin{proof}
\begin{align*}
\begin{split}
h&(Y_2^{N}|W_2) = \sum h(Y_{2i}|Y_2^{i-1},W_2) \\
&\overset{(a)}{=} \sum h(S_{1i}|Y_2^{i-1},W_2) \\
&\overset{(b)}{=}  \sum h(S_{1i}|Y_2^{i-1},W_2,X_2^{i},S_1^{i-1}) \\
&\overset{(c)}{=}  \sum h(S_{1i}|W_2,S_1^{i-1}) =h(S_1^{N}|W_2),
\end{split}
\end{align*}
where ($a$) follows from the fact that $Y_{2i}$ is a function of $(X_{2i},S_{1i})$ and $X_{2i}$ is a function of $(W_2, Y_2^{i-1})$;
($b$) follows from the fact that $X_2^{i}$ is a function of $(W_2,Y_2^{i-1})$ and $S_1^{i-1}$ is a function of $(Y_2^{i-1},X_2^{i-1})$; ($c$) follows from Claim \ref{claim-3_Gaussian} (see below).
\end{proof}

\begin{claim}
\label{claim-3_Gaussian}
For all $i\geq 1$, $X_1^{i}$ is a function of $(W_1,S_2^{i-1})$ and $X_2^{i}$ is a function of $(W_2,S_1^{i-1})$.
\end{claim}
\begin{proof}
By symmetry, it is enough to prove only one. Notice that $X_2^{i}$ is a function of $(W_2, Y_2^{i-1})$ and $Y_{2}^{i-1}$ is a function of $(X_2^{i-1},S_1^{i-1}$). Hence, $X_2^{i}$ is a function of $(W_2, X_2^{i-1}, S_1^{i-1})$. Iterating the same argument, we conclude that $X_{2}^{i}$ is a function of $(W_2, X_{21}, S_1^{i-1})$. Since $X_{21}$ depends only on $W_2$, we complete the proof.
\end{proof}

From the above, we get
\begin{align}
\begin{split}
R_1 + R_2 & \leq  h(Y_1|S_1 ) - h(Z_1) + h(Y_2|S_2 ) - h(Z_2) \\
& + C_{\sf FB1} + C_{\sf FB2}.
\end{split}
\end{align}
where we have used the fact that $H(\tilde{Y}_{ki}|X_{ki}) \leq C_{\sf FB k}$ and conditioning reduces entropy.


Finally note that for $\rho = \mathbb{E}[X_1X_2^*]$, we have\footnote{$\rho$ captures the power gain that can be achieved by making the transmit signals correlated.}
{\small \begin{align}
\label{eq:akhar1}
h(Y_1|S_1 ) \leq \log 2 \pi e \left( 1 + (1- \rho^2)  {\sf INR}_{21} + \frac{ (1- \rho^2){\sf SNR}_1 }{1+ (1- \rho^2) {\sf INR}_{12}} \right).
\end{align}}
Using (\ref{eq:akhar1}), we get the desired upper bound in~(\ref{eq:outR12_1}).
\end{proof}

If we consider the symmetric channel gains, \emph{i.e.},
\begin{equation}
\begin{split}
& |h_{11}| = |h_{22}| = |h_d|, \\
& |h_{12}| = |h_{21}| = |h_c|,
\end{split}
\end{equation}
and
\begin{equation}
\begin{split}
& {\sf SNR}_1 = {\sf SNR}_2 = {\sf SNR} = |h_d|^2, \\
& {\sf INR}_{12} = {\sf INR}_{21} = {\sf INR} = |h_c|^2,
\end{split}
\end{equation}
we get the following outer-bound result.

\begin{corollary} \label{COL:Gaussian-sumrate}
The sum-rate capacity of the symmetric Gaussian IC with rate-limited feedback is included by the set $\mathcal{\bar{C}_{\sf sym}}$ of $R_1 + R_2$ satisfying
\begin{subequations}
\begin{eqnarray}
\label{eq:sum1}
R_1 + R_2 & \leq & 2 \log \left( 1 + \mathsf{SNR} \right) + C_{\sf FB1 } + C_{\sf FB2 } \\
\label{eq:sum2}
R_1 + R_2 & \leq & \log \left( 1 +  \frac{  \mathsf{SNR}}{ 1 +  \mathsf{INR}} \right) \\
& + &  \log \left( 1+ \mathsf{SNR} + \mathsf{INR} + 2 \sqrt{ \mathsf{SNR}  \cdot \mathsf{INR}} \right) \nonumber \\
\label{eq:sum4}
R_1 + R_2 & \leq & 2 \log \left( 1 + {\sf INR} + \frac{ {\sf SNR} }{1+ {\sf INR}} \right) \\
& + &C_{\sf FB1} + C_{\sf FB2}. \nonumber
\end{eqnarray}
\end{subequations}
\end{corollary}

\begin{proof}
The proof is straight forward and is a direct consequence of the bounds in (\ref{eq:outR1_3}), (\ref{eq:outR2_3}), (\ref{eq:outR12_1}), and (\ref{eq:outR12_2}). For instance, (\ref{eq:sum1}) is derived by combining (\ref{eq:outR1_3}) and (\ref{eq:outR2_3}) for $\rho=0$. Note that $\rho=0$ maximizes (\ref{eq:outR1_3}) and (\ref{eq:outR2_3}).
\end{proof}

\subsection{Acievability Strategy}

We first provide a brief outline of the achievability. Our achievable scheme is  based on block Markov encoding with backward decoding where the scheme is implemented over $B$ blocks. In each block (with the exception of the last two), new messages are transmitted. At the end of a block, each receiver creates a feedback signal and sends it back to its corresponding transmitter. This will provide each transmitter with part of the other user's information that caused interference. Each transmitter encodes this interfering message and transmit it to its receiver during a different block. Through this part of the transmitted signal, receivers will be able to complete the decoding of the previously received messages. During the last two blocks, no new messages will be transmitted and each transmitter provides its receiver with the interfering message coming from the other transmitter. Later, we let $B$ go to infinity to get our desired result.

As we have seen in Section~\ref{motivation}, each receiver may need to decode the superposition of the two codewords (corresponding to the other user's cooperative common message and part of its own private message). In order to accomplish this in the Gaussian case, we employ lattice codes.

\subsubsection{Lattice Coding Preliminaries}

We briefly go over some preliminaries on lattice coding and summerize the results that will be used later. A lattice is a discrete additive subgroup of $\mathbb{R}^n$. The fundamental volume $V_f(\Lambda)$ of a lattice $\Lambda$ is the reciprocal of the number of lattice points per unit volume.

Given integer $p$, denote the set of integers modulo $p$ by $\mathbb{Z}_p$. Let $\mathbb{Z}^n \rightarrow \mathbb{Z}_p^n: v \mapsto \bar{v}$ be the componentwise modulo $p$ operation over integer vectors. Also, let $C$ be a linear $(n,k)$ code over $\mathbb{Z}_p$. The lattice $\Lambda_C$ defined as
\begin{equation}
\Lambda_C = \{ v \in \mathbb{Z}^n : \bar{v} \in C \},
\end{equation}
is generated with respect to the linear code $C$ (see \cite{lattice} for details). In \cite{lattice}, it has been shown that there exists good lattice codes for point-to-point communication channels, \emph{i.e.}, codes that achieve a rate close to the capacity of the channel with arbitrary small decoding error probability. We summarize the result here.

Consider a point-to-point communication scenario over an additive noise channel
\begin{equation}
\label{eq:p2p}
Y = X + Z,
\end{equation}
where $X$ is the transmitted signal with power constraint $P$, $Y$ is the received signal and $Z$ is the additive noise process with zero mean and variance $\sigma^2$.

A set $\mathcal{B}$ of linear codes over $\mathbb{Z}_p$ is called balanced if every nonzero element of $\mathbb{Z}_p^n$ is contained in the same number of codes in $\mathcal{B}$. Define $\mathcal{L}_\mathcal{B}$ as
\begin{equation}
\label{eq:L}
\mathcal{L}_\mathcal{B} = \{ \Lambda_C : C \in \mathcal{B} \}.
\end{equation}

\begin{lemma}[\cite{lattice}]
Consider a point-to-point additive noise channel described in (\ref{eq:p2p}). Let $\mathcal{B}$ be a balanced set of linear $(n, k)$ codes over $\mathbb{Z}_p$. Averaged over all lattices from the set $\mathcal{L}_\mathcal{B}$ defined in (\ref{eq:L}), each scaled by $\gamma > 0$ and with a fundamental volume $V$, we have that for any $\delta > 0$, the
average probability of decoding error is bounded by
\begin{equation}
\label{}
\bar{P}_e < (1+\delta) \frac{n\frac{1}{2}\log \left( 2 \pi e \sigma^2 \right)}{V},
\end{equation}
for sufficiently large $p$ and small $\gamma$ such that $\gamma^n p^{n-k} = V$.
\end{lemma}

See~\cite{lattice} for the proof. The next lemma describes the existence of a good lattice code for a point-to-point AWGN channel.

\begin{lemma}[\cite{lattice}]
\label{lemma:lattice}
Consider a point-to-point additive noise channel described in (\ref{eq:p2p}) such that the transmitter satisfies a power constraint of $P$. Then, we can choose a lattice $\Lambda$ generated using construction A, a shift $s$\footnotemark \footnotetext{Shift $s$ is a vector in $\mathbb{R}^n$ and it is required in order to prove of existence of good lattice codes, see \cite{lattice} for more details.} and a shaping region $S$\footnotemark \footnotetext{We need to consider the intersection of a lattice with some shaping region $S \subset \mathbb{R}^n$ to satisfy the power constraint.} such that the codebook $(\Lambda + s) \cap S$ achieves a rate $R$ with arbitrarily small probability of error if
\begin{equation}
\label{}
R \leq \frac{1}{2} \log \left( \frac{P}{\sigma^2} \right).
\end{equation}
\end{lemma}

In other words, Lemma~\ref{lemma:lattice} describes the existence of a lattice code with sufficient codewords. See~\cite{lattice} for the proof. For a more comprehensive review of lattice codes see~\cite{lattice,poltyrev1994coding,erez2005lattices}.

\begin{remark}
In this paper, we consider complex AWGN channels. Similar to Lemma~\ref{lemma:lattice}, one can show that using lattice codes, a rate of $\log \left( \frac{P}{\sigma^2} \right)$ is achievable in the complex channel setting.
\end{remark}

\subsubsection{Acievability Strategy for $C_{\sf FB1} = C_{\sf FB}$ and $C_{\sf FB2} = 0$}
\label{achievability}

We describe our strategy for the extreme case where $C_{\sf FB1} = C_{\sf FB}$ and $C_{\sf FB2} = 0$ (interchanging user IDs, one can get similar results for $C_{\sf FB2} = C_{\sf FB}$ and $C_{\sf FB1} = 0$). Our strategy for any other feedback configuration will be based on a combination of the strategies for these extreme cases.

{\bf Codebook Generation and Encoding:} The communication strategy consists of $B$ blocks, each of length $N$ channel uses. In block $b$, $b = 1,2,\ldots,B-2$, transmitter $1$ has four messages $W_1^{(1,b)}, W_1^{(2,b)}, W_1^{(3,b)}$ and $W_1^{(4,b)}$, where $W_1^{(i,b)} \in \{ 1, 2, \ldots, 2^{N R_1^{(i)}} \}$. Out of these four messages, $W_1^{(1,b)}, W_1^{(2,b)}$ and $W_1^{(4,b)}$ are new messages and in particular $W_1^{(1,b)}$ and $W_1^{(2,b)}$ form the private message of transmitter $1$ while $W_1^{(4,b)}$ is the non-cooperative message (as it will be clarified shortly for the feedback strategy, the reason for splitting the private message of transmitter $1$ into two parts is that  in order to be able to use lattice codes, we would like the codeword corresponding to the cooperative common message of transmitter $2$ to be received at the same power level as part of the codeword corresponding to the parivate message of transmitter $1$). We will describe $W_1^{(3,b)}$ when we explain the feedback strategy. On the other hand, transmitter $2$ has three new independent messages $W_2^{(1,b)}, W_2^{(2,b)}$ and $W_2^{(3,b)}$, the private, the cooperative common, and the non-cooperative common message of transmitter $2$ respectively.

At transmitter $k$, message $W_k^{(i,b)}$ is mapped to a Gaussian codeword $X_k^{(i,b)}$ picked from a codebook of size $2^{N R_k^{(i)}}$ and any element of this codebook is drawn i.i.d. from $\mathcal{CN}(0,P_k^{(i)})$, $(k,i) \in \{ (1,1), (1,3), (1,4), (2,1) , (2,3) \}$. For notational simplicity, we have removed the superscript $N$.

Message $W_k^{(2,b)}$ is mapped to $X_k^{(2,b)}$ encoded by lattice $\Lambda_k^{(2,b)}$ with shift $s_k^{(2,b)}$ and spherical shaping region $S_k^{(2,b)}$. This gives a codebook of size $2^{N R_k^{(2)}}$ with power constraint of $P_k^{(2)}$, $k = 1,2$. Denote this codebook by  $(\Lambda_k^{(2,b)} + s_k^{(2,b)}) \cap S_k^{(2,b)}$.

Transmitter $k$ will superimpose all of its transmitted signals to create $X_k^{(b)}$, its transmitted signal during block $b$, \emph{i.e.}, $X_1^{(b)} = X_1^{(1,b)} + X_1^{(2,b)} + X_1^{(3,b)} + X_1^{(4,b)}$ and $X_2^{(b)} = X_2^{(1,b)} + X_2^{(2,b)} + X_2^{(3,b)}$.

The power assignments should be such that they are non-negative and satisfy the power constraint at each transmitter:
\begin{equation}
\label{}
\begin{split}
&P_1 = P_1^{(1)} + P_1^{(2)} + P_1^{(3)} + P_1^{(4)} \leq 1, \\
&P_2 = P_2^{(1)} + P_2^{(2)} + P_2^{(3)} \leq 1.
\end{split}
\end{equation}

{\bf Feedback Strategy:} Our feedback strategy is inspired by the motivating example in Section~\ref{motivation}. Remember that in this example, receiver $2$ had to feed back the superposition of the two codewords (corresponding to transmitter 1's cooperative common message and part of its private message). To realize this in the Gaussian case, we incorporate lattice coding with appropriate power assignment as part of our strategy.

We set ${\sf SNR} P_1^{(2)} = {\sf INR} P_2^{(2)}$, so that $X_1^{(2,b)}$ and $X_2^{(2,b)}$ arrive at the same power level at receiver $1$ and therefore $h_{d} X_1^{(2,b)} + h_{c} X_2^{(2,b)}$ is a lattice point. We refer to this lattice index as $I_{\Lambda_{1,2}}^{(b)}$. Receiver $1$ then feeds $\left( I_{\Lambda_{1,2}}^{(b)} \hspace{1mm} {\mathbf{mod}} \hspace{1mm} 2^{N C_{\sf FB}} \right)$ back to transmitter $1$.

Given $\left( I_{\Lambda_{1,2}}^{(b)} \hspace{1mm} {\mathbf{mod}} \hspace{1mm} 2^{N C_{\sf FB}} \right)$, transmitter $1$ removes $h_{d} X_1^{(2,b)}$ and decodes the message index of $W_2^{(2,b)}$. This can be done as long as the total number of lattice points for either of the two aligned messages is less than $2^{ N C_{\sf FB}}$, \emph{i.e.}, $R_1^{(2,b)}, R_2^{(2,b)} \leq C_{\sf FB}$. Since the feedback transmission itself lasts a block, we set $W_1^{(3,b+2)} = W_2^{(2,b)}$.

{\bf Decoding:} For notational simplicity, we ignore the block index and from our description it is clear whether the two signals belong to the same block or different ones. We also use the following shorthand notation:
\begin{equation}
P_k^{(1:j)} = P_k^{(1)} + P_k^{(2)} + \ldots + P_k^{(j)} \hspace{4mm} k = 1, 2.
\end{equation}

Our achievable scheme employs different decoding orders depending on the channel gains. In other words, based on the channel gains the number of required messages to achieve the desired sum-rate might vary. In fact based on the channel gains, it might be sufficient to consider fewer messages than suggested above. In such cases, we assume the unnecessary messages to be deterministic (\emph{i.e.}, the corresponding rate to be zero). In particular, we have three different cases.

\noindent  {\bf Case (a) $\log \left( {\sf INR} \right) \leq \frac{1}{2} \log \left( {\sf SNR} \right)$:}

In this case, we set $R_1^{(4)} = R_2^{(3)} = 0$. In other words, $W_1^{(4)}$ and $W_2^{(3)}$ are deterministic messages. We then get
\begin{equation}
\label{}
Y_1 = h_d \left( X_1^{(1)} + X_1^{(2)} + X_1^{(3)} \right) + h_c \left( X_2^{(1)} + X_2^{(2)} \right) + Z_1.
\end{equation}

At the end of each block, receiver $1$ first decodes $X_1^{(3)}$ by treating all other codewords as noise. $X_1^{(3)}$ can be decoded with small error probability if
\begin{equation}
R_1^{(3)} \leq \log \left( 1 + \frac{ {\sf SNR} P_1^{(3)}}{1 + {\sf INR} P_2 + {\sf SNR} P_1^{(1:2)}} \right).
\end{equation}

It then removes $h_d X_1^{(3)}$ from the received signal and decodes $X_1^{(1)}$ by treating other codewords as noise. $X_1^{(1)}$ is decodable at receiver $1$ with arbitrary small error probability if
\begin{equation}
R_1^{(1)} \leq \log \left( 1 + \frac{ {\sf SNR} P_1^{(1)}}{1 + {\sf INR} P_2 + {\sf SNR} P_1^{(2)} } \right).
\end{equation}

After removing $h_d X_1^{(1)}$, receiver $1$ has access to $h_{d} X_1^{(2)} + h_{c} X_2^{(2)} + h_{c} X_2^{(1)} + Z_1$. Since we have set ${\sf SNR} P_1^{(2)} = {\sf INR} P_2^{(2)}$, $h_{d} X_1^{(2)} + h_{c} X_2^{(2)}$ is a lattice point with some index $I_{\Lambda_{1,2}}^{(b)}$. Receiver $1$ decodes $I_{\Lambda_{1,2}}^{(b)}$ by treating other codewords as noise, and sends back $\left( I_{\Lambda_{1,2}}^{(b)} \hspace{1mm} {\mathbf{mod}} \hspace{1mm} 2^{N C_{\sf FB}} \right)$ to transmitter $1$. From Lemma~\ref{lemma:lattice}, decoding with arbitrary small error probability is feasible if
\begin{equation}
\begin{split}
R_1^{(2)} &\leq \left[ \log \left( \frac{ {\sf SNR} P_1^{(2)}}{1 + {\sf INR} P_2^{(1)} } \right) \right]^+,\\
R_2^{(2)} &\leq \left[ \log \left( \frac{ {\sf INR} P_2^{(2)}}{1 + {\sf INR} P_2^{(1)} } \right) \right]^+,
\end{split}
\end{equation}
Here $[\cdot]^+ = \max \{ \cdot, 0 \}$.

The decoding at receiver $2$ proceeds as follows. At the end of each block, receiver $2$ removes $h_c X_1^{(3)}$ from its received signal. Note that $X_1^{(3)}$ is in fact a function of $W_2^{(2,b-2)}$ and thus it is known to receiver $2$ (assuming successful decoding in the previous blocks). Therefore, after removing $h_c X_1^{(3)}$, we get
\begin{equation}
\label{}
Y_2 = h_d \left( X_2^{(1)} + X_2^{(2)} \right) + h_c \left( X_1^{(1)} + X_1^{(2)} \right) + Z_2.
\end{equation}

Receiver $2$ now decodes $X_2^{(2)}$ and $X_2^{(1)}$ by treating other codewords as noise. This can be done with arbitrary small error probability if
\begin{equation}
\begin{split}
R_2^{(2)} &\leq \left[ \log \left( \frac{ {\sf SNR} P_2^{(2)}}{1 + {\sf SNR} P_2^{(1)} + {\sf INR} P_1^{(1:2)} } \right) \right]^+, \\
R_2^{(1)} &\leq \log \left( 1 + \frac{ {\sf SNR} P_2^{(1)}}{1 + {\sf INR} P_1^{(1:2)} } \right).
\end{split}
\end{equation}

The decoding strategy presented above describes a set of constraints on the rates, which is summarized as follows:
\begin{equation}
\label{}
\left\{ \begin{array}{ll}
R_1^{(1)} &\leq \log \left( 1 + \frac{ {\sf SNR} P_1^{(1)}}{1 + {\sf INR} P_2 + {\sf SNR} P_1^{(2)} } \right) \\
R_1^{(2)} &\leq \min \left\{ \log \left( \frac{ {\sf SNR} P_1^{(2)}}{1 + {\sf INR} P_2^{(1)} } \right)^+, C_{\sf FB} \right\} \\
R_1^{(3)} &\leq \log \left( 1 + \frac{ {\sf SNR} P_1^{(3)}}{1 + {\sf INR} P_2 + {\sf SNR} P_1^{(1:2)}} \right) \\
R_2^{(1)} &\leq \log \left( 1 + \frac{ {\sf SNR} P_2^{(1)}}{1 + {\sf INR} P_1^{(1:2)} } \right) \\
R_2^{(2)} &\leq \min \left\{ \left[\log \left( \frac{ {\sf INR} P_2^{(2)}}{1 + {\sf INR} P_2^{(1)} } \right) \right]^+, C_{\sf FB} \right\}
\end{array} \right.
\end{equation}

Therefore, we can achieve a sum-rate $R_{\sf SUM}^{(a)} = R_1^{(1)} + R_1^{(2)} + R_2^{(1)} + R_2^{(2)}$, arbitrary close to\footnotemark \footnotetext{Note that $X_1^{(3,b)}$ is a function of the cooperative common message of transmitter $2$, \emph{i.e.}, $W_2^{(2,b-2)}$, hence, it does not contain any new information and it is not considered in the sum-rate.}
\begin{align}
\label{eq:rsuma}
R_{\sf SUM}^{(a)} &=  \log \left( 1 + \frac{ {\sf SNR} P_1^{(1)}}{1 + {\sf INR} P_2 + {\sf SNR} P_1^{(2)} } \right) \nonumber \\
&\; + \min \left\{ \left[ \log \left( \frac{ {\sf SNR} P_1^{(2)}}{1 + {\sf INR} P_2^{(1)} } \right)\right]^+, C_{\sf FB} \right\} \nonumber \\
&\; + \log \left( 1 + \frac{ {\sf SNR} P_2^{(1)}}{1 + {\sf INR} P_1^{(1:2)} } \right) \nonumber \\
&\; + \min \left\{ \left[\log \left( \frac{ {\sf INR} P_2^{(2)}}{1 + {\sf INR} P_2^{(1)} } \right)\right]^+, C_{\sf FB} \right\}.
\end{align}

\noindent {\bf Case (b) $\frac{1}{2} \log \left( {\sf SNR} \right) \leq \log \left( {\sf INR} \right) \leq \frac{2}{3} \log \left( {\sf SNR} \right)$:}
In this case, we have
\begin{align}
\label{}
Y_1 &= h_d \left( X_1^{(1)} + X_1^{(2)} + X_1^{(3)} + X_1^{(4)} \right) \nonumber \\
&+ h_c \left( X_2^{(1)} + X_2^{(2)} + X_2^{(3)} \right) + Z_1.
\end{align}
At the end of each block, receiver $1$ first decodes $X_1^{(4)}$ by treating all other codewords as noise, and removes $h_d X_1^{(4)}$ from the received signal. This can be decoded with small error probability if
\begin{equation}
R_1^{(4)} \leq \log \left( 1 + \frac{ {\sf SNR} P_1^{(4)}}{1 + {\sf INR} P_2 + {\sf SNR} P_1^{(1:3)}} \right).
\end{equation}

Next, it decodes $X_1^{(3)}$ by treating other codewords as noise and removes $h_d X_1^{(3)}$ from the received signal. This can be decoded with arbitrary small error probability if
\begin{equation}
R_1^{(3)} \leq \log \left( 1 + \frac{ {\sf SNR} P_1^{(3)}}{1 + {\sf INR} P_2 + {\sf SNR} P_1^{(1:2)} } \right).
\end{equation}

We proceed by decoding the non-cooperative common message of transmitter $2$, \emph{i.e.}, $X_2^{(3)}$ by treating other codewords as noise. This can be decoded with arbitrary small error probability if
\begin{equation}
R_2^{(3)} \leq \log \left( 1 + \frac{ {\sf INR} P_2^{(3)}}{1 + {\sf INR} P_2^{(1:2)} + {\sf SNR} P_1^{(1:2)}} \right).
\end{equation}

It then removes $h_c X_2^{(3)}$ from the received signal, having now access to $h_d X_1^{(2)} + h_c X_2^{(2)} + h_d X_1^{(1)} + h_c X_2^{(1)} + Z_1$. We decode the lattice index of $h_d X_1^{(2)} + h_c X_2^{(2)}$, \emph{i.e.}, $I_{\Lambda_{1,2}}^{(b)}$, by treating other codewords as noise. It then sends back $\left( I_{\Lambda_{1,2}}^{(b)} \hspace{1mm} {\mathbf{mod}} \hspace{1mm} 2^{N C_{\sf FB}} \right)$ to transmitter $1$. From Lemma~\ref{lemma:lattice}, decoding with arbitrary small error probability is feasible if
\begin{equation}
\begin{split}
R_1^{(2)} &\leq \left[ \log \left( \frac{ {\sf SNR} P_1^{(2)}}{1 + {\sf INR} P_2^{(1)} + {\sf SNR} P_1^{(1)} } \right)\right]^+, \\
R_2^{(2)} &\leq \left[ \log \left( \frac{ {\sf INR} P_2^{(2)}}{1 + {\sf INR} P_2^{(1)} + {\sf SNR} P_1^{(1)} } \right)\right]^+.
\end{split}
\end{equation}

After decoding and removing $h_{d} X_1^{(2)} + h_{c} X_2^{(2)}$, receiver $1$ decodes $X_1^{(1)}$. This can be done with arbitrary small error probability if
\begin{equation}
R_1^{(1)} \leq \log \left( 1 + \frac{ {\sf SNR} P_1^{(1)}}{1 + {\sf INR} P_2^{(1)} } \right).
\end{equation}

Similar to the previous case, receiver $2$ removes $X_1^{(3)}$ from its received signal. The decoding at receiver $2$ proceeds as follows. Receiver $2$ decodes $X_2^{(3)}$ by treating other codewords as noise and removes $h_d X_2^{(3)}$ from the received signal. Next, $X_2^{(2)}$, the non-cooperative common message of transmitter $1$, will be decoded while treating other codewords as noise. Receiver $2$ removes $h_d X_2^{(2)}$ from the received signal and then, decodes $X_1^{(4)}$ by treating other codewords as noise. After removing $h_c X_1^{(4)}$, we now decode the private message of transmitter $2$, \emph{i.e.}, $X_2^{(1)}$. This can be done with arbitrary small error probability if
\begin{equation}
\begin{split}
R_2^{(3)} &\leq \log \left( 1 + \frac{ {\sf SNR} P_2^{(3)}}{1 + {\sf SNR} P_2^{(1:2)} + {\sf INR} ( P_1 - P_1^{(3)} ) } \right), \\
R_2^{(2)} &\leq \left[ \log \left( \frac{ {\sf SNR} P_2^{(2)}}{1 + {\sf SNR} P_2^{(1)} + {\sf INR} ( P_1 - P_1^{(3)} ) } \right) \right]^+,\\
R_1^{(4)} &\leq \log \left( 1 + \frac{ {\sf INR} P_1^{(4)}}{1 + {\sf INR} P_1^{(1:2)} + {\sf SNR} P_2^{(1)}} \right), \\
R_2^{(1)} &\leq \log \left( 1 + \frac{ {\sf SNR} P_2^{(1)}}{1 + {\sf INR} P_1^{(1:2)} } \right).
\end{split}
\end{equation}

The decoding strategy presented above describes a set of constraints on the rates, which is summarized as follows:
\begin{equation}
\label{}
\left\{ \begin{array}{ll}
R_1^{(1)} &\leq \log \left( 1 + \frac{ {\sf SNR} P_1^{(1)}}{1 + {\sf INR} P_2^{(1)} } \right) \\
R_1^{(2)} &\leq \min \left\{ \left[\log \left( \frac{ {\sf SNR} P_1^{(2)}}{1 + {\sf INR} P_2^{(1)} + {\sf SNR} P_1^{(1)} } \right)\right]^+, C_{\sf FB} \right\} \\
R_1^{(3)} &\leq \log \left( 1 + \frac{ {\sf SNR} P_1^{(3)}}{1 + {\sf INR} P_2 + {\sf SNR} P_1^{(1:2)} } \right) \\
R_1^{(4)} &\leq \log \left( 1 + \frac{ {\sf INR} P_1^{(4)}}{1 + {\sf INR} P_1^{(1:2)} + {\sf SNR} P_2^{(1)}} \right) \\
R_2^{(1)} &\leq \log \left( 1 + \frac{ {\sf SNR} P_2^{(1)}}{1 + {\sf INR} P_1^{(1:2)} } \right) \\
R_2^{(2)} &\leq \min \left\{ \left[\log \left( \frac{ {\sf INR} P_2^{(2)}}{1 + {\sf INR} P_2^{(1)} + {\sf SNR} P_1^{(1)} } \right)\right]^+, C_{\sf FB} \right\} \\
R_2^{(3)} &\leq \log \left( 1 + \frac{ {\sf INR} P_2^{(3)}}{1 + {\sf INR} P_2^{(1:2)} + {\sf SNR} P_1^{(1:2)} } \right).
\end{array} \right.
\end{equation}

Therefore, we can achieve a sum-rate $R_{\sf SUM}^{(b)} = R_1^{(1)} + R_1^{(2)} + R_1^{(4)} + R_2^{(1)} + R_2^{(2)} + R_2^{(3)}$, arbitrary close to\footnotemark \footnotetext{Note that $X_1^{(3,b)}$ is a function of the cooperative common message of transmitter $2$, \emph{i.e.}, $W_2^{(2,b-2)}$, hence it is not considered in the sum-rate.}
\begin{align}
\label{eq:rsumb}
R_{\sf SUM}^{(b)} &=  \log \left( 1 + \frac{ {\sf SNR} P_1^{(1)}}{1 + {\sf INR} P_2^{(1)} } \right)
\end{align}
\begin{align}
&+ \min \left\{ \left[ \log \left( \frac{ {\sf SNR} P_1^{(2)}}{1 + {\sf INR} P_2^{(1)} + {\sf SNR} P_1^{(1)} } \right) \right]^+, C_{\sf FB} \right\} \nonumber \\
&+ \log \left( 1 + \frac{ {\sf INR} P_1^{(4)}}{1 + {\sf INR} P_1^{(1:2)} + {\sf SNR} P_2^{(1)}} \right) \nonumber \\
&+ \log \left( 1 + \frac{ {\sf SNR} P_2^{(1)}}{1 + {\sf INR} P_1^{(1:2)} } \right) \nonumber \\
&+ \min \left\{ \left[ \log \left( \frac{ {\sf INR} P_2^{(2)}}{1 + {\sf INR} P_2^{(1)} + {\sf SNR} P_1^{(1)} } \right) \right]^+, C_{\sf FB} \right\} \nonumber \\
&+ \log \left( 1 + \frac{ {\sf INR} P_2^{(3)}}{1 + {\sf INR} P_2^{(1:2)} + {\sf SNR} P_1^{(1:2)} } \right). \nonumber
\end{align}


\noindent {\bf Case (c) $2 \log \left( {\sf SNR} \right) \leq \log \left( {\sf INR} \right)$: }

In this case, there is no need to decode the superposition of the two messages. So set $R_1^{(1)}, R_1^{(2)}$ and $R_2^{(1)}$ equal to zero. We then get
\begin{align}
Y_1 &= h_d \left( X_1^{(3)} + X_1^{(4)} \right) + h_c \left( X_2^{(2)} + X_2^{(3)} \right) + Z_1,\\
\label{}
Y_2 &= h_d \left( X_2^{(2)} + X_2^{(3)} \right) + h_c \left( X_1^{(3)} + X_1^{(4)} \right) + Z_2.
\end{align}

As for the feedback strategy, receiver $1$ decodes $X_2^{(2)}$ by treating other codewords as noise, and sends the lattice index of $W_2^{(2)}$ back to transmitter $1$ during the following block. Transmitter $1$ later encodes this message as $X_1^{(3)}$ and transmits it. It is worth mentioning that in this case, it is in fact receiver $2$ who wants to exploit the feedback link of user $1$ to get part of its message. In other words, we have two paths for information flow from transmitter $2$ to receiver $2$; one through the direct link between them and the other one through receiver $1$, feedback link and transmitter $1$. The decoding works very similar to what we described above and we get the following set of constraints to guarantee small error probability at the decoders.
\begin{equation}
\label{ach1}
\left\{ \begin{array}{ll}
R_1^{(3)} &\leq \log \left( 1 + \frac{{\sf INR} P_1^{(3)}}{1+{\sf SNR} P_2} \right) \\
R_1^{(4)} &\leq \log \left( 1 + \frac{{\sf SNR} P_1^{(4)}}{1+{\sf SNR} P_1^{(3)}} \right) \\
R_2^{(2)} &\leq \min \left\{ \left[ \log \left( \frac{{\sf INR} P_2^{(2)}}{1+{\sf SNR} P_1} \right) \right]^+, C_{\sf FB} \right\} \\
R_2^{(3)} &\leq \log \left( 1 + \frac{{\sf SNR} P_2^{(3)}}{1+{\sf SNR} P_2^{(2)}} \right)
\end{array} \right.
\end{equation}

As before $X_1^{(3,b)}$ is a function of $W_2^{(2,b-2)}$. Therefore, we can achieve a sum-rate $R_{\sf SUM}^{(c)} = R_1^{(4)} + R_2^{(2)} + R_2^{(3)}$, arbitrary close to
\begin{align}
\label{eq:rsumc}
R_{\sf SUM}^{(c)} &= \log \left( 1 + \frac{{\sf SNR} P_1^{(4)}}{1+{\sf SNR} P_1^{(3)}} \right) \nonumber \\
&+ \min \left\{ \left[ \log \left( \frac{{\sf INR} P_2^{(2)}}{1+{\sf SNR} P_1} \right) \right]^+, C_{\sf FB} \right\} \nonumber \\
& + \log \left( 1 + \frac{{\sf SNR} P_2^{(3)}}{1+{\sf SNR} P_2^{(2)}} \right).
\end{align}

\noindent {\bf Case (d) $ \frac{2}{3} \log \left( {\sf SNR} \right) \leq \log \left( {\sf INR} \right) \leq 2 \log \left( {\sf SNR} \right)$:}

As we will show in Appendix~\ref{Appendix:Gap}, in this regime feedback can at most increase the sum-rate capacity by 4 bits/sec/Hz. Hence, we ignore the feedback and use the non-feedback transmission strategy in~\cite{Etkin} (\emph{i.e.}, having only one private and one common message at each transmitter and jointly decoding at receivers).

\subsubsection{General Feedback Assignment}

We now describe our achievable scheme for general feedback capacity assignment based on a combination of the achievability schemes for the extreme cases. Let $C_{\sf FB1} = \lambda C_{\sf FB}$ and $C_{\sf FB2} = ( 1 - \lambda) C_{\sf FB}$, such that $0 \leq \lambda \leq 1$. We call the achievable sum-rate of the extreme case $C_{\sf FB1} = C_{\sf FB}$ and $C_{\sf FB2} = 0$ by $R_{\sf SUM}^{C_{\sf FB2} = 0}$, and similarly, we refer to the achievable sum-rate of the other extreme case by $R_{\sf SUM}^{C_{\sf FB1} = 0}$.
\begin{figure}[ht]
\centering
\includegraphics[width=8cm]{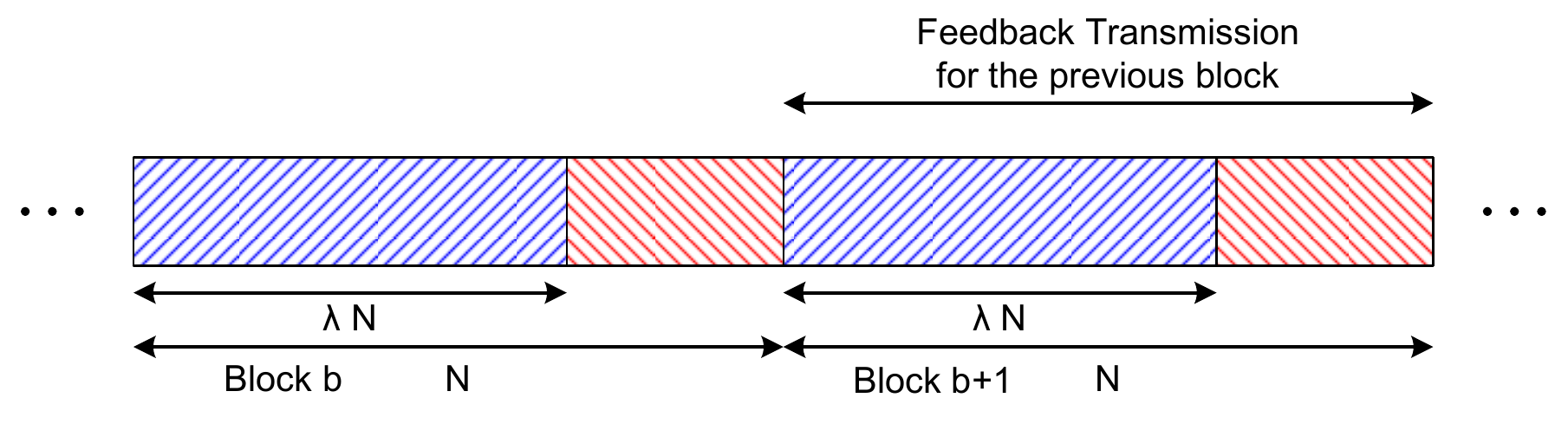}
\caption{Achievability strategy for $C_{\sf FB1} = \lambda C_{\sf FB}$ and $C_{\sf FB2} = ( 1 - \lambda) C_{\sf FB}$.}
\label{feedback12}
\end{figure}
We split any block $b$, $b = 1,2,\ldots,B-2$, of length $N$ into two sub-blocks: $b_1$ of length $\lambda N$ and $b_2$ of length $( 1 - \lambda ) N$. See Figure~\ref{feedback12} for a depiction. During block $b_1$, we implement the transmission strategy of the extreme case $C_{\sf FB1} = C_{\sf FB}$ and $C_{\sf FB2} = 0$, with a block length of $\lambda N$; and during block $b_2$, the achievability scheme of the extreme case $C_{\sf FB1} = 0$ and $C_{\sf FB2} = C_{\sf FB}$, with a block length of $\left( 1 - \lambda \right) N$.

At the end of each sub-block, receivers decode the messages as described before and create the feedback messages. During block $b+1$ the feedback messages of sub-blocks $b_1$ and $b_2$ will be sent back to corresponding transmitters, as shown in Figure~\ref{feedback12}. Note that we use $C_{\sf FB1}$ during the entire length of block $b+1$, hence the effective feedback rate of user $1$ (total feedback use divided by number of transmission time slots), would be
\begin{equation}
C_{\sf FB1}^{\sf eff} = \frac{N C_{\sf FB1}}{\lambda N} = \frac{\lambda N C_{\sf FB}}{\lambda N} = C_{\sf FB}.
\end{equation}

Hence, we can implement the achievability strategy corresponding to the extreme case $C_{\sf FB1} = C_{\sf FB}$ and $C_{\sf FB2} = 0$. Similar argument is valid for the other extreme case. With this achievability scheme, as $N$ goes to infinity, we achieve a sum-rate of $\lambda R_{\sf SUM}^{C_{\sf FB2} = 0} + \left( 1 - \lambda \right) R_{\sf SUM}^{C_{\sf FB1} = 0}$.

\subsubsection{Power Splitting}

We have yet to specify the values of the powers associated with the codewords at the transmitters (\emph{i.e.}, $P_k^{(i)}$: $k\in \{1,2\}$, i $\in \{1,2,3,4\}$). In general, one can solve an optimization problem to find the optimal choice of power level assignments that maximizes the achievable sum-rate. We have performed numerical analysis for this optimization problem. Figure~\ref{Fig:gapnumerical} shows the gap between our proposed achievable scheme and the outer-bounds in Corollary~\ref{COL:Gaussian-sumrate} at (a) ${\sf SNR} = 20 {\sf dB}$, (b) ${\sf SNR} = 40 {\sf dB}$, and (c) ${\sf SNR} = 60 {\sf dB}$, for $C_{FB} = 10$ bits. In fact through our numerical analysis, we can see that the gap is at most $4$, $5$, and $5.5$ bits/sec/Hz for the given values of ${\sf SNR}$, respectively. Note that sharp points in Figure~\ref{Fig:gapnumerical} are due to the change of achievability scheme for different values of ${\sf INR}$ as described before.
\begin{figure}[ht]
\centering
\subfigure[]{\includegraphics[width=7cm]{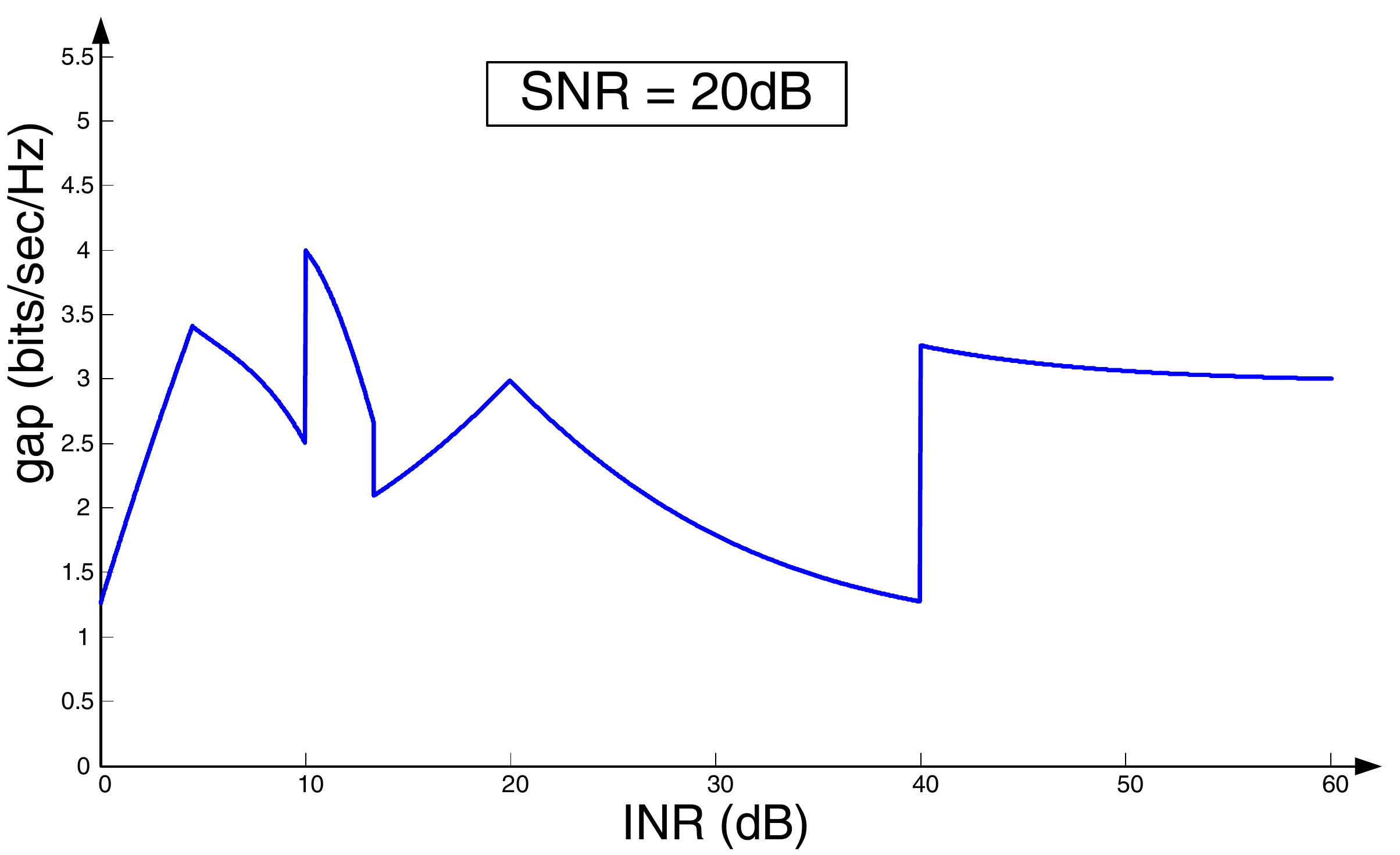}}
\subfigure[]{\includegraphics[width=7cm]{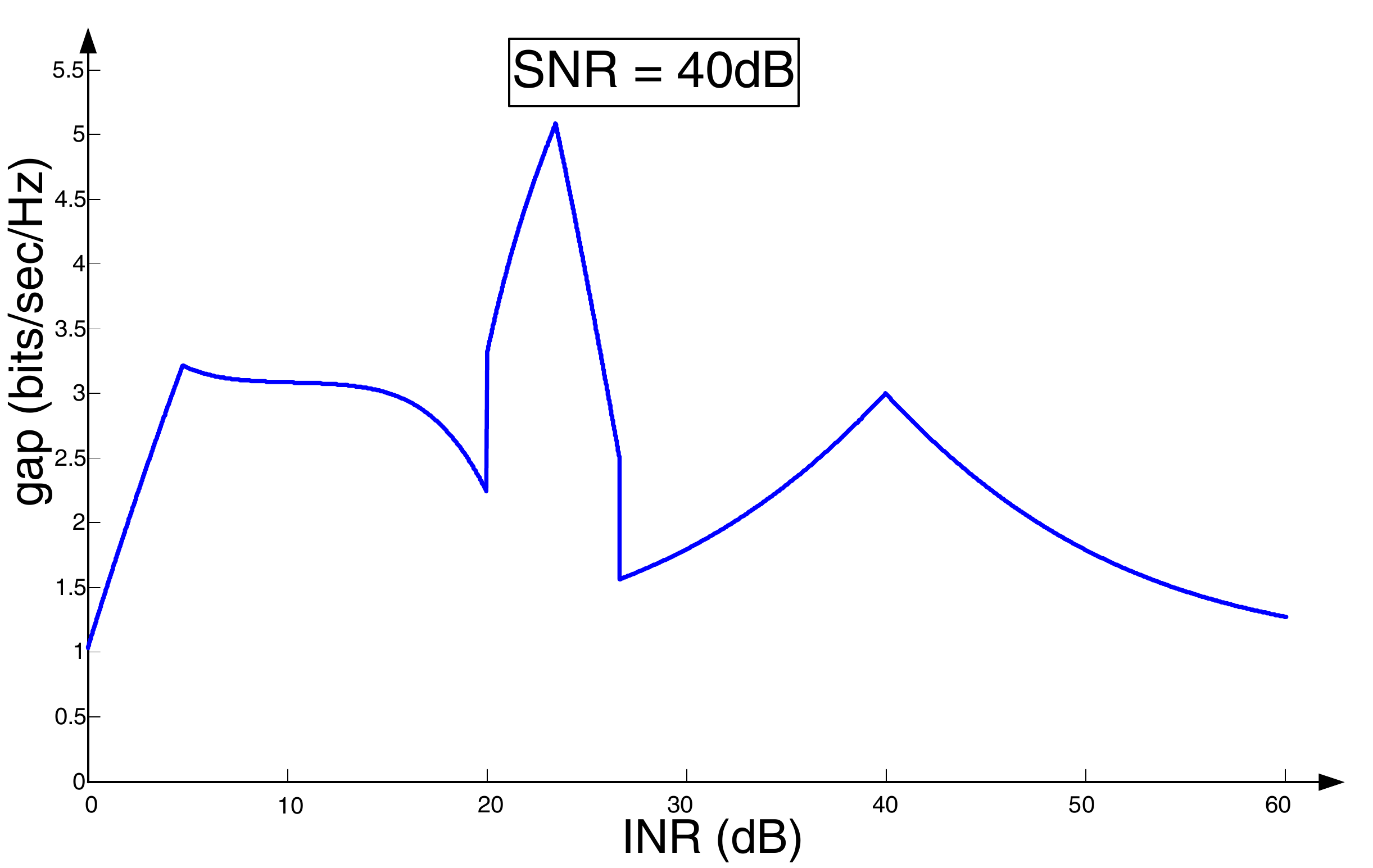}}
\subfigure[]{\includegraphics[width=7cm]{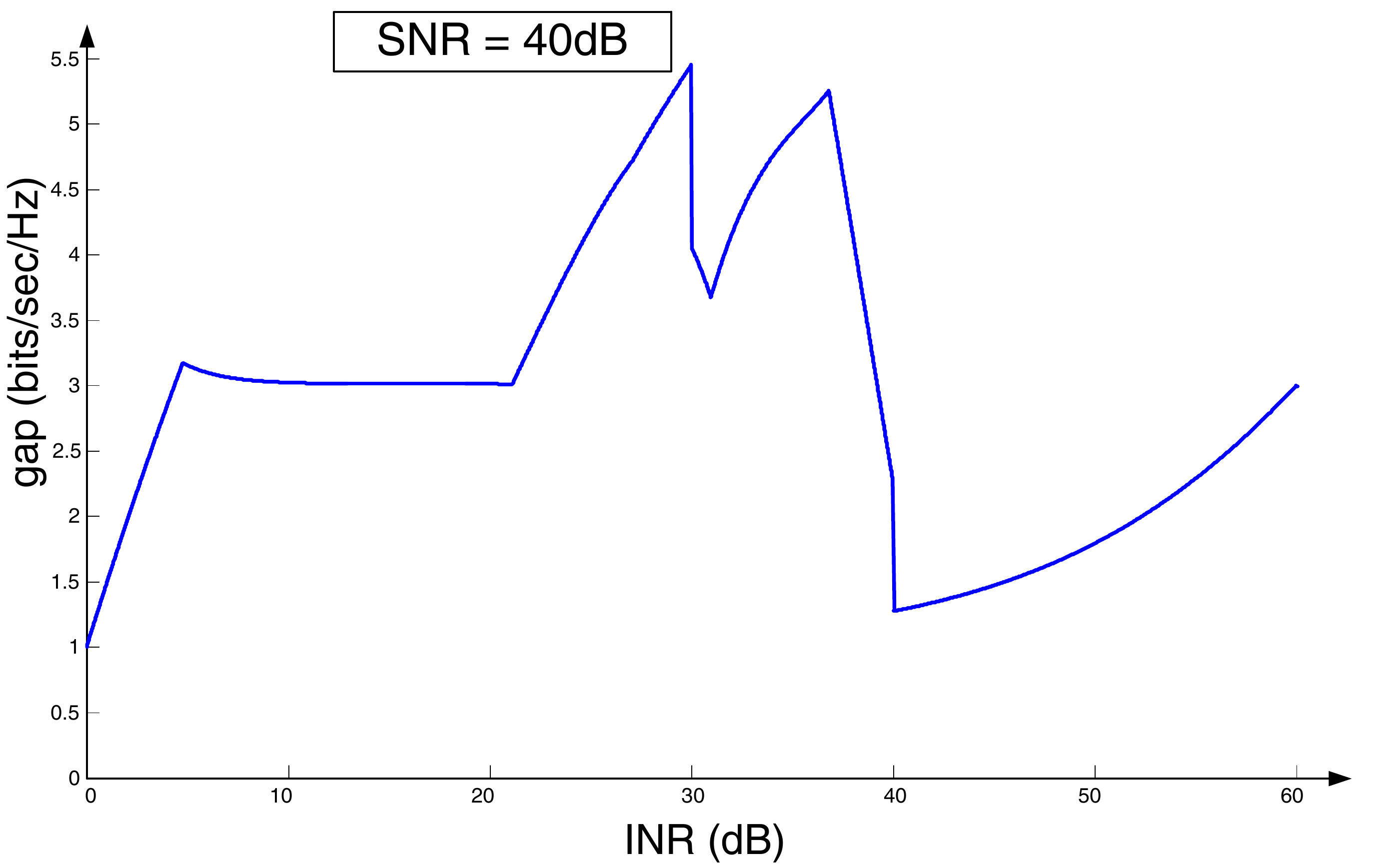}}
\caption{Numerical analysis: gap between achievable scheme and the outer-bounds in Corollary~\ref{COL:Gaussian-sumrate} at (a) ${\sf SNR} = 20 {\sf dB}$, (b) ${\sf SNR} = 40 {\sf dB}$, and (c) ${\sf SNR} = 60 {\sf dB}$ for $C_{FB} = 10$ bits.\label{Fig:gapnumerical}}
\end{figure}


In Appendix~\ref{Appendix:Gap}, we present an explicit choice of power assignments such that the gap between the achievability scheme and the outer-bounds does not scale with ${\sf SNR}$.
As a result, we get the following Theorem.

\begin{theorem} \label{THM:Gaussian-sumrate}
The sum-rate capacity of the Gaussian IC with rate-limited feedback is within at most $14.8$ bits/sec/Hz of the maximum $R_1 + R_2$ satisfying
\begin{subequations}
\begin{eqnarray}
0 \leq R_1 + R_2 & \leq & 2 \log \left( 1 + \mathsf{SNR} \right) + C_{\sf FB1 } + C_{\sf FB2 } \\
0 \leq R_1 + R_2 & \leq & \log \left( 1 +  \frac{  \mathsf{SNR}}{ 1 +  \mathsf{INR}} \right) \\
& + &  \log \left( 1+ \mathsf{SNR} + \mathsf{INR} + 2 \sqrt{ \mathsf{SNR}  \cdot \mathsf{INR}} \right) \nonumber \\
0 \leq R_1 + R_2 & \leq & 2 \log \left( 1 + {\sf INR} + \frac{ {\sf SNR} }{1+ {\sf INR}} \right) \\
& + & C_{\sf FB1} + C_{\sf FB2}. \nonumber
\end{eqnarray}
\end{subequations}
\end{theorem}

\begin{remark}
Note that the given choice of power assignment in Appendix~\ref{Appendix:Gap} is not necessarily optimal, and our analysis is pessimistic in the sense that we consider the worst case scenario, and we calculate the gap for the worst case.
\end{remark}

As a corollary, we characterize the symmetric capacity of the two-user Gaussian IC with rate-limited feedback, as defined below, to within a constant number of bits.

\begin{definition}
The symmetric capacity is defined by
\begin{equation}\
\label{csym}
\mathcal{C}_{\sf sym} = \sup \{ R: (R,R) \in \mathcal{C} \},
\end{equation}
where $\mathcal{C}$ is the capacity region.
\end{definition}

\begin{corollary}
For the symmetric Gaussian IC with equal feedback link capacities, \emph{i.e.}, $C_{\sf FB1} = C_{\sf FB2}$, the  presented achievability strategy achieves to within at most $7.4$ bits/sec/Hz/user to the symmetric capacity $\mathcal{C}_{\sf sym}$ defined in (\ref{csym}), for all channel gains.
\end{corollary}

\begin{proof}
Theorem \ref{THM:Gaussian-sumrate} says that we can achieve to within at most $14.8$ bits/sec/Hz of the outerbounds in Corollary~\ref{COL:Gaussian-sumrate} for any feedback assignment. Therefore, in symmetric IC with equal feedback link capacities $C_{\sf FB1} = C_{\sf FB2} = \frac{1}{2} C_{\sf FB1}$, the gap between the achievability and the symmetric capacity is at most $7.4$ bits/sec/Hz/user.
\end{proof} 

\section{Concluding Remarks}
\label{conclusion}
We have addressed the two-user interference channel with rate-limited feedback under three different models: the El Gamal-Costa deterministic model~\cite{ElGamal:it82}, the linear deterministic model~\cite{ADT10}, and the Gaussian model. We developed new achievable schemes and new outer-bounds for all of the three models. We showed the optimality of our scheme under the linear deterministic model. Under the Gaussian model, we established new outer-bounds on the capacity region with rate-limited feedback, and we proposed a transmission strategy employing lattice codes and the ideas developed in the first two models. Furthermore, we proved that the gap between the achievable sum-rate of the proposed scheme and the outer-bound is bounded by a constant number of bits, independent of the channel gains.

One of the future directions would be to extend this result to the capacity region of the asymmetric two-user Gaussian interference channel with rate-limited feedback. The same achievability scheme can be applied there, however, the gap analysis will be cumbersome. Therefore, one interesting direction is to find out new techniques to bound the gap between the achievable region and the outer-bounds on the capacity region of the asymmetric two-user Gaussian interference channel with rate-limited feedback.

\appendices

\section{Proof of Theorem~\ref{theorem:outerboundregion}}
\label{deterministic-proofs}
\textbf{Proof of (\ref{eq:outerR1_1}) (cutset bound):} Starting with Fano's inequality, we get
\begin{align*}
\begin{split}
N( R_1 - \epsilon_N) \leq I(W_1;Y_1^{N}) \leq \sum  H(Y_{1i}),\\
\end{split}
\end{align*}
where the second inequality follows from the fact that conditioning reduces entropy.
If $(R_1,R_2)$ is achievable, then $\epsilon_N \rightarrow 0$ as $N \rightarrow \infty$. Thus we obtain the left term of the bound. Notice that this is a cutset bound, as the bound is obtained assuming that the two transmitters fully collaborate.

To obtain the right term, we consider
\begin{align*}
\begin{split}
N&(R_1 - \epsilon_N) \leq I(W_1;Y_1^N, Y_2^N, W_2) \\
&\overset{(a)}= \sum H(Y_{1i},Y_{2i}|W_2,Y_1^{i-1},Y_2^{i-1},X_2^{i}) \\
&= \sum H(Y_{1i}|W_2,Y_1^{i-1},Y_2^{i}, X_{2}^{i}) \\
&+ \sum H(Y_{2i}|W_2,Y_1^{i-1},Y_2^{i-1},X_{2}^{i}) \\
&\overset{(b)}{=}  \sum H(Y_{1i}|W_2,Y_{1}^{i-1},Y_{2i},X_{2i},U_{1i}) \\
&+ \sum H(Y_{2i}|W_2, Y_{1}^{i-1}, X_{2i}, U_{1i})  \\
&\overset{(c)}{\leq} \sum H(Y_{1i}|X_{2i}, Y_{2i}, U_{1i}) + \sum H(Y_{2i}|X_{2i},U_{1i}) \\
&\overset{(d)}{\leq} \sum H(Y_{1i}|V_{1i}, V_{2i}, U_{1i}) + \sum H(Y_{2i}|X_{2i},U_{1i}),
\end{split}
\end{align*}
where ($a$) follows from the fact that $W_1$ is independent from $W_2$, and $X_{2}^{i}$ is a function of $(W_2,Y_2^{i-1})$; ($b$) follows from the fact that $U_{1i}:=(X_{2}^{i-1},\tilde{Y}_2^{i-1})$ and $\tilde{Y}_2^{i-1}$ is a function of $Y_2^{i-1}$; ($c$) follows from the fact that conditioning reduces entropy; ($d$) follows from the fact that $(V_{1i}$ is a function of $(X_{2i}, Y_{2i})$, $Y_{2i}$ is a function of $(X_{2i}, V_{1i})$ and $V_{2i}$ is a function of $X_{2i}$. Thus we get the right term of the bound. Notice that this is a cutset bound, as the bound is obtained assuming that the two receivers fully collaborate.

\textbf{Proof of (\ref{eq:outerR1_2}) (cutset bound):}
Starting with Fano's inequality, we get:
\begin{align*}
\begin{split}
N&(R_1 - \epsilon_N) \leq I(W_1;Y_1^N, \tilde{Y}_2^N, W_2) \\
&\overset{(a)}= \sum H(Y_{1i},\tilde{Y}_{2i}|W_2,Y_1^{i-1},\tilde{Y}_2^{i-1},X_2^{i}) \\
&= \sum H(Y_{1i}|W_2,Y_1^{i-1},\tilde{Y}_2^{i-1}, X_{2}^{i}) \\
&+ \sum H(\tilde{Y}_{2i}|W_2,Y_1^{i},\tilde{Y}_2^{i-1},X_{2}^{i}) \\
&\overset{(b)}{\leq} \sum H(Y_{1i}|X_{2i}, U_{1i}) + \sum H(\tilde{Y}_{2i}|Y_{1i},X_{2i}) \\
&\overset{(c)}{\leq} \sum H(Y_{1i}|X_{2i}, U_{1i}) + N C_{\sf FB2},
\end{split}
\end{align*}
where ($a$) follows from the fact that $W_1$ is independent from $W_2$, and $X_{2}^{i}$ is a function of $(W_2,\tilde{Y}_2^{i-1})$; ($b$) follows from the fact that conditioning reduces entropy; ($d$) follows from $H(\tilde{Y}_{2i}| Y_{1i}, X_{2i}) \leq C_{\sf FB 2}$. Therefore, we get the desired bound.

\textbf{Proof of (\ref{eq:outerR12_2}):}
Starting with Fano's inequality, we get
\begin{align*}
\begin{split}
N&(R_1 + R_2- \epsilon_N) \leq I(W_1;Y_1^{N}|W_2) + I(W_2;Y_2^{N})  \\
&= H(Y_1^{N}|W_2) + I(W_2;Y_2^{N})  \\
&= H(Y_1^{N}|W_2) + H(Y_2^{N}) \\
&- \left\{ H(Y_1^{N}, Y_2^{N}|W_2) - H(Y_1^{N}|Y_2^{N},W_2) \right\} \\
& = H(Y_1^{N}|Y_2^{N},W_2) - H(Y_2^{N}|Y_1^{N},W_2) + H(Y_2^N)  \\
&\overset{(a)}{=} \sum H(Y_{1i} | W_2, Y_1^{i-1}, Y_2^{N}, X_2^{i},V_{1i}) + H(Y_2^N)  \\
&\overset{(b)}{\leq} \sum \left[ H(Y_{1i}|V_{2i}, V_{1i}, U_{1i} ) + H(Y_{2i}) \right],
\end{split}
\end{align*}
where ($a$) follows from the fact that $X_{2}^{i}$ is a function of $(W_2, Y_2^{i-1})$ and $V_{1i}$ is a function of $(X_{2i},Y_{2i})$; ($b$) follows from the fact that $U_{1i}$ is a function of $(X_2^{i-1}, Y_2^{i-1})$ and conditioning reduces entropy.

\textbf{Proof of (\ref{eq:outer2R1R2_1}):}
\begin{align*}
\begin{split}
N&(2R_1 + R_2- \epsilon_N) \\
&\leq I(W_1;Y_1^{N}) + I(W_1;Y_1^{N}|W_2)  + I(W_2;Y_2^{N})  \\
&\overset{(a)}{\leq} [H(Y_1^{N}) - H(Y_1^N|W_1)]  \\
&+ I(W_1;Y_1^{N}, V_1^{N}|W_2) + [H(Y_2^{N}) -H(Y_2^{N}|W_2)] \\
&\overset{(b)}{=}   [H(Y_1^{N}) - H(V_2^N|W_1)]  + H(V_1^{N}|W_2) \\
&+ H(Y_1^N|W_2, V_1^N) + [H(Y_2^{N}) -H(V_1^{N}|W_2)] \\
&=   H(Y_1^{N}) +  H(Y_1^N|W_2, V_1^N) \\
&+ H(Y_2^{N}) - [H(V_2^N) - I(W_1; V_2^N) ]\\
& \overset{(c)}{\leq}  I(W_1; V_2^N) +  H(Y_1^{N}) +  H(Y_1^N|W_2, V_1^N) \\
&+ H(Y_2^{N}, V_2^N) - H(V_2^N) \\
&  \overset{(d)}{\leq}   I(W_1; V_2^N, W_2, \tilde{Y}_2^N) +  H(Y_1^{N}) \\
&+  H(Y_1^N|W_2, V_1^N) + H(Y_2^{N}|V_2^N) \\
& \overset{(e)}{=}    I(W_1; \tilde{Y}_2^N |W_2) +  H(Y_1^{N}) \\
&+  H(Y_1^N|W_2, V_1^N,X_2^N,\tilde{Y}_2^N) + H(Y_2^{N}|V_2^N) \\
&\overset{(f)}{\leq} N C_{\sf FB2} + \sum H(Y_{1i}) + \sum H(Y_{2i} | V_{2i} ) \\
&+ \sum H(Y_{1i} |V_{1i}, V_{2i}, U_{1i} ),
\end{split}
\end{align*}
where ($a$) follows from the fact that adding information increases mutual information; ($b$) follows from Claim~\ref{claim1}; ($c$) follows from providing $V_2^N$ to receiver 2; ($d$) follows from the fact that adding information increases mutual information;  follows from the fact that $V_k^N$ is a function of $(W_k,\tilde{Y}_k^{N-1})$; ($e$) follows from the fact that $X_2^N$ is a function of $(W_2,V_1^{N-1})$ (by Claim~\ref{claim2}) and $\tilde{Y}_2^N$ is a function of $(X_2^N,V_1^N)$; ($f$) follows from the fact that $U_{1i}:=(X_{2}^{i-1},\tilde{Y}_2^{i-1})$, $U_{2i}:=(X_{1}^{i-1},\tilde{Y}_1^{i-1})$, $H(\tilde{Y}_2^N| W_2) \leq NC_{\sf FB 2}$ and conditioning reduces entropy.

To complete the proof, we will show that given $U_i:=(U_{1i},U_{2i})$, $X_{1i}$ and $X_{2i}$ are conditionally independent. Remember that our input distribution is of the form of $p(u_1,u_2) p(x_1|u_1, u_2) p(x_2|u_1, u_2)$.
\begin{claim}
\label{claim3}
Given $U_i:=(U_{1i},U_{2i})=(X_{2}^{i-1},\tilde{Y}_{2}^{i-1},X_{1}^{i-1},\tilde{Y}_{1}^{i-1} )$, $X_{1i}$ and $X_{2i}$ are conditionally independent.
\end{claim}
\begin{proof}
The proof is based on the dependence-balance-bound technique~\cite{Willems:it82, Willems:it89}.
For completeness we describe details. We first show that $I(W_1;W_2|U_i)=0$,  implying that $W_1$ and $W_2$ are independent given $U_i$. We will then show that $X_{1i}$ and $X_{2i}$ are conditionally independent given $U_i$.

Consider
\begin{align*}
0 &\leq I(W_1;W_2|U_i) \overset{(a)}{=} I(W_1;W_2|U_i) - I(W_1;W_2) \\
& \overset{(b)}{=} - H(W_1) - H(W_2) - H(U_i) + H(W_1, W_2) \\
&+ H(W_1,U_i) + H(W_2,U_i) - H(W_1,W_2, U_i) \\
& \overset{(c)}{=} - H(U_i) + H(U_i|W_1) + H(U_i|W_2) \\
& \overset{(d)}{=} \sum_{j=1}^{i-1} \left[ - H(X_{1j},X_{2j} |X_1^{j-1}, X_2^{j-1}) \right. \\
&\left. \qquad + H(X_{1j},X_{2j}|W_1, X_1^{j-1}, X_2^{j-1}) \right. \\
&\left. \qquad +   H(X_{1j},X_{2j} |W_2, X_1^{j-1}, X_2^{j-1}) \right] \\
& \overset{(e)}{=} \sum_{j=1}^{i-1} \left[ - H(X_{1j},X_{2j} |X_1^{j-1}, X_2^{j-1})  \right. \\
&\left. \qquad + H(X_{2j}|W_1, X_1^{j}, X_2^{j-1}) + H(X_{1j}|W_2, X_1^{j-1}, X_2^{j} ) \right] \\
& = \sum_{j=1}^{i-1} \left[ - H(X_{1j} |X_1^{j-1}, X_2^{j-1} ) + H(X_{1j} |W_2, X_1^{j-1}, X_2^{j}) \right. \\&\left. \qquad - H(X_{2j}|X_1^{j}, X_2^{j-1}) + H(X_{2j}|W_1, X_1^{j}, X_2^{j-1}) \right]  \overset{(f)} \leq 0,
\end{align*}
where ($a$) follows from $I(W_1;W_2)=0$; ($b$) follows from the chain rule; ($c$) follows from the chain rule and $H(U_i|W_1,W_2)=0$; ($d$) follows from the fact that $(\tilde{Y}_{1j},\tilde{Y}_{2j})$ is a function of $(\tilde{X}_{1j},\tilde{X}_{2j})$;
($e$) follows from the fact that $X_{kj}$ is a function of $(W_k,X_1^{j-1}, X_2^{j-1})$; ($f$) follows from the fact that conditioning reduces entropy.
Therefore, $I(W_1;W_2|U_i)=0$, which shows the independence of $W_1$ and $W_2$ given $U_i$. 

Notice that $X_{ki}$ is a function of $(W_k,X_1^{i-1}, X_2^{i-1})$. Hence, it easily follows that $
I(X_{1i};X_{2i}|U_i) = I(X_{1i};X_{2i}|X_{1}^{i-1}, X_2^{i-1} ) = 0$. This proves the independence of $X_{1i}$ and $X_{2i}$ given $U_i$.
\end{proof} 

\section{Achievability Proof of Theorem~\ref{theorem:DICcapacity}}
\label{Appendix:DICcapacity}
With the choice of distribution given in (\ref{choice}), we have
\begin{subequations}
\begin{eqnarray}
& &\delta_1 = I(\hat{Y}_1; Y_1|U, U_2, X_1) = 0, \\
& & \delta_2 = I(\hat{Y}_2; Y_2|U, U_1, X_2) = 0, \\
\label{eq:Th1_element1}
& &I(U,V_2, X_1;Y_1) = \max (n_{11}, n_{21}), \\
& &I(U,V_1, X_2;Y_2) = \max (n_{22}, n_{12}), \\
& &I(X_1;Y_1|U,V_1, V_2) = (n_{11} - n_{12} )^+, \\
& &I(X_2;Y_2|U,V_1, V_2) = (n_{22} - n_{21} )^+, \\
\label{eq:Th1_element3}
& &I(U_2; Y_1|U, X_1) = \min( n_{21}, C_{\sf FB1}), \\
& &I(U_1; Y_2|U, X_2) = \min( n_{12}, C_{\sf FB2}), \\
\label{eq:Th1_element4}
& &I(X_1;Y_1|U, U_1, V_2) = (n_{11} - n_{12})^+  \nonumber \\
& &+ \min  \left \{  n_{11}, (n_{12} - C_{\sf FB2})^+ \right \},   \\
& &I(X_2;Y_2|U, U_2, V_1) = (n_{22} - n_{21})^+  \nonumber \\
& &+ \min  \left \{  n_{22}, (n_{21} - C_{\sf FB1})^+ \right \}, \\
\label{eq:Th1_element5}
& &I(X_1,V_2;Y_1|U, V_1, U_2)  =  \\
& &(n_{21} - C_{\sf FB1})^+ + \left[ (n_{11}- n_{12} )^+ -  n_{21} \right]^+  \nonumber \\
& &\;\;+ \min \left\{ (n_{11}- n_{12})^+, \min ( n_{21}, C_{\sf FB1}) \right \}, \nonumber \\
& &I(X_2,V_1;Y_2|U, V_2, U_1)  =  \\
& &(n_{12} - C_{\sf FB2})^+ + \left[ (n_{22}- n_{21} )^+ -  n_{12} \right]^+  \nonumber \\
& &\;\;+ \min \left\{ (n_{22}- n_{21})^+, \min ( n_{12}, C_{\sf FB2}) \right \}, \nonumber \\
\label{eq:Th1_element6}
& &I(X_1,V_2;Y_1|U, U_1, U_2) = \\
& &(n_{11} - n_{12} )^+ + \min \left \{ n_{11}, (n_{12} - C_{\sf FB2})^+ \right \}  \nonumber \\
& &\;\;+ \left[ n_{21} - \max \left\{ n_{11}, \min (n_{21}, C_{\sf FB1}) \right\} \right]^+ \nonumber \\
& &\;\;+ \left[  \min \left \{ n_{21}, \left[ n_{11} - (n_{12} - C_{\sf FB2})^+ \right]^+ \right \} \right. \nonumber \\
& &\left. - \max \left\{ \min (n_{21}, C_{\sf FB1}), (n_{11} - n_{12})^+ \right \} \right]^+, \nonumber \\
& &I(X_2,V_1;Y_2|U, U_1, U_2) = \\
& &(n_{22} - n_{21} )^+ + \min \left \{ n_{22}, (n_{21} - C_{\sf FB1})^+ \right \}  \nonumber \\
& &\;\;+ \left[ n_{12} - \max \left\{ n_{22}, \min (n_{12}, C_{\sf FB2}) \right\} \right]^+ \nonumber \\
& &\;\;+ \left[  \min \left \{ n_{12}, \left[ n_{22} - (n_{21} - C_{\sf FB1})^+ \right]^+ \right \} \right. \nonumber \\
& &\left. - \max \left\{ \min (n_{12}, C_{\sf FB2}), (n_{22} - n_{21})^+ \right \} \right]^+. \nonumber
\end{eqnarray}
\end{subequations}
Using this computation, one can show that the inequalities of (\ref{lemmaeq:R1R2-dummy1}) and (\ref{lemmaeq:R1R2-dummy2}) are implied by (\ref{lemmaeq:R1-2}), (\ref{lemmaeq:R2-2}), (\ref{lemmaeq:R1R2-1}) and (\ref{lemmaeq:R1R2-2}); the inequality (\ref{lemmaeq:2R1R2-dummy}) is implied by (\ref{lemmaeq:R1-2}), (\ref{lemmaeq:R2-2}) and (\ref{lemmaeq:2R1R2}); and the inequality (\ref{lemmaeq:R1-2R2-dummy}) is implied by (\ref{lemmaeq:R1-2}), (\ref{lemmaeq:R2-2}) and (\ref{lemmaeq:R1-2R2}). We omit the tedious calculation. With further computation, we get:
\begin{subequations}
\begin{eqnarray}
\label{eq2:R1bound-1}
  R_1 & \leq &  \max( n_{11}, n_{21}) \\
\label{eq2:R1bound-2}
  R_1 & \leq & (n_{11} -n_{12})^+ + \min \left\{ n_{11}, (n_{12}-C_{\sf FB2})^+ \right\} \nonumber \\
  &+& \min (n_{12}, C_{\sf FB 2} )   \\
\label{eq2:R2bound-1}
    R_2 & \leq &  \max( n_{22}, n_{12}) \\
\label{eq2:R2bound-2}
  R_2 & \leq & (n_{22} -n_{21})^+ + \min \left\{ n_{22}, (n_{21}-C_{\sf FB1})^+ \right\} \nonumber \\
  &+& \min (n_{21}, C_{\sf FB 1} ) \\
\label{eq2:R1R2bound-1}
  R_1 + R_2 &\leq&  (n_{11} - n_{12})^+  +  \max( n_{22}, n_{12})  \\
\label{eq2:R1R2bound-2}
  R_1 + R_2 & \leq &(n_{22} - n_{21})^+  +  \max( n_{11}, n_{21}) 
\end{eqnarray}
\begin{eqnarray}
\label{eq2:R1R2bound-3}
R_1 + R_2 & \leq & \max \left \{ (n_{11}-n_{12})^+,  n_{21} \right \} \\
&+&  \max \left \{ (n_{22}-n_{21})^+, n_{12} \right \} \nonumber \\
&+& \min \left\{ (n_{11} - n_{12})^+, n_{21}, C_{\sf FB1} \right \} \nonumber \\
&+& \min \left\{ (n_{22} - n_{21})^+, n_{12}, C_{\sf FB2} \right \} \nonumber \\
\label{eq2:2R1R2bound}
2R_1 + R_2 & \leq & (n_{11} - n_{12})^+  +  \max( n_{11}, n_{21}) \\
&+&  \max \left \{ (n_{22}-n_{21})^+,  n_{12}  \right \} \nonumber \\
&+& \min \left\{ (n_{22} - n_{21})^+, n_{12}, C_{\sf FB2} \right \} \nonumber \\
\label{eq2:R12R2bound}
R_1 + 2R_2  & \leq & (n_{22} - n_{21})^+  +  \max( n_{22}, n_{12}) \\
&+& \max \left\{ (n_{11}-n_{12})^+,  n_{21} \right\} \nonumber \\
&+& \min \left\{ (n_{11} - n_{12})^+, n_{21}, C_{\sf FB1} \right \}. \nonumber
\end{eqnarray}
\end{subequations}

We will show that the inequalities developed above are equivalent to the capacity region in Theorem~\ref{theorem:DICcapacity}. Note that (\ref{eq2:R1bound-2}) can be written as
\begin{align*}
R_1 \leq \left\{
  \begin{array}{ll}
n_{11} + C_{\sf FB2}, & \hbox{ $n_{11} + C_{\sf FB2} \leq n_{12}$;} \\
    \max( n_{11}, n_{12}),
& \hbox{otherwise.}
  \end{array}
\right.
\end{align*}
This shows that this inequality is implied by (\ref{eq:R1bound-1}) and (\ref{eq:R1bound-2}). Similarly, (\ref{eq2:R2bound-2}) is implied by (\ref{eq:R2bound-1}) and (\ref{eq:R2bound-2}).
Next consider (\ref{eq2:R1R2bound-3}),

\noindent $\bullet$ if $C_{\sf FB1} \leq \min \left\{ (n_{11} - n_{12})^+, n_{21} \right\}$, \\ $C_{\sf FB2} \leq \min \left\{ (n_{22} - n_{21})^+, n_{12} \right\}$:
\begin{align}
R_1 + R_2 &\leq \max \left \{ (n_{11}-n_{12})^+,  n_{21} \right \}  \\
&+  \max \left \{ (n_{22}-n_{21})^+, n_{12} \right \}  + C_{\sf FB1} + C_{\sf FB2}, \nonumber
\end{align}
\noindent $\bullet$ if $C_{\sf FB1} > \min \left\{ (n_{11} - n_{12})^+, n_{21} \right\}$, \\ $C_{\sf FB2} > \min \left\{ (n_{22} - n_{21})^+, n_{12} \right\}$:
\begin{align}
R_1 + R_2 \leq \max( n_{11}, n_{12}) + \max( n_{21}, n_{22}),
\end{align}
\noindent $\bullet$ if $C_{\sf FB1} \leq \min \left\{ (n_{11} - n_{12})^+, n_{21} \right\}$, \\ $C_{\sf FB2} > \min \left\{ (n_{22} - n_{21})^+, n_{12} \right\}$:
\begin{align}
R_1 + R_2 &\leq \max \left \{ (n_{11}-n_{12})^+,  n_{21} \right \} + C_{\sf FB1} + n_{12} \\
&+ (n_{22}-n_{21})^+, \nonumber
\end{align}
\noindent $\bullet$ and finally, if $C_{\sf FB1} > \min \left\{ (n_{11} - n_{12})^+, n_{21} \right\}$, \\ $C_{\sf FB2} \leq \min \left\{ (n_{22} - n_{21})^+, n_{12} \right\}$:
\begin{align}
R_1 + R_2 &\leq \max \left \{ (n_{22}-n_{21})^+,  n_{12} \right \} + C_{\sf FB2} + n_{21} \\
&+ (n_{11}-n_{12})^+. \nonumber
\end{align}

Note that the first case is implied by (\ref{eq:R1R2bound-3}); and the second case is implied by (\ref{eq:R1bound-1}) and (\ref{eq:R2bound-1}). Also notice that the third case is implied by (\ref{eq:R2bound-1}) and (\ref{eq:R12R2bound}); and the last case is implied by  (\ref{eq:R1bound-1}) and (\ref{eq:2R1R2bound}). Lastly, we consider (\ref{eq2:2R1R2bound}):
\begin{align*}
2R_1 +R_2 \leq \left\{
  \begin{array}{ll}
(n_{11} - n_{12})^+  +  \max( n_{11}, n_{21}) \\
+  \max \left \{ (n_{22}-n_{21})^+,  n_{12}  \right \} + C_{\sf FB2}, \\
 \qquad  \hbox{ if $C_{\sf FB2} \leq \min \left\{ (n_{22} - n_{21})^+, n_{12} \right\}$;} \\
    \max ( n_{11}, n_{12})  +  \max( n_{11}, n_{21}) \\
    +  (n_{22}-n_{21})^+, \\
 \qquad  \hbox{ if $C_{\sf FB2} > \min \left\{ (n_{22} - n_{21})^+, n_{12} \right\}$.}
  \end{array}
\right.
\end{align*}

Note that the first case is implied by (\ref{eq:2R1R2bound}); and the second case is implied by (\ref{eq:R1bound-1}) and (\ref{eq:R1R2bound-2}).
Similarly, it can be shown that (\ref{eq2:R12R2bound}) is implied by (\ref{eq:R12R2bound}), (\ref{eq:R2bound-1}) and (\ref{eq:R1R2bound-1}). Therefore, the inequalities of (\ref{eq2:R1bound-1})-(\ref{eq2:R12R2bound}) are equivalent to those of (\ref{eq:R1bound-1})-(\ref{eq:R12R2bound}), thus proving the achievablity of Theorem~\ref{theorem:DICcapacity}. 

\section{Proof of Theorem~\ref{theorem:Gaussianouterbound}}
\label{Appendix:Gaussian}
\textbf{Proof of~(\ref{eq:outR1_1}) and~(\ref{eq:outR1_2}):} Starting with Fano's inequality, we get
\begin{align*}
\begin{split}
N( R_1 - \epsilon_N) & \leq I(W_1;Y_1^{N}) \\
& \leq \sum  [h(Y_{1i}) - h(Y_{1i}|W_1,Y_{1}^{i-1},X_{1i})] \\
& = \sum  [h(Y_{1i}) - h(Z_{1i})],
\end{split}
\end{align*}
where the second inequality follows from the fact that conditioning reduces entropy and $X_{1i}$ is a function of $(W_1, Y_1^{i-1})$; and the third equality follows from the memoryless property of the channel. If $(R_1,R_2)$ is achievable, then $\epsilon_N \rightarrow 0$ as $N \rightarrow \infty$. Assume that $X_1$ and $X_2$ have covariance $\rho$, \emph{i.e.}, $\rho = \mathbb{E}[X_1X_2^*]$. We can then obtain~(\ref{eq:outR1_1}).

To obtain~(\ref{eq:outR1_2}), consider
\begin{align*}
\begin{split}
N&(R_1 - \epsilon_N) \leq I(W_1;Y_1^N, Y_2^N|W_2)  \\
& = \sum h(Y_{1i}, Y_{2i}|W_2,Y_1^{i-1},Y_2^{i-1}) - h(Y_1^N, Y_2^N |W_1, W_2) \\
&\overset{(a)}{=} \sum h(Y_{1i}, Y_{2i}|W_2,Y_1^{i-1}, Y_2^{i-1},X_{2i}) \\
&- \sum [h(Z_{1i}) + h(Z_{2i})]\\
&\overset{(b)}{=} \sum h( Y_{2i}|W_2,Y_1^{i-1},Y_2^{i-1}, X_{2i}) \\
&+ \sum h(Y_{1i}|W_2, Y_2^{i}, X_{2i},S_{1i}) - \sum [h(Z_{1i}) + h(Z_{2i}) ]  \\
&\overset{(c)}{\leq} \sum \left[ h(Y_{2i}|X_{2i} ) - h(Z_{2i}) \right] \\
&+ \sum \left[ h(Y_{1i}|X_{2i}, S_{1i}) - h(Z_{1i}) \right],
\end{split}
\end{align*}
where ($a$) follows from the fact that $X_{2}^{i}$ is a function of $(W_2,Y_2^{i-1})$ and $h(Y_1^{N},Y_2^{N}|W_1,W_2)= \sum [h(Z_{1i}) + h(Z_{2i})]$ (see Claim~\ref{claim-GaussianNoise} below); $(b)$ follows from the fact that $S_{1}^{i}:=h_{12} X_1^i + Z_2^i$ is a function of $(Y_{2}^{i},X_{2}^{i})$; ($c$) follows from the fact that conditioning reduces entropy. Hence, we get
\begin{align*}
\begin{split}
R_1 &\leq  h(Y_{2}|X_{2}) - h (Z_{2}) + h(Y_{1}|X_{2},S_{1}) - h (Z_{1}) \\
&\overset{(a)}{\leq} \log \left( 1 +   (1- |\rho|^2) \mathsf{INR}_{12} \right) \\
&+   \log \left( 1 +  \frac{ (1- |\rho|^2) \mathsf{SNR}_1}{ 1 + (1- |\rho|^2) \mathsf{INR}_{12}} \right),
\end{split}
\end{align*}
where $(a)$ follows from the fact that
\begin{align}
h(Y_{2}|X_{2}) &\leq \log 2 \pi e \left( 1 +   (1- |\rho|^2) \mathsf{INR}_{12} \right), \\
\label{eq:h_Y1_X2S1}
h(Y_1|X_2,S_1) & \leq \log 2 \pi e \left( 1+  \frac{(1-|\rho|^2)\mathsf{SNR}_1}{1 + (1-|\rho|^2) \mathsf{INR}_{12} } \right).
\end{align}
The inequality of~(\ref{eq:h_Y1_X2S1}) is obtained as follows. Given $(X_2,S_1)$, the variance of $Y_1$ is upper-bounded by
\begin{align*}
\begin{split}
\textrm{Var} \left[ Y_1| X_2, S_1 \right] &\leq K_{Y_1} - K_{Y_1 (X_2, S_1)}K_{(X_2,S_1)}^{-1}K_{Y_1 (X_2, S_1)}^{*},
\end{split}
\end{align*}
where
\begin{align}
\begin{split}
K_{Y_1} &= E\left[ |Y_1|^2 \right] \\
& = 1  + \mathsf{SNR}_1 + \mathsf{INR}_{21} +  \rho h_{11}^* h_{21} + \rho^* h_{11}h_{21}^*, \\
K_{Y_1 (X_2,S_1)} &= E \left[ Y_1 [X_2^*, S_1^*]  \right] \\
& = \left[\rho h_{11} + h_{21}, h_{12}^* h_{11} + \rho^* h_{21} h_{12}^* \right], \\
K_{(X_2,S_1)} & = E \left[ \left[
                             \begin{array}{cc}
                               |X_2|^2 & X_2 S_1^* \\
                               X_2^* S_1 & |S_1|^2 \\
                             \end{array}
                           \right]
 \right] \\
 & = \left[
             \begin{array}{cc}
               1 & \rho^* h_{12}^* \\
               \rho h_{12} & 1+ \mathsf{INR}_{12} \\
             \end{array}
           \right].
\end{split}
\end{align}
By further calculation, we can get~(\ref{eq:h_Y1_X2S1}).

\begin{claim}
\label{claim-GaussianNoise}
$h(Y_1^N, Y_2^N |W_1, W_2) = h(Y_1^N, S_1^N|W_1, W_2) = \sum [h(Z_{1i}) + h(Z_{2i})]$.
\end{claim}
\begin{proof}
\begin{align*}
\begin{split}
h&(Y_1^{N},Y_2^{N}|W_1,W_2) = \sum h(Y_{1i},Y_{2i}|W_1,W_2,Y_1^{i-1},Y_2^{i-1}) \\
&\overset{(a)}{=} \sum h(Y_{1i},Y_{2i}|W_1,W_2,Y_1^{i-1},Y_2^{i-1},X_{1i},X_{2i}) \\
&\overset{(b)}{=}  \sum h(Z_{1i},Z_{2i}|W_1,W_2,Y_1^{i-1},Y_2^{i-1},X_{1i},X_{2i}) \\
&\overset{(c)}{=}  \sum \left[ h(Z_{1i}) + h(Z_{2i}) \right],
\end{split}
\end{align*}
where ($a$) follows from the fact that $X_{1i}$ is a function of $(W_1, Y_1^{i-1})$ and $X_{2i}$ is a function of $(W_2,Y_{2}^{i-1})$; ($b$) follows from the fact that $Y_{1i} = h_{11} X_{1i} + h_{21} X_{2i} + Z_{1i}$ and $S_{1i}:= h_{12} X_{1i} + Z_{2i}$; ($c$) follows from the memoryless property of the channel and the independence assumption of $Z_{1i}$ and $Z_{2i}$. Similarly, one can show that $ h(Y_1^N, S_1^N|W_1, W_2) = \sum [h(Z_{1i}) + h(Z_{2i})]$.
\end{proof}

\textbf{Proof of~(\ref{eq:outR1_3}):} Starting with Fano's inequality, we get
\begin{align*}
\begin{split}
N&(R_1 - \epsilon_N) \leq I(W_1;Y_1^N, \tilde{Y}_2^N|W_2)  \\
& = \sum h(Y_{1i}, \tilde{Y}_{2i}|W_2,Y_1^{i-1},\tilde{Y}_2^{i-1}) - h(Y_1^N, \tilde{Y}_2^N |W_1, W_2) \\
&\overset{(a)}{\leq} \sum h(Y_{1i}, \tilde{Y}_{2i}|W_2,Y_1^{i-1}, \tilde{Y}_2^{i-1},X_{2i}) - \sum h(Z_{1i}) \\
&{=} \sum H(\tilde{Y}_{2i}|W_2,Y_1^{i-1},\tilde{Y}_2^{i-1}, X_{2i}) \\
&+ \sum h(Y_{1i}|W_2, Y_1^{i-1}, \tilde{Y}_{2}^{i}, X_{2i}) - \sum h(Z_{1i})  \\
&\overset{(b)}{\leq} NC_{\sf FB2} + \sum [h(Y_{1i}|X_{2i}) - h(Z_{1i}) ],
\end{split}
\end{align*}
where ($a$) follows from the fact that $h(Y_1^N, \tilde{Y}_2^N |W_1, W_2) \geq \sum h(Z_{1i})$ (see Claim~\ref{claim-GaussianNoise2} below) and $X_{2i}$ is a function of $(W_2, \tilde{Y}_{2}^{i-1})$; ($c$) follows from the fact that $H(\tilde{Y}_{2i}) \leq C_{\sf FB2}$ and conditioning reduces entropy.
So we get
\begin{align*}
\begin{split}
R_1 &\leq  h(Y_{1}|X_{2}) - h (Z_{1}) + C_{\sf FB2} \\
& \leq \log \left( 1 +   (1- |\rho|^2) \mathsf{SNR}_{1} \right) +   C_{\sf FB2}.
\end{split}
\end{align*}

\begin{claim}
\label{claim-GaussianNoise2}
$h(Y_1^N, \tilde{Y}_2^N |W_1, W_2) \geq \sum h(Z_{1i})$.
\end{claim}
\begin{proof}
\begin{align*}
\begin{split}
h&(Y_1^{N},\tilde{Y}_2^{N}|W_1,W_2) = H(\tilde{Y}_2^{N}|W_1,W_2) \\
&+ h(Y_1^{N}|W_1,W_2,\tilde{Y}_2^{N}) \\
&\overset{(a)}{\geq} \sum h(Y_{1i}|W_1,W_2,Y_1^{i-1},\tilde{Y}_2^{i-1},X_{1i},X_{2i}) 
\end{split}
\end{align*}
\begin{align*}
\begin{split}
&=  \sum h(Z_{1i}|W_1,W_2,Y_1^{i-1},\tilde{Y}_2^{i-1},X_{1i},X_{2i}) \\
&\overset{(b)}{=}  \sum h(Z_{1i}),
\end{split}
\end{align*}
where ($a$) follows from the fact that entropy is nonnegative and $X_{1i}$ is a function of $(W_1, Y_1^{i-1})$ and $X_{2i}$ is a function of $(W_2,\tilde{Y}_{2}^{i-1})$; ($b$) follows from the memoryless property of the channel.
\end{proof}

\textbf{Proof of (\ref{eq:outR12_2}):}
\begin{align*}
\begin{split}
&N(R_1 + R_2 - \epsilon_N) \leq I(W_1;Y_1^{N}) + I(W_2;Y_2^{N})  \\
&\overset{(a)}{\leq} I(W_1;Y_1^{N},S_1^{N}, W_2) + I(W_2;Y_2^{N}) \\
&\overset{(b)}= h(Y_1^{N},S_1^{N}|W_2)  - h(Y_1^{N},S_1^{N}|W_1,W_2) + I(W_2;Y_2^{N}) \\
&\overset{(c)}{=} h(Y_1^{N},S_1^{N}|W_2) - \sum \left[ h(Z_{1i}) + h(Z_{2i}) \right]  + I(W_2;Y_2^{N}) \\
&\overset{(d)}{=} h(Y_1^{N}|S_1^{N},W_2) - \sum  h(Z_{1i}) + h(Y_2^{N})  - \sum h(Z_{2i})  \\
&\overset{(e)}{=} h(Y_1^{N}|S_1^{N},W_2,X_2^{N}) - \sum h(Z_{1i}) + h(Y_2^{N}) - \sum h(Z_{2i})  \\
&\overset{(f)}{\leq} \sum_{i=1}^{N} \left[ h(Y_{1i}|S_{1i},X_{2i}) - h(Z_{1i}) + h(Y_{2i}) - h(Z_{2i}) \right],
\end{split}
\end{align*}
where ($a$) follows from the fact that adding information increases mutual information (providing a \emph{genie}); ($b$) follows from the independence of $W_1$ and $W_2$; ($c$) follows from $h(Y_1^{N},S_1^{N}|W_1,W_2) =\sum \left[ h(Z_{1i}) + h(Z_{2i}) \right]$ (see Claim \ref{claim-GaussianNoise}); $(d)$ follows from  $h(S_1^N|W_2)=h(Y_2^N|W_2)$ (see Claim \ref{claim-4}); ($e$) follows from the fact that $X_2^{N}$ is a function of $(W_2,S_1^{N-1})$ (see Claim \ref{claim-3_Gaussian}); ($f$) follows from the fact that conditioning reduces entropy. Hence, we get
\begin{align*}
R_1 + R_2 &\leq  h(Y_{1}|S_{1},X_{2}) - h(Z_1) + h(Y_{2}) - h(Z_2).
\end{align*}
Note that
\begin{align}
\begin{split}
\label{eq:entropyofY2}
h(Y_2) 
\leq \log 2 \pi e \left( 1 + \mathsf{SNR}_2 + \mathsf{INR}_{12} + 2 |\rho| \sqrt{\mathsf{SNR}_2 \cdot \mathsf{INR}_{12} } \right).
\end{split}
\end{align}
From (\ref{eq:h_Y1_X2S1}) and (\ref{eq:entropyofY2}),  we get the desired upper bound.

\textbf{Proof of (\ref{eq:out2R1R2}):}
\begin{align*}
N&(2R_1 + R_2- \epsilon_N) \leq I(W_1;Y_1^{N}) \\
&+ I(W_1;Y_1^{N}|W_2)  + I(W_2;Y_2^{N})  \\
&\overset{(a)}{\leq} [h(Y_1^{N}) - h(Y_1^N|W_1)]  + I(W_1;Y_1^{N}, S_1^{N}|W_2) \\
&+ [h(Y_2^{N}) -h (Y_2^{N}|W_2)] \\
&\overset{(b)}{=} [h(Y_1^{N}) - h(Y_1^N|W_1)]  +h(Y_1^{N}, S_1^{N}|W_2) \\
&- \sum [h(Z_{1i}) + h(Z_{2i})] + [h(Y_2^{N}) -h (Y_2^{N}|W_2)] \\
&\overset{(c)}{=}   [h(Y_1^{N}) - h(S_2^N|W_1)]  +h(Y_2^{N}) + h(Y_1^{N}|W_2, S_1^{N}) \\
&- \sum [h(Z_{1i}) + h(Z_{2i})] \\
&= h(Y_1^{N}) - \sum [h(Z_{1i}) + h(Z_{2i})] + h(Y_1^N|W_2, S_1^N)  \\
&+ I(W_1; S_2^N) - h(S_2^N) + h(Y_2^{N}) 
\end{align*}
\begin{align}
&= h(Y_1^{N}) - \sum [h(Z_{1i}) + h(Z_{2i})] + h(Y_1^N|W_2, S_1^N) \nonumber \\
&+ I(W_1; S_2^N) + I(S_2^N;Y_2^N) + h(Y_2^N|S_2^N) \nonumber \\
&- I(S_2^N;Y_2^N) - h(S_2^N|Y_2^N) \nonumber \\
&\overset{(d)}\leq   {\small \underbrace{h(Y_1^{N}) - h(Z_{1}^N) - h(Z_{2}^N) + h(Y_1^N|W_2, S_1^N) + h(Y_2^{N}|S_2^N)}_{T}} \nonumber \\
&+ I\left( S_2^N, \tilde{Y}_2^N,W_2; W_1 \right) - h\left( S_2^N| Y_2^N, W_2, \tilde{Y}_2^N \right) \nonumber \\
&= T + I\left( \tilde{Y}_2^N ; W_1 | W_2 \right) + I\left( S_2^N; W_1 | W_2,\tilde{Y}_2^N  \right) \nonumber \\
&- h\left( S_2^N| Y_2^N, W_2, \tilde{Y}_2^N \right) \nonumber \\
&= T + I\left( \tilde{Y}_2^N ; W_1 | W_2 \right) + I\left( Z_1^N; W_1 | W_2,\tilde{Y}_2^N  \right) \nonumber \\
&- h\left( Z_1^N| S_1^N, W_2, \tilde{Y}_2^N \right) \nonumber \\
&= \underbrace{T - h(Z_1^N)}_{T^\prime} + I\left( \tilde{Y}_2^N ; W_1 | W_2 \right) + h(Z_1^N) \nonumber \\
&- h\left( Z_1^N| W_1, W_2, \tilde{Y}_2^N \right) + I(Z_1^N;S_1^N|W_2,\tilde{Y}_2^N) \nonumber \\
&\overset{(e)}= T^\prime + I(Z_1^N;S_1^N|W_2,\tilde{Y}_2^N) + I\left( \tilde{Y}_2^N ; W_1,Z_1^N | W_2 \right) \nonumber \\
&= T^\prime + I(Z_1^N;S_1^N,\tilde{Y}_2^N|W_2) + I\left( \tilde{Y}_2^N ; W_1| W_2,Z_1^N \right) \nonumber \\
&\overset{(f)}\leq T^\prime + I\left( \tilde{Y}_2^N ; W_1| W_2,Z_1^N \right) \nonumber \\
&+ I(Z_1^N;\tilde{Y}_2^N,W_1,\tilde{Y}_1^N,Z_2^N|W_2) \nonumber \\
&\overset{(g)}= T^\prime + I\left( \tilde{Y}_2^N ; W_1| W_2,Z_1^N \right) \nonumber \\
&+ I(Z_1^N;\tilde{Y}_1^N|W_1,W_2,Z_2^N) \nonumber \\
&\overset{(h)}= h(Y_1^N) - h(Z_1^N) + h(Y_2^N|S_2^N) - h(Z_2^N) \nonumber \\
&+ h(Y_1^N|W_2, S_1^N,X_2^N) - h(Z_1^N) \nonumber \\
&+ I\left( \tilde{Y}_2^N ; W_1| W_2,Z_1^N \right) \nonumber \\
&+ I(\tilde{Y}_1^N;Z_1^N|W_1,W_2,Z_2^N) \nonumber \\
&\overset{(i)}{\leq} N C_{\sf FB1} + N C_{\sf FB2} + \sum [h(Y_{1i}) - h(Z_{1i})] \nonumber \\
&+ \sum [h(Y_{1i} |S_{1i}, X_{2i} ) - h(Z_{1i}) ] \nonumber \\ 
&+ \sum [h(Y_{2i} | S_{2i} ) - h(Z_{2i}) ],
\end{align}
where ($a$) follows from the non-negativity of mutual information; ($b$) follows from
$h(Y_1^N, S_1^N|W_1,W_2) = \sum [h(Z_{1i}) + h(Z_{2i})]$ (by Claim~\ref{claim-GaussianNoise}); ($c$) follows from Claim~\ref{claim-4};
($d$) follows from the non-negativity of mutual information and the fact that conditioning reduces entropy; ($e$) is true since $Z_k^N$, $W_1$, and $W_2$ are mutually independent; ($f$) follows from the fact that $S_1^N$ is a function of $(W_1,\tilde{Y}_1^{N},Z_2^N)$; ($g$) can be obtained by taking similar steps as in $(\ref{eq:zeromutual})$; ($h$) follows from Claim~\ref{claim-3_Gaussian}; ($i$) follows from the fact that $H(\tilde{Y}_k^N) \leq NC_{\sf FB k}$ and conditioning reduces entropy.

Also note that
\begin{align}
h(Y_1|X_2,S_1) \leq \log 2 \pi e \left( 1+  \frac{(1-|\rho|^2)\mathsf{SNR}_1}{1 + (1-|\rho|^2) \mathsf{INR}_{12} } \right).
\end{align}

Therefore, we get the desired bound. 

\section{Gap Analysis of Theorem~\ref{THM:Gaussian-sumrate}}
\label{Appendix:Gap}
We show that our proposed achievability strategy in Section~\ref{achievability} results in a sum-rate to within at most $14.8$ bits/sec/Hz of the outerbounds in Corollary~\ref{COL:Gaussian-sumrate}. It is sufficient to prove this for the extreme case of feedback capacity assignment, \emph{i.e.}, where $C_{\sf FB1} = C_{\sf FB}$ and $C_{\sf FB2} = 0$ (or symmetrically $C_{\sf FB1} = 0$ and $C_{\sf FB2} = C_{\sf FB}$). The reason is as follows. Consider our achievability strategy for the general feedback strategy described previously, and let $C_{\sf FB1} = \lambda C_{\sf FB}$ and $C_{\sf FB2} =( 1 - \lambda ) C_{\sf FB}$, such that $0 \leq \lambda \leq 1$. Under these assumptions, for any value of $\lambda$, the outerbounds on sum-rate in (\ref{eq:sum1}),(\ref{eq:sum2}), and (\ref{eq:sum4}) would be the same, call the minimum of them $C^\ast$. Assuming that we can achieve to within $14.8$ bits/sec/Hz of this outer-bound in the extreme cases, then, with the described achievability scheme for general feedback assignment, we can achieve 
\begin{align}
& R_{\sf SUM} = \lambda R_{\sf SUM}^{C_{\sf FB2} = 0} + \left( 1 - \lambda \right) R_{\sf SUM}^{C_{\sf FB1} = 0} \\
& \quad \geq \lambda \left( C^\ast - 14.8 \right) + \left( 1 - \lambda \right) \left( C^\ast - 14.8 \right) = \left( C^\ast - 14.8 \right). \nonumber
\end{align}

We now prove our claim for the extreme cases. By symmetry, we only need to analyze the gap in one case, say $C_{\sf FB1} = C_{\sf FB}$ and $C_{\sf FB2} = 0$. We assume that ${\sf INR} \geq 1$, since for the case when ${\sf INR} < 1$, by ignoring the feedback and treating interference as noise, we can achieve a sum-rate of
\begin{align}
\label{}
2 \log \left( 1 + \frac{\sf SNR}{1 + {\sf INR}} \right),
\end{align}
which is at most within $2.6$ bits of outerbound (\ref{eq:sum2}) in Corollary~\ref{COL:Gaussian-sumrate}:
\begin{align}
\label{}
&\; \log \left( 1 +  \frac{  \mathsf{SNR}}{ 1 +  \mathsf{INR}} \right) +  \log \left( 1+ \mathsf{SNR} + \mathsf{INR} + 2 \sqrt{ \mathsf{SNR}  \cdot \mathsf{INR}} \right) \nonumber \\
&\; - 2 \log \left( 1 + \frac{\sf SNR}{1 + {\sf INR}} \right) \nonumber \\
&\; \overset{({\sf INR} \leq 1)}{\leq} \log \left( 1+ \mathsf{SNR} + \mathsf{INR} + 2 \sqrt{ \mathsf{SNR}  \cdot \mathsf{INR}} \right) \nonumber \\
&\; - \log \left( 1 + \frac{\sf SNR}{2} \right) \nonumber \\
&\; \overset{({\sf INR} \leq 1)}{\leq} \log \left( 2 + 3 \mathsf{SNR} \right) - \log \left( 1 + {\sf SNR} \right) + 1 \nonumber \\
&\; = \log \left( \frac{2 + 3 \mathsf{SNR}}{1 + {\sf SNR}} \right) + 1 \nonumber \\
&\; \leq \log \left( 3 \right) + 1 \leq 2.6.
\end{align}

We consider five different subcases.

\noindent {\bf Case (a) $\log \left( {\sf INR} \right) \leq \frac{1}{2} \log \left( {\sf SNR} \right)$:} 

For this case, we pick the following set of power levels\footnotemark \footnotetext{Remember that starting the beginning of Section~\ref{achievability}, we have assumed that ${\sf INR} \geq 1$. Hence, we are not encountering division by zero in power assignments of (\ref{pwr1}).}:
\begin{equation}
\label{pwr1}
\left\{ \begin{array}{ll}
& P_1^{(1)} = \left( \frac{1}{\sf INR} - \frac{1}{\sf SNR} \min \{ 2^{C_{\sf FB}}, {\sf INR} - 1 \} \right)^+ \\
& P_1^{(2)} = \frac{1}{\sf SNR} \min \{ 2^{C_{\sf FB}}, {\sf INR} - 1 \} \\ 
& P_1^{(3)} = \frac{1}{\sf INR} \min \{ 2^{C_{\sf FB}}, {\sf INR} - 1 \} \\ 
& P_2^{(1)} = \frac{1}{\sf INR} \\
& P_2^{(2)} = \frac{1}{2\sf INR} \min \{ 2^{C_{\sf FB}}, {\sf INR} - 1 \} 
\end{array} \right.
\end{equation}

Note that the power levels are non-negative, and at transmitter $1$, we have
\begin{align}
P_1 &= P_1^{(1)} + P_1^{(2)} + P_1^{(3)} \nonumber \\
&= \left( \frac{1}{\sf INR} - \frac{1}{\sf SNR} \min \{ 2^{C_{\sf FB}}, {\sf INR} - 1 \} \right)^+ \nonumber \\
&+ \frac{1}{\sf SNR} \min \{ 2^{C_{\sf FB}}, {\sf INR} - 1 \} + \frac{1}{\sf INR} \min \{ 2^{C_{\sf FB}}, {\sf INR} - 1 \} \nonumber \\
&\overset{(a)}{=} \frac{1}{\sf INR} - \frac{1}{\sf SNR} \min \{ 2^{C_{\sf FB}}, {\sf INR} - 1 \} \nonumber \\
&+ \frac{1}{\sf SNR} \min \{ 2^{C_{\sf FB}}, {\sf INR} - 1 \} + \frac{1}{\sf INR} \min \{ 2^{C_{\sf FB}}, {\sf INR} - 1 \} \nonumber \\
&= \frac{1}{\sf INR} + \frac{1}{\sf INR} \min \{ 2^{C_{\sf FB}}, {\sf INR} - 1 \} \nonumber \\
&\leq \frac{1}{\sf INR} + \frac{{\sf INR} - 1}{\sf INR} \leq 1,
\end{align}
where $(a)$ follows from the fact that
\begin{align}
\frac{1}{\sf SNR} \min \{ 2^{C_{\sf FB}}, {\sf INR} - 1 \} \leq \frac{{\sf INR} - 1}{\sf SNR} \overset{\left( {\sf INR}^2 \leq {\sf SNR} \right)}{\leq} \frac{1}{\sf INR}.
\end{align}
At transmitter $2$,
\begin{align}
P_2 &= P_2^{(1)} + P_2^{(2)} \nonumber \\
&= \frac{1}{\sf INR} + \frac{1}{2{\sf INR}} \min \{ 2^{C_{\sf FB}}, {\sf INR} - 1 \} \nonumber \\
&\leq \frac{1}{\sf INR} + \frac{{\sf INR} - 1}{\sf INR} \leq 1.
\end{align}

By plugging the given values of power levels into our achievable sum-rate $R_{\sf SUM}^{(a)}$ defined in (\ref{eq:rsuma}), we get
\begin{align}
R_{\sf SUM}^{(a)} &= \log \left( 1 + \frac{ {\sf SNR} P_1^{(1)}}{1 + {\sf INR} P_2^{(1:2)} + {\sf SNR} P_1^{(2)} } \right) \nonumber \\
& + \log \left( \frac{{\sf SNR} P_1^{(2)}}{2} \right)^+ + \log \left( 1 + \frac{ {\sf SNR} P_2^{(1)}}{2} \right) \nonumber \\
& + \log \left( \frac{ {\sf INR} P_2^{(2)}}{1 + {\sf INR} P_2^{(1)} } \right)^+ \nonumber \\
& = \log \left( \frac{ 2 + {\sf INR} P_2^{(2)} + {\sf SNR} P_1^{(1:2)}}{2 + {\sf INR} P_2^{(2)} + {\sf SNR} P_1^{(2)} } \right) \nonumber \\
& + \log \left( \frac{ {\sf SNR} P_1^{(2)}}{2} \right)^+ + \log \left( 1 + \frac{\sf SNR}{2 {\sf INR}} \right) \nonumber \\
& + \log \left( \frac{ {\sf INR} P_2^{(2)}}{1 + {\sf INR} P_2^{(1)} } \right)^+ \nonumber \\
& \overset{(a)}{\geq} \log \left( \frac{ 2 + {\sf INR} P_2^{(2)} + {\sf SNR} P_1^{(1:2)}}{2 \left( 1 + {\sf SNR} P_1^{(2)} \right)} \times \frac{1 + {\sf SNR} P_1^{(2)}}{4} \right) \nonumber \\
& + \log \left( 1 + \frac{\sf SNR}{2 {\sf INR}} \right) + \log \left( \frac{ {\sf INR} P_2^{(2)}}{1 + {\sf INR} P_2^{(1)} } \right)^+ \nonumber \\
& = \log \left( 2 + {\sf INR} P_2^{(2)} + {\sf SNR} P_1^{(1:2)} \right) + \log \left( 1 + \frac{\sf SNR}{2 {\sf INR}} \right)  \nonumber \\
& + \log \left( \frac{\min \{ 2^{C_{\sf FB}}, {\sf INR} - 1 \} }{4} \right)^+ - 3 \nonumber 
\end{align}
\begin{align}
\label{acha}
& \geq \log \left( 2 + {\sf INR} P_2^{(2)} + {\sf SNR} P_1^{(1:2)} \right) + \log \left( 1 + \frac{\sf SNR}{\sf INR} \right) \nonumber \\
& + \min \{ C_{\sf FB}, \log \left( {\sf INR} - 1 \right)^+ \} - 6 \nonumber \\
& \geq \log \left( 2 + {\sf INR} P_2^{(2)} + {\sf SNR} P_1^{(1:2)} \right) + \log \left( 1 + \frac{\sf SNR}{\sf INR} \right) \nonumber \\
& + \min \{ C_{\sf FB}, \log \left( 1 + {\sf INR} \right) \} - \log \left( 3 \right) - 6, 
\end{align}
where $(a)$ follows from the assumption ${\sf SNR} \geq {\sf INR} \geq 1$, and the last inequality holds since
\begin{equation}
\log \left( {\sf INR} - 1 \right)^+ \geq \log \left( 1 + {\sf INR} \right) - 3 \hspace{2mm} \forall~{\sf INR} \geq 1.
\end{equation} 

\noindent $\bullet$ If $C_{\sf FB} \leq \log \left( 1 + {\sf INR} \right)$: Considering the outerbound in (\ref{eq:sum4}), in this case we can write
\begin{align}
\label{outa}
&R_1 + R_2 \leq 2 \log \left( 1 + {\sf INR} + \frac{ {\sf SNR} }{1+ {\sf INR}} \right) + C_{\sf FB} \nonumber \\
&\; \leq 2 \log \left( 1 + {\sf INR} + \frac{ {\sf SNR} }{\sf INR} \right) + C_{\sf FB} \nonumber \\
&\; = 2 \log \left( 1 + \frac{ {\sf INR}^2 +  {\sf SNR} }{\sf INR} \right) + C_{\sf FB} \nonumber \\
&\; \overset{\left( {\sf INR}^2 \leq {\sf SNR} \right)}{\leq} 2 \log \left( 1 + \frac{ 2 {\sf SNR} }{\sf INR} \right) + C_{\sf FB} \nonumber \\
&\; \leq 2 \log \left( 1 + \frac{\sf SNR}{\sf INR} \right) + C_{\sf FB} + 2.
\end{align}

The gap between the achievable sum-rate of (\ref{acha}) and the outerbound in (\ref{outa}), is upper bounded by
\begin{align*}
\begin{split}
&\; 8 + \log \left( 3 \right) + 2 \log \left( 1 + \frac{\sf SNR}{\sf INR} \right) + C_{\sf FB} \\
&\; - \log \left( 2 + {\sf INR} P_2^{(2)} + {\sf SNR} P_1^{(1:2)} \right) \nonumber \\
&\; - \log \left( 1 + \frac{\sf SNR}{\sf INR} \right) - C_{\sf FB} \nonumber \\
&\; \leq 8 + \log \left( 3 \right) + \log \left( 1 + \frac{\sf SNR}{\sf INR} \right) - \log \left( 2 + \frac{\sf SNR}{\sf INR} \right) \nonumber \\
&\; \leq 8 + \log \left( 3 \right).
\end{split}
\end{align*}

\noindent $\bullet$ If $C_{\sf FB} > \log \left( 1 + {\sf INR} \right)$: Considering the outerbound in (\ref{eq:sum2}), in this case we can write
\begin{align}
\label{outa1}
&R_1 + R_2 \leq \log \left( 1 +  \frac{  \mathsf{SNR}}{ 1 +  \mathsf{INR}} \right) \nonumber \\
&\; + \log \left( 1+ \mathsf{SNR} + \mathsf{INR} + 2 \sqrt{ \mathsf{SNR}  \cdot \mathsf{INR}} \right) \nonumber \\
&\; \leq \log \left( 1 + \frac{\sf SNR}{\sf INR} \right) \nonumber \\
&\; + \log \left( 1+ \mathsf{SNR} + \mathsf{INR} + 2 \sqrt{ \mathsf{SNR}  \cdot \mathsf{INR}} \right).
\end{align}

The gap between the achievable sum-rate of (\ref{acha}) and the outerbound in (\ref{outa1}), is upper bounded by
\begin{align*}
\begin{split}
&\; 6 + \log \left( 3 \right) + \log \left( 1 + \frac{\sf SNR}{\sf INR} \right) \nonumber \\
&\; + \log \left( 1+ \mathsf{SNR} + \mathsf{INR} + 2 \sqrt{ \mathsf{SNR}  \cdot \mathsf{INR}} \right) \nonumber \\
&\; - \log \left( 1 + \frac{\sf SNR}{\sf INR} \right) - \log \left( 2 + \frac{\sf SNR}{\sf INR} \right) \nonumber \\
&\; - \log \left( 1 + {\sf INR} \right) \nonumber \\
&\; \leq 6 + \log \left( 3 \right) - \log \left( 2 + \frac{\sf SNR}{\sf INR} \right) \nonumber \\
&\; + \log \left( 1 + \frac{{\sf SNR} + 2 \sqrt{ \mathsf{SNR}  \cdot \mathsf{INR}}}{\sf INR} \right) \nonumber \\
&\; \leq 6 + \log \left( 3 \right) - \log \left( 2 + \frac{\sf SNR}{\sf INR} \right) \nonumber \\
&\; + \log \left( 1 + \frac{3{\sf SNR}}{\sf INR} \right) \nonumber \\
&\; \leq 6 + 2 \log \left( 3 \right).
\end{split}
\end{align*}

Hence, we conclude that the gap between the inner-bound and the outer-bound is at most $8 + \log \left( 3 \right) \leq 9.6$ bits/sec/Hz.

\noindent {\bf Case (b) $\frac{1}{2} \log \left( {\sf SNR} \right) \leq \log \left( {\sf INR} \right) \leq \frac{2}{3} \log \left( {\sf SNR} \right)$:}

For this case, we pick the following set of power levels\footnotemark \footnotetext{${\sf INR} \geq 1$.}:
\begin{equation}
\label{}
\left\{ \begin{array}{ll}
& P_1^{(1)} = \left( \frac{1}{\sf INR} - \frac{1}{\sf SNR} \min \{ 2^{C_{\sf FB}}, \frac{{\sf SNR}^2}{{\sf INR}^3} - 1 \} \right)^+ \\
& P_1^{(2)} = \frac{1}{\sf SNR} \min \{ 2^{C_{\sf FB}}, \frac{{\sf SNR}^2}{{\sf INR}^3} - 1 \} \\ 
& P_1^{(3)} = \frac{1}{\sf INR} \min \{ 2^{C_{\sf FB}}, \frac{{\sf SNR}^2}{{\sf INR}^3} - 1 \} \\
& P_1^{(4)} = \left( 1 - P_1^{(1:3)} \right)^+ \\ 
& P_2^{(1)} = \frac{1}{\sf INR} \\
& P_2^{(2)} = \frac{1}{2{\sf INR}} \min \{ 2^{C_{\sf FB}}, \frac{{\sf SNR}^2}{{\sf INR}^3} - 1 \} \\
& P_2^{(3)} = \left( 1 - P_2^{(1:2)} \right)^+
\end{array} \right.
\end{equation}

All the power levels are non-negative. Also, we have
\begin{align}
& P_1^{(1)} + P_1^{(2)} + P_1^{(3)} \nonumber \\
&= \left( \frac{1}{\sf INR} - \frac{1}{\sf SNR} \min \{ 2^{C_{\sf FB}}, \frac{{\sf SNR}^2}{{\sf INR}^3} - 1 \} \right)^+ \nonumber \\
&+ \frac{1}{\sf SNR} \min \{ 2^{C_{\sf FB}}, \frac{{\sf SNR}^2}{{\sf INR}^3} - 1 \} + \frac{1}{\sf INR} \min \{ 2^{C_{\sf FB}}, \frac{{\sf SNR}^2}{{\sf INR}^3} - 1 \} \nonumber \\
&\overset{(a)}{=} \frac{1}{\sf INR} - \frac{1}{\sf SNR} \min \{ 2^{C_{\sf FB}}, \frac{{\sf SNR}^2}{{\sf INR}^3} - 1 \} \nonumber \\
&+ \frac{1}{\sf SNR} \min \{ 2^{C_{\sf FB}}, \frac{{\sf SNR}^2}{{\sf INR}^3} - 1 \} + \frac{1}{\sf INR} \min \{ 2^{C_{\sf FB}}, \frac{{\sf SNR}^2}{{\sf INR}^3} - 1 \} \nonumber \\
&= \frac{1}{\sf INR} + \frac{1}{\sf INR} \min \{ 2^{C_{\sf FB}}, \frac{{\sf SNR}^2}{{\sf INR}^3} - 1 \} \nonumber \\
&\leq \frac{{\sf SNR}^2}{{\sf INR}^4} \overset{(\sqrt{\sf SNR} \leq {\sf INR})}{\leq} \frac{\sqrt{\sf SNR}}{\sf INR} \leq 1,
\end{align}
where $(a)$ follows from the fact that
\begin{align}
\frac{1}{\sf SNR} \min \{ 2^{C_{\sf FB}}, \frac{{\sf SNR}^2}{{\sf INR}^3} - 1 \} \leq \frac{1}{\sqrt {\sf SNR}} \leq \frac{1}{\sf INR}.
\end{align}

Since $P_1^{(4)} = \left( 1 - P_1^{(1:3)} \right)^+$, we conclude that $P_1 \leq 1$. Similarly, we can show that $P_2 \leq 1$.

By plugging the given values of power levels into our achievable sum-rate $R_{\sf SUM}^{(b)}$ defined in (\ref{eq:rsumb}), we have
\begin{align}
\label{achb}
& R_{\sf SUM}^{(b)} = \log \left( \frac{1 + {\sf INR} P_2^{(1)} + {\sf SNR} P_1^{(1:2)}}{1 + {\sf INR} P_2^{(1)}} \right) \nonumber \\
&\; + \log \left( \frac{2 {\sf INR} + {\sf SNR} + {\sf INR}^2 - 2 {\sf INR}}{2 {\sf INR} + {\sf SNR}} \right) \nonumber \\
&\; + \log \left( 1 + \frac{\sf SNR}{2 {\sf INR}} \right) \nonumber \\
&\; + \log \left( \frac{2 {\sf INR} + {\sf SNR} + {\sf INR}^2 - 3/2 {\sf INR}}{2 {\sf INR} + {\sf SNR}} \right) \nonumber \\
&\; + \min \{ C_{\sf FB}, \log \left( \frac{{\sf SNR}^2}{{\sf INR}^3} - 1 \right)^+ \} - 3 \nonumber \\
&\; \geq \log \left( \frac{2 + {\sf SNR} P_1^{(1:2)}}{2} \right) \nonumber \\
&\; + \log \left( 1 + \frac{\sf SNR}{2 {\sf INR}} \right) + 2 \log \left( \frac{{\sf INR}^2}{3 {\sf SNR}} \right) \nonumber \\
&\; + \min \{ C_{\sf FB}, \log \left( \frac{{\sf SNR}^2}{{\sf INR}^3} - 1 \right)^+ \} - 3 \nonumber \\
&\; \geq 2 \log \left( 1 + \frac{\sf SNR}{2 {\sf INR}} \right) + 2 \log \left( \frac{ {\sf INR}^2}{3 {\sf SNR}} \right) \nonumber \\
&\; + \min \{ C_{\sf FB}, \log \left( \frac{{\sf SNR}^2}{{\sf INR}^3} - 1 \right)^+ \} - 3 \nonumber \\
&\; \geq 2 \log \left( \frac{1 + {\sf SNR}}{2 {\sf INR}} \right) + 2 \log \left( {\sf INR}^2 \right)  - 2 \log \left( 3 {\sf SNR} \right) \nonumber \\
&\; + \min \{ C_{\sf FB}, \log \left( \frac{{\sf SNR}^2}{{\sf INR}^3} - 1 \right)^+ \} - 3 \nonumber \\
&\; = 2 \log \left( 1 + {\sf SNR} \right) - 2 \log \left( {\sf INR} \right) \nonumber\\
&\; + 4 \log \left( {\sf INR} \right) - 2 \log \left( 3 {\sf SNR} \right) \nonumber \\
&\; + \min \{ C_{\sf FB}, \log \left( \frac{{\sf SNR}^2}{{\sf INR}^3} - 1 \right)^+ \} - 3 \nonumber \\
&\; \geq  2 \log \left( 1 + {\sf INR} \right) \nonumber\\
&\; + \min \{ C_{\sf FB}, \log \left( \frac{{\sf SNR}^2}{{\sf INR}^3} - 1 \right)^+ \} - 5 - 2 \log \left( 3 \right). 
\end{align}

We first simplify the outerbounds in (\ref{eq:sum2}) and (\ref{eq:sum4}). Considering the outerbound in (\ref{eq:sum4}), for this case we can write
\begin{align}
\label{outb1}
& R_1 + R_2 \leq 2 \log \left( 1 + {\sf INR} + \frac{ {\sf SNR} }{1+ {\sf INR}} \right) + C_{\sf FB} \nonumber \\
&\; \leq 2 \log \left( 1 + {\sf INR} + \frac{ {\sf SNR} }{{\sf INR}} \right) + C_{\sf FB} \nonumber \\
&\; \overset{({\sf SNR} \leq {\sf INR}^2)}{\leq} 2 \log \left( 1 + {\sf INR} + {\sf INR} \right) + C_{\sf FB} \nonumber \\
&\; = 2 \log \left( 1 + 2 {\sf INR} \right) + C_{\sf FB} \nonumber \\
&\; \leq 2 \log \left( 1 + {\sf INR} \right) + C_{\sf FB} + 2,
\end{align}
and for the outerbound in (\ref{eq:sum2}), we have
\begin{align}
\label{outb2}
& R_1 + R_2 \leq \log \left( 1 +  \frac{  \mathsf{SNR}}{ 1 +  \mathsf{INR}} \right) \nonumber \\
&+  \log \left( 1+ \mathsf{SNR} + \mathsf{INR} + 2 \sqrt{ \mathsf{SNR}  \cdot \mathsf{INR}} \right) \nonumber \\
&\overset{({\sf INR} \leq {\sf SNR})}{\leq} \log \left( 1 + {\sf INR} + {\sf SNR} \right) - \log \left( 1 + {\sf INR} \right) \nonumber \\
&+  \log \left( 1+ 4 {\sf SNR} \right) \nonumber \\
&\overset{({\sf INR} \leq {\sf SNR})}{\leq} 2 \log \left( 1 + {\sf SNR} \right) - \log \left( 1 + {\sf INR} \right) + 3.
\end{align}

We consider two possible scenarios based on $C_{\sf FB}$,

\noindent (1) $C_{\sf FB} \leq \log \left( \frac{{\sf SNR}^2}{{\sf INR}^3} - 1 \right)^+$:

With this assumption, we have
\begin{align}
\label{achb1}
R_{\sf SUM}^{(b)}  \geq 2 \log \left( 1 + {\sf INR} \right) + C_{\sf FB} - 5 - 2 \log \left( 3 \right),
\end{align}
note that, in this case
\begin{align}
C_{\sf FB} &\leq \log \left( \frac{{\sf SNR}^2}{{\sf INR}^3} - 1 \right)^+ \leq \log \left( 1 + \frac{{\sf SNR}^2}{{\sf INR}^3} \right) \nonumber \\
&\overset{(\sqrt{\sf SNR} \leq {\sf INR})}{\leq} \log \left( 1 + {\sf INR} \right),
\end{align}
therefore, from (\ref{outb1}) we get
\begin{align}
\label{outb11}
R_1 + R_2 & \leq 2 \log \left( 1 + {\sf INR} \right) + C_{\sf FB} + 2.
\end{align}

Hence, the gap between the achievable sum-rate of (\ref{achb1}) and the outer-bound in (\ref{outb11}) is at most $7 + 2 \log \left( 3 \right) \leq 10.2$ bits/sec/Hz.

\noindent (2) $C_{\sf FB} \geq \log \left( \frac{{\sf SNR}^2}{{\sf INR}^3} - 1 \right)^+$:

We have
\begin{align}
\label{achb2}
R_{\sf SUM}^{(b)} &\geq 2 \log \left( 1 + {\sf INR} \right) + \log \left( 1 + \frac{{\sf SNR}^2}{{\sf INR}^3} \right)^+ - 5 - 3 \log \left( 3 \right) \nonumber \\
&\geq 2 \log \left( 1 + {\sf INR} \right) + 2 \log \left( {\sf SNR} \right) \nonumber \\
&- 3 \log \left( {\sf INR} \right) - 5 - 3 \log \left( 3 \right) \\
&\geq 2 \log \left( 1 + {\sf INR} \right) + 2 \log \left( 1 + {\sf SNR} \right) \nonumber \\
&- 3 \log \left( 1 + {\sf INR} \right) - 7 - \log \left( 3 \right) \nonumber \\
&= 2 \log \left( 1 + {\sf SNR} \right) - \log \left( 1 + {\sf INR} \right) - 7 - 3 \log \left( 3 \right). \nonumber
\end{align}

Hence, the gap between the achievable sum-rate of (\ref{achb2}) and the outer-bound in (\ref{outb2}) is at most $10 + 3 \log \left( 3 \right) \leq 14.8$ bits/sec/Hz. As a result,  the gap between the inner-bound and the minimum of the outer-bounds in (\ref{eq:sum2}) and (\ref{eq:sum4}) is at most $14.8$ bits/sec/Hz.

\noindent {\bf Case (c) $2 \log \left( {\sf SNR} \right) \leq \log \left( {\sf INR} \right)$:}

\noindent $\bullet$ If ${\sf SNR \leq 1}$, pick 
$$P_2^{(2)} = P_1^{(3)} = \frac{1}{\sf INR} \min \{ 2^{C_{\sf FB}}, {\sf INR} \},$$ 
and set all other power levels equal to zero. By plugging the given values of power levels into our achievable sum-rate $R_{\sf SUM}^{(c)}$ defined in (\ref{eq:rsumc}), we get
\begin{align}
\label{eq:rsumc1}
& R_{\sf SUM}^{(c)} =  \log \left( \frac{\min \{ 2^{C_{\sf FB}}, {\sf INR} \} }{2} \right) \nonumber \\
&\; = \min \{ C_{\sf FB}, \log \left( {\sf INR} \right) \} - 1 \nonumber \\
&\; \geq \min \{ C_{\sf FB}, \log \left( 1 + {\sf INR} \right) \} - 2. 
\end{align}

Consider the outerbounds in (\ref{eq:sum1}) and (\ref{eq:sum2}), under the assumptions of case (c) and ${\sf SNR \leq 1}$, we have
\begin{align}
\label{}
& R_1 + R_2 \leq \min \{ 2 \log \left( 1 + \mathsf{SNR} \right) + C_{\sf FB}, \nonumber \\
& \log \left( 1 +  \frac{  \mathsf{SNR}}{ 1 +  \mathsf{INR}} \right) \nonumber \\
& + \log \left( 1+ \mathsf{SNR} + \mathsf{INR} + 2 \sqrt{ \mathsf{SNR}  \cdot \mathsf{INR}} \right) \} \nonumber \\
& \leq \min \{ 2 \log \left( 2 \right) + C_{\sf FB}, \log \left( 2 \right) +   \log \left( 2 + 3 \mathsf{INR} \right) \} \nonumber \\
& \leq \min \{ 2 + C_{\sf FB}, 1 +  \log \left( 3 \right) + \log \left( 1 + \mathsf{INR} \right) \} \nonumber \\
& \leq \min \{ C_{\sf FB}, \log \left( 1 + \mathsf{INR} \right) \} + 2.6.
\end{align}

Therefore, with the given choice of power levels the achievable sum-rate of (\ref{eq:rsumc1}) is within $2.6$ bits/sec/Hz of the minimum of the outerbounds in (\ref{eq:sum1}) and (\ref{eq:sum2}). 

\noindent $\bullet$ If ${\sf SNR \geq 1}$, pick 

Pick the following set of power levels:
\begin{equation}
\label{}
\left\{ \begin{array}{ll}
& P_1^{(1)} = 0 \\
& P_1^{(2)} = 0 \\ 
& P_1^{(3)} = \frac{\sf SNR}{\sf INR} \min \{ 2^{C_{\sf FB}}, \frac{\sf INR}{{\sf SNR}^2} \} \\
& P_1^{(4)} = 1 - P_1^{(3)} \\ 
& P_2^{(1)} = 0 \\
& P_2^{(2)} = \frac{1}{\sf INR} \min \{ 2^{C_{\sf FB}}, \frac{\sf INR}{{\sf SNR}^2} \} \\
& P_2^{(3)} = 1 - P_2^{(2)}
\end{array} \right.
\end{equation}

It is straight forward to check that the power levels are non-negative and they satisfy the power constraint at the transmitters. By plugging the given values of power levels into our achievable sum-rate $R_{\sf SUM}^{(c)}$ defined in (\ref{eq:rsumc}), we get
\begin{align}
\label{}
& R_{\sf SUM}^{(c)} = \log \left( 1 + \frac{{\sf SNR} (1 - \frac{1}{\sf SNR})}{2} \right) \nonumber \\
&\; + \log \left( 1 + \frac{{\sf SNR} (1 - \frac{1}{\sf SNR})}{2} \right) + \log \left( \frac{\min \{ 2^{C_{\sf FB}}, \frac{\sf INR}{{\sf SNR}^2} \}}{2} \right) \nonumber \\
&\; \geq 2 \log \left( 1 + {\sf SNR} \right) + \min \{ C_{\sf FB}, \log \left( \frac{\sf INR}{{\sf SNR}^2} \right) \} -3 \nonumber \\
&\; = \min \{ 2 \log \left( 1 + {\sf SNR} \right) + C_{\sf FB}, \nonumber \\
&\; 2 \log \left( 1 + {\sf SNR} \right) + \log \left( \frac{\sf INR}{{\sf SNR}^2} \right) \} - 3 \nonumber \\
&\; \geq \min \{ 2 \log \left( 1 + {\sf SNR} \right) + C_{\sf FB}, \log \left( 1 + {\sf INR} \right) \} - 4.
\end{align}

Consider the outerbounds in (\ref{eq:sum1}) and (\ref{eq:sum2}), under the assumptions of case (c), we have
\begin{align}
\label{outc}
& R_1 + R_2 \leq \min \{ 2 \log \left( 1 + \mathsf{SNR} \right) + C_{\sf FB}, \nonumber \\
& \log \left( 1 +  \frac{  \mathsf{SNR}}{ 1 +  \mathsf{INR}} \right) \nonumber \\
& + \log \left( 1+ \mathsf{SNR} + \mathsf{INR} + 2 \sqrt{ \mathsf{SNR}  \cdot \mathsf{INR}} \right) \} \nonumber \\
& \leq \min \{ 2 \log \left( 1 + \mathsf{SNR} \right) + C_{\sf FB}, \nonumber \\
& \log \left( 2 \right) + \log \left( 1 + 4 \mathsf{INR} \right) \} \nonumber \\
& \leq \min \{ 2 \log \left( 1 + \mathsf{SNR} \right) + C_{\sf FB}, 3 + \log \left( 1 + \mathsf{INR} \right) \} \nonumber \\
& \leq \min \{ 2 \log \left( 1 + \mathsf{SNR} \right) + C_{\sf FB}, \log \left( 1 + \mathsf{INR} \right) \} + 3.
\end{align}

Therefore, with the given choice of power levels the achievable sum-rate of (\ref{eq:rsumc}) is within $7$ bits/sec/Hz of the minimum of the outerbounds in (\ref{eq:sum1}) and (\ref{eq:sum2}). 

\noindent {\bf Case (d) $\frac{2}{3} \log \left( {\sf SNR} \right) \leq \log \left( {\sf INR} \right) \leq \log \left( {\sf SNR} \right)$:}

In this case feedback is not needed. The achievability scheme of \cite{Etkin} for Gaussian IC without feedback, results in a sum-rate to within $1$ bit/sec/Hz of
\begin{equation}
\log \left( 1 + {\sf SNR} \right) + \log \left( 1 + \frac{\sf SNR}{1 + {\sf INR}} \right).
\end{equation}

For the outerbound in (\ref{eq:sum2}), in this case we have
\begin{align}
\label{outc}
R_1 + R_2 &\leq \log \left( 1 +  \frac{  \mathsf{SNR}}{ 1 +  \mathsf{INR}} \right) \nonumber \\
&+  \log \left( 1+ \mathsf{SNR} + \mathsf{INR} + 2 \sqrt{ \mathsf{SNR}  \cdot \mathsf{INR}} \right) \nonumber \\
&\leq \log \left( 1 + \frac{\sf SNR}{1 + {\sf INR}} \right) + \log \left( 1 + 4 {\sf SNR} \right) \nonumber \\
&\leq \log \left( 1 + \frac{\sf SNR}{1 + {\sf INR}} \right) + \log \left( 1 + {\sf SNR} \right) + 2.
\end{align}

Therefore, the achievable sum-rate of \cite{Etkin} is within $3$ bits/sec/Hz of the outerbound in (\ref{eq:sum2}).

\noindent {\bf Case (e) $\log \left( {\sf SNR} \right) \leq \log \left( {\sf INR} \right) \leq 2 \log \left( {\sf SNR} \right)$: }

In this case feedback is not needed. The achievability scheme of \cite{Etkin} for Gaussian IC without feedback, results in a sum-rate to within $1$ bit/sec/Hz of
\begin{equation}
\log \left( 1 + {\sf SNR} + {\sf INR} \right).
\end{equation}

For the outerbound in (\ref{eq:sum2}), in this case we have
\begin{align}
\label{outd}
R_1 + R_2 &\leq \log \left( 1 +  \frac{  \mathsf{SNR}}{ 1 +  \mathsf{INR}} \right) \nonumber \\
&+  \log \left( 1+ \mathsf{SNR} + \mathsf{INR} + 2 \sqrt{ \mathsf{SNR}  \cdot \mathsf{INR}} \right) \nonumber \\
&\overset{({\sf SNR} \leq {\sf INR})}{\leq} \log \left( 2 \right) + \log \left( 1+ 4 \mathsf{INR} \right) \nonumber \\
&\leq \log \left( 1 + {\sf INR} \right) + 3.
\end{align}

Therefore, the achievable sum-rate of \cite{Etkin} is within $4$ bits/sec/Hz of the outerbound in (\ref{eq:sum2}).

\bibliographystyle{ieeetr}
\bibliography{bib_feedback}

\begin{biography}[{\includegraphics[width=1in,height=1.25in,clip,keepaspectratio]{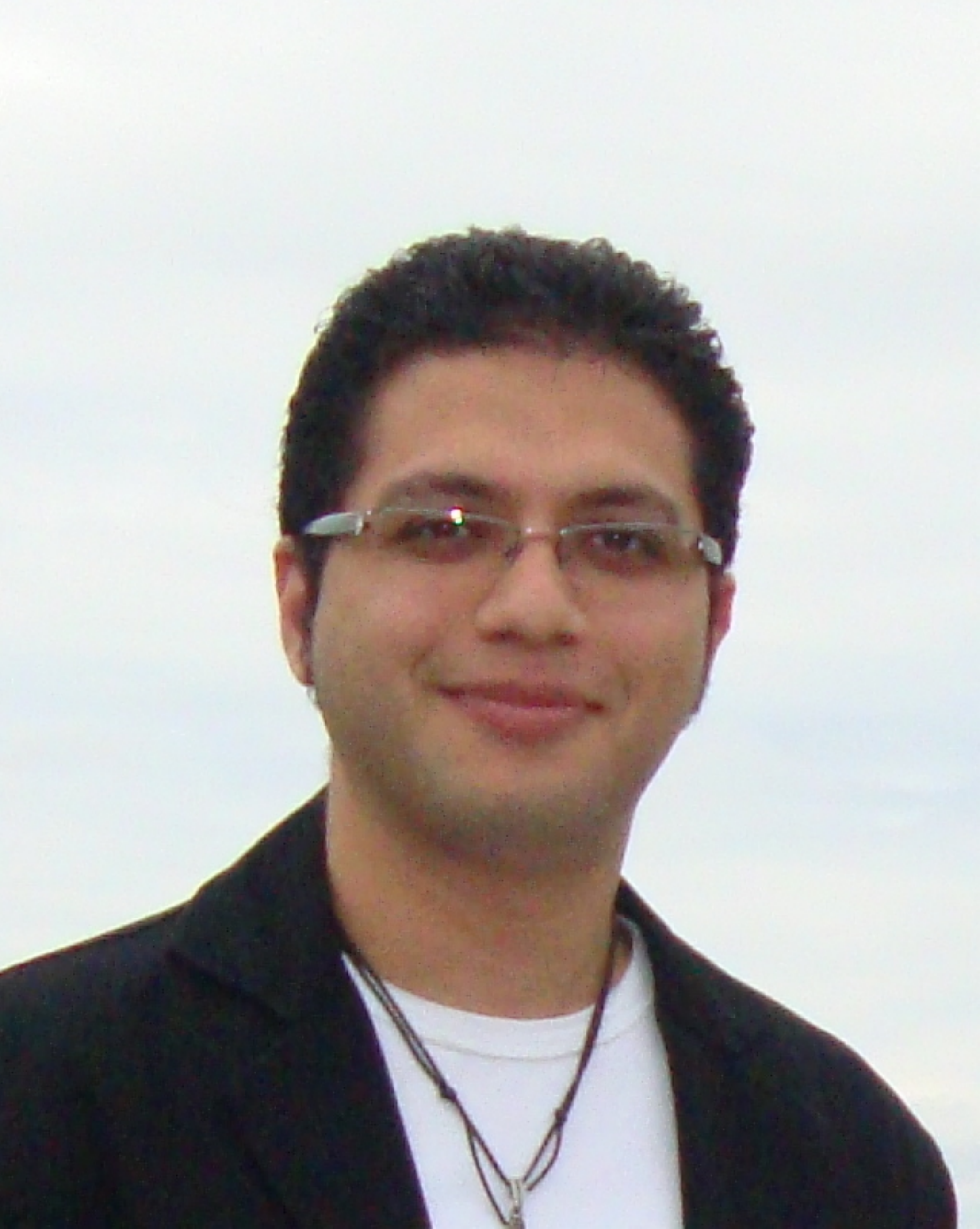}}] {Alireza Vahid} received his B.Sc. degree in electrical engineering from Sharif University of Technology, Tehran, Iran, in 2009. He is currently a Ph.D. student at the School of Electrical and Computer Engineering, Cornell University, Ithaca, NY. His research interests include information theory and wireless communications.

He has received the Director's Ph.D. Teaching Assistant Award in 2010, and Jacobs Scholar Fellowship in 2009.
\end{biography}


\begin{biography}[{\includegraphics[width=1in,height=1.25in,clip,keepaspectratio]{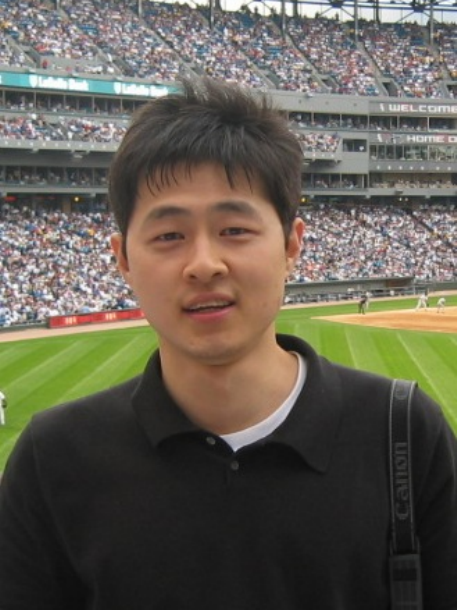}}] {Changho Suh} is a postdoctoral associate at MIT. He received the Ph.D. degree in Electrical Engineering and Computer Sciences from UC-Berkeley in 2011. He received the B.S. and M.S. degrees in Electrical Engineering from Korea Advanced Institute of Science and Technology in 2000 and 2002 respectively. Before joining Berkeley, he had been with the Telecommunication R\&D Center, Samsung Electronics.

Dr. Suh received the David J. Sakrison Memorial Prize for outstanding doctoral research from the UC-Berkeley EECS Department in 2011, the Best Student Paper Award of the IEEE International Symposium on Information Theory in 2009 and the Outstanding Graduate Student Instructor Award in 2010. He was awarded several fellowships, including the Vodafone U.S. Foundation Fellowship in 2006 and 2007; the Kwanjeong Educational Foundation Fellowship in 2009; and the Korea Government Fellowship from 1996 to 2002.
\end{biography}


\begin{biography}[{\includegraphics[width=1in,height=1.25in,clip,keepaspectratio]{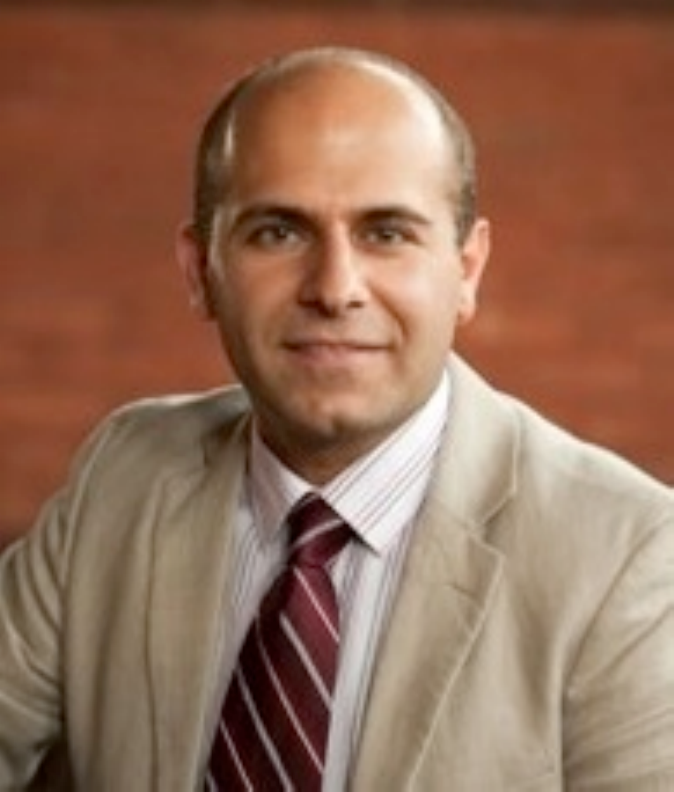}}] {A. Salman Avestimehr} received the B.S. degree in electrical engineering from Sharif University of Technology, Tehran, Iran, in 2003 and the M.S. degree and Ph.D. degree in electrical engineering and computer science, both from the University of California, Berkeley, in 2005 and 2008, respectively.

He is currently an Assistant Professor at the School of Electrical and Computer Engineering at Cornell University, Ithaca, NY. He was also a postdoctoral scholar at the Center for the Mathematics of Informa- tion (CMI) at the California Institute of Technology, Pasadena, in 2008. His research interests include information theory, communications, and networking.

Dr. Avestimehr has received a number of awards, including the Presidential Early Career Award for Scientists and Engineers (PECASE) in 2011, the Young Investigator Program (YIP) award from the U. S. Air Force Office of Scientific Research (2011), the National Science Foundation CAREER award (2010), the David J. Sakrison Memorial Prize from the U.C. Berkeley EECS Department (2008), and the Vodafone U.S. Foundation Fellows Initiative Research Merit Award (2005). He has been a Guest Associate Editor for the IEEE Transactions on Information Theory Special Issue on Interference Networks.
\end{biography}

\end{document}